\documentclass[11pt,english]{article}
\usepackage[T1]{fontenc}
\usepackage{rotating}
\usepackage{booktabs}
\usepackage[usenames,dvipsnames]{xcolor}
\usepackage{dcolumn}
\usepackage{units}
\usepackage{array}
\usepackage{tikz}
\usetikzlibrary{arrows}
\usepackage{graphicx}
\usepackage{xcolor}
\usepackage{enumerate}
\usepackage{subcaption}
\usepackage{caption}
\usepackage{amsthm}
\usepackage{amsfonts,amssymb,amsmath,amsthm,epsfig,euscript,bbm,amsthm}
\usepackage{verbatim}
\usepackage{amsthm}
\usepackage{amsmath}
\usepackage{amssymb}
\usepackage{graphicx}
\usepackage{setspace}
\usepackage[authoryear]{natbib}
\usepackage{epstopdf}
 
\usepackage[capposition=top,font=small]{floatrow}

 \usepackage{setspace}

  \newcommand{\blind}{1}
\def\spacingset#1{\renewcommand{\baselinestretch}%
{#1}\small\normalsize} \spacingset{1}

  \renewenvironment{thebibliography}[1]{%
    \begin{oldthebibliography}{#1}%
      \setlength{\parskip}{0ex}%
      \setlength{\itemsep}{0ex}%
  }%
  {%
    \end{oldthebibliography}%
  }

\makeatletter

\providecommand{\tabularnewline}{\\}

\long\def\symbolfootnote[#1]#2{\begingroup
\def\thefootnote{\fnsymbol{footnote}}\footnote[#1]{#2}\endgroup}

\numberwithin{equation}{section}
\numberwithin{figure}{section}
  \theoremstyle{remark}
  \newtheorem{rem}{\protect\remarkname}
  \theoremstyle{plain}
  \newtheorem{assumption}{\protect\assumptionname}
\theoremstyle{plain}
\newtheorem{thm}{\protect\theoremname}
  \theoremstyle{plain}
  
  \theoremstyle{plain}
  \newtheorem{lem}{\protect\lemmaname}

\allowdisplaybreaks
\usepackage{hyperref}
\hypersetup{
     colorlinks   = true,
     citecolor    = blue
}
\makeatother

\usepackage{babel}
  \providecommand{\algorithmname}{Algorithm}
  \providecommand{\assumptionname}{Assumption}
  \providecommand{\lemmaname}{Lemma}
  \providecommand{\remarkname}{Remark}
\providecommand{\theoremname}{Theorem}

\usepackage[bindingoffset=0.2in,%
            left=0.8in,right=0.8in,top=1in,bottom=1.2in,%
            footskip=.35in]{geometry}

\begin{document}

\def\spacingset#1{\renewcommand{\baselinestretch}%
{#1}\small\normalsize} \spacingset{1}

\if1\blind
{
  \title{\bf   Linear Hypothesis Testing \\ in Dense High-Dimensional Linear Models}
  \author{Yinchu Zhu  and Jelena Bradic \thanks{
    The authors gratefully acknowledge \textit{NSF support} through the grant DMS-1205296}\hspace{.2cm}\\
    Rady School of Management and Department of Mathematics\\ University of California at San Diego}
    \date{}
  \maketitle
} \fi

\if0\blind
{
  \bigskip
  \bigskip
  \bigskip
  \begin{center}
    {\LARGE\bf Linear Hypothesis Testing \\ in Dense High-Dimensional Linear Models}
\end{center}
  \medskip
} \fi

\begin{abstract}
\spacingset{1 }


We propose a methodology for   testing linear hypothesis in high-dimensional linear models. The proposed test 
does not  impose any  restriction on the size of the model, i.e. model  sparsity  or the loading vector representing the hypothesis.  
 Providing asymptotically valid methods for testing general linear functions of  the  regression parameters in high-dimensions is extremely challenging -- especially  without making restrictive or unverifiable assumptions on the number of non-zero elements.  
We propose to test the moment conditions  related to the newly designed restructured regression,  where the inputs are transformed and augmented features. These new features incorporate the structure of the null hypothesis directly. 
The test statistics  are constructed in such a way that lack of sparsity in the original model parameter does not present a problem for the theoretical justification of our procedures. 
We  establish asymptotically exact control on Type I error without imposing any sparsity assumptions on model parameter or the vector representing the linear hypothesis. 
 Our method is  also shown to achieve certain optimality in detecting deviations from the null hypothesis. We demonstrate  the favorable finite-sample performance of the proposed methods, via  a number of   numerical and  a real data example.

 \end{abstract}

\noindent%
\vfill

\spacingset{1 }

\section{Introduction}
A high-dimensional inference is a fundamental topic of interest in  modern scientific problems that are   widely recognized to be of  high-dimensional nature, i.e., that  require estimation of  parameters with dimensionality exceeding the number of observations. Applications   span a wide variety of scientific fields, such as biology, medicine, genetics,
neuroscience, economics, and finance.    Minimizing a  suitably regularized (quasi-)likelihood function was developed \citep{tibshirani96regression,fan2001variable} as a suitable approach for the estimation in such models.   In particular, high-dimensional linear models have been studied extensively in   recent years and take the following form 
\begin{equation}
y_{i}=x_{i}^\top \beta_{*}+\varepsilon_{i}, \qquad i=1,2,\dots,n \ \label{eq: original model}
\end{equation}
for a response   $y_i \in \mathbb{R}$, a feature vector
  $x_{i}\in\mathbb{R}^{p}$ and  the noise $\varepsilon_{i}\in\mathbb{R}$, such that  $E[\varepsilon_{i}]=0$
and $E[\varepsilon_{i}^{2}]=\sigma_{\varepsilon}^{2}$ with $0<\sigma_{\varepsilon}^{2}<\infty$.  The vector $\beta_* \in \mathbb{R}^p$ is the unknown model parameter and we allow  for $p\gg n$. We consider a random design setting with the feature vectors  satisfying $Ex_{i}=0$ and  $E[x_{i}x_{i}^\top ]=\Sigma_{X}$.    Under certain 
regularity conditions on the design matrix $X=(x_{1},x_2 ,\dots, x_n)^\top $, regularized methods with a suitable
choice of the tuning parameter have been shown to achieve the optimal rate of estimation
  as long as the vector $\beta_*$ is sparse in that   $ \|\beta_*\|_0 =o(n /\log p)$.

 The goal of the present article is to address  the testing problem for  linear hypotheses  of the form 
 \begin{equation} 
 H_{0}:\ a^\top \beta_{*}=g_{0}\label{eq: null hypo},
 \end{equation}
 where the loading vector $a\in \mathbb   R^{p}$ is pre-specified  and $g_{0}\in \mathbb R$ is  given, and design an asymptotically valid test statistic that does not rely on sparsity assumptions.
 Some central challenges have hindered the systematic
development of tools for statistical inference in such settings. 
The non-sparse nature of the  model parameter $\beta_*$
poses serious challenges to consistent estimation; moreover, the  size and structure of the loading vector $a$ introduce additional difficulty for the inference.  However,  in this article  we consider potentially dense vectors $\beta_*$ with $0\leq \|\beta_*\|_0 \leq p$. We  also  allow for the non-sparse loadings with \(1\leq \|a\|_0\leq p\).
 The inference problem for the mean of the response \(y_{i}\) conditional on \(x_{i}=a\),
 is a prototypical case for the general functional $a^\top \beta_*$ and 
 is a representative case for dense loading $a$.

 
 We develop the principles of  \textit{restructured regression,} where a hypothesis-driven {\it feature synthetization} is introduced. The feature augmentation is done in such a way to separate useful inferential information from the useless one, by ``projecting'' the original feature
 space to the space spanned by the vector $a$ and the space orthogonal
 to $a$. This orthogonal projection is introduced to achieve the above separation and avoid the curse of dimensionality. 
 Then, an appropriate moment condition in invoked on the restricted regression and a suitable test statistic constructed. 
 The structure of the moment condition and its test depends on whether or not  the covariance of the features \(\Sigma_X\) is known. When prior knowledge of \(\Sigma_X\) is available, the synthesized features can be created in such a way that the resulting moment condition and testing procedure do not depend on \(\beta_*\); thus,  estimation of \(\beta_*\) is completely avoided. As a result, no assumption on the sparsity of \(\beta_*\) is required. We establish  theoretical guarantees for Type I error control and show that the test can detect  the deviation from the null hypothesis  of the order \(O(\|a\|_2/\sqrt n)\). To the best of our knowledge, our approach  provides the first
 result on testing general linear hypothesis (\ref{eq: null hypo})   in     high-dimensional linear 
 models with potentially non-sparse (dense) parameters.

 When prior knowledge of \(\Sigma_X\) is unavailable, the  orthogonalization and perfect separation is not achievable due to the unknown projection matrix. We design an estimator of the  projection matrix and  further condition the new and augmented features in such a way that their correlations are estimable and yet the format of the restructured regression remains unchanged. The developed hypothesis-driven feature separation diminishes the impact of the  inaccuracy of an estimator of a transformation of \(\beta_*\).
%
%
%
Consequently, we can establish asymptotically exact control of  Type I error. We believe there is currently no result on  testing \(a^{\top }\beta_*\) in the case where \(\Sigma_X\) is unknown, and  both \(\beta_*\) and \(a\) are allowed to be dense. Moreover, when sparsity assumptions hold, our procedure is shown to achieve   optimality guarantees; hence, it does not  loose  efficiency.

  Since we do not assume sparsity in \(\beta_*\), our work does not directly compare to the existing results, which are only valid for sparse \(\beta_*\). However, in some cases, our work generalizes existing results to the non-sparse models. For example, \cite{cai2015confidence} show that when \(\Sigma_X\) is known, the   minimax length of the confidence
  interval for $a^{\top}\beta_{*}$ is of the order $O(\|a\|_{2}/\sqrt{n})$ if \(\|\beta_*\|_0=O(n/\log p)\). As confidence sets for \(a^{\top}\beta_*\) can be easily constructed by inverting the proposed tests, our results indicate that their  conclusion continues to hold for non-sparse models, where  \(\|\beta_*\|_0\) can be as large as \(p\).
  For the case of dense \(a\), we do not impose any constraint on \(a\). However,  existing work, such as 
  \cite{cai2015confidence}, imposes a lower bound (in terms of \(\|a\|_\infty\)) on the minimal non-zero coordinate of \(a\) -- a condition that is seldom satisfied for inference  of conditional mean,  when \(a\) is typically drawn   from   a continuous distribution (e.g. $a$ is drawn from the same distribution as  the distribution of the $x_i$'s).

 \subsection{Relation to existing literature}

Confidence intervals  and hypothesis testing play a fundamental role in statistical theory and applications. However, compared to the point estimation  there is still much
work to be done   for  statistical inference of high-dimensional models. Existing work on the inference problems  predominantly focuses on  individual
coordinates of \(\beta_*\).  Early work typically imposes conditions that guarantee consistent variable selection (see \cite{fan2001variable,zou2006adaptive,Zhao:2006})
or develops methods that lead  to conservative inferential guarantees (e.g. \cite{Buhlmann2013}). However, recent work focusses on asymptotically accurate inference  without  relying on  the variable selection consistency. Current advances in this domain are, however,  restricted to the ultra-sparse case, where  $\|\beta_*\|_0= o(\sqrt{n}/\log p)$; see \cite{zhang2014confidence,belloni2014inference, van2014asymptotically,javanmard2014confidence, ning2014general, 2015arXiv150802757J,mitra2014benefit,buhlmann2015,Belloni2015,Chernozhukov2015}. Under
such  sparsity condition, the expected length of the confidence intervals for individual coordinates
is of the order  $O(1/\sqrt{n})$ \citep{2016arXiv160100815V}.  \cite{cai2015confidence}
study the length of the confidence intervals allowing for \(\|\beta_*\|_0=o(n/\log p)_{}\) and discover that lack of explicit knowledge of $\|\beta_*\|_0$ can fundamentally limit the efficiency of confidence intervals.

However, there is little reason to believe that the sparsity of $\beta_*$
needs to hold in practice \citep{hall2014,5319752,2014arXiv1410.4578J,pritchard2001rare}. Unfortunately, there
is  almost no work  on estimating or testing  the true sparsity level of the underlying
parameter. Hence, the theory of  hypothesis testing   under general sparsity structures  is still a very challenging and important open problem. In particular, progress
is very much required when  $\|\beta_*\|_0$  is allowed to grow faster
than $n/\log p$ and perhaps even larger than the sample size $n$.
There are several
articles showing that the  regularized procedures
 have non-vanishing estimation errors in such settings \citep{donoho1994minimax,raskutti2011minimax,2016arXiv160303474C}.
However, is it still possible to develop a general methodology for
 testing ${\beta_*}$ in this case? Can one construct
valid inference procedures that do not require knowledge of  \(\|\beta_*\|_0\)?

In the proposed inference procedure, we handle the  high-dimensional, possibly non-sparse model parameters and/or non-sparse loadings,  by developing a new methodology for testing. The proposed methodology is centered around a construction of  augmented and synthesized features that are driven by a specific form of the null hypothesis. Compared with the previous approaches of de-biasing \citep{zhang2014confidence, javanmard2014confidence,van2014asymptotically,mitra2014benefit}, scoring  \citep{ning2014general,Chernozhukov2015}, double-selection \citep{belloni2014inference,Belloni2015}, our new  approach    has two major distinctive features:
\begin{itemize}
	\item We do not rely on a $l_1$norm consistent estimation of the unknown model parameters. In high-dimensional models with the lack of sparsity in the parameters, this may no longer be possible. Instead, we propose to reformulate the original parametric null hypothesis into a 
	moment condition  that can be successfully estimated even without sparsity in the model. 
	This moment condition is   different from the score equations employed for estimation as those are not estimable in non-sparse high-dimensional models. 
	\item  We advocate for  a study  and exploration of the correlation between feature vectors (and not the model parameters); this proves to be a valuable tool that overcomes the limit of estimation. Namely,  we propose that the features be  split and projected  onto the loading vector $a$ of the hypothesis \eqref{eq: null hypo}, thereby fully utilizing the null hypothesis structure. 
	This ``decoupling'' scheme allows for a successful estimation of the moment condition even without sparsity assumption. As a result the developed method provides a rich alternative to the classical Wald or Score principles.
\end{itemize}

  \subsection{Notation and organization of the article}
 We briefly describe notations used in the article.  
We use 
 $\rightarrow^{d}$  to denote convergence in distribution and 
 $\mathcal N(0,1)$ to denote  the standard normal distribution with its cumulative distribution function  denoted by  $\Phi(\cdot)$. The (multivariate) normal distribution with mean (vector) \(\mu\) and variance (matrix) \(\Sigma\) is denoted by \(\mathcal N (\mu,\Sigma)\).
We use  $^\top $ to denote the transpose of (a vector or matrix) and denote by    $I_p$ the $p \times p$ identity matrix.
For a vector \(a=(a_1,\cdots,a_p)^{\top }\in \mathbb R ^p\),  its $l_0$ norm  is the cardinality of $\text{supp}(a)= \{ i\mid a_i \neq 0 \}  $ and $\|a\|_\infty=\max \{ |a_1|,\cdots,|a_p|\} $; \(\|a\|_1\) and \(\|a\|_2\) denote the $l_1$ and $l_2$  norm of  $a$, respectively. In this case, \(a_{-i}\) denotes the vector \(a\) with its \(i\)th coordinate removed. For two sequences of positive constants \(a_n\) and \(b_n\), we use \(a_n \asymp b_n\) to denote that \(a_n/b_n=O(1)\) and \(b_n/a_n=O(1)\). For two real numbers \(a_1\) and \(a_2\), \(a_1 \vee a_2\) and \(a_1\wedge a_2\) denote \(\max\{a_1,a_2\}\) and \(\min\{a_1,a_2\}\), respectively.

The rest of this article is organized as follows.  Section \ref{sec:Methodology} introduces the main methodology under known \(\Sigma_X\) and establishes theoretical properties of the proposed test.  Section  \ref{sec: extension} extends the proposed methodology to the case of the unknown $\Sigma_X$ and provides    theoretical results.  Section \ref{sec:examples} contains examples illustrating new methods that the proposed methodology brings to the literature on high-dimensional inference. 
 Section \ref{sec:Numerical-results} contains detailed numerical experiments on a number
of dense  high-dimensional linear models, including  sparse and dense loadings $a$. In Section \ref{sec:MC}, we demonstrate the excellent finite-sample performance of  the proposed methods through Monte Carlo simulations; in Section \ref{sec:rd}, we illustrate our method via  a real data study. Appendix \ref{ap:a} contains complete  details of the theoretical derivations.

\section{\label{sec:Methodology} Testing \(H_0:\ a^\top \beta_*=g_0\) with prior knowledge of  $\Sigma_X$}

In this section we promote a unified approach to a wide class of decision problems.   Our main building block (which we believe is important in
its own right) is a construction, named {\it restructured regression}
allowing, under weak assumptions, to build 
tests for hypotheses on  \(a^{\top}\beta_*\), where \(\beta_*\) and/or \(a\)  can be  non-sparse. Considering the potential failure of sparsity in many practical problems, we strongly believe that our
approach  permits  a diverse spectrum of applications. In this section our focus is to introduce the method with known
$\Sigma_{X}$ (an assumption relaxed in the next section).

Throughout the paper, we denote   $\Omega_X = \Sigma_X^{-1}$. In the sequel, given the feature vector $x_{i}\in\mathbb{R}^{p}$
and loading vector $a\in\mathbb{R}^{p }$, we consider the following  decomposition
\begin{equation}
x_{i}=az_{i}+w_{i},\label{eq: feature decomposition}
\end{equation}
with a scalar  
$$z_{i}=\left( \frac{\Omega_X a}{a^\top\Omega_X a} \right)^\top x_{i}$$
and  a $p$-dimensional vector 
$$w_{i}=\left[I_{p}- \frac{aa^\top\Omega_X }{a^\top\Omega_X a}\right]x_{i}.$$
Observe that  $az_i$ can be  viewed as the  projection   of $x_{i}$ onto the vector $a$ -- taking into account \( \Omega_X\), hence extracting information in \(x_{i}\) regarding the null hypothesis. Notice that the model
(\ref{eq: original model}) and decomposition (\ref{eq: feature decomposition})
imply 
\begin{equation}
y_{i}=z_{i}\cdot(a^\top\beta_{*})+w_{i}^\top \beta_{*}+\varepsilon_{i},\label{eq: restructured model}
\end{equation}
referred to as \textit{restructured regression}.  
The proposed construction gives rise to the method of  \textit{feature
customization}. Given covariate vector $x_{i}$ and the loading vector  $a$
representing the structure of the null hypothesis, we create the synthesized features $\tilde{x}_{i}:=(z_{i},w_{i}^\top )^\top $
so that the regression coefficient for $z_{i}$ in the restructured
regression (\ref{eq: restructured model}) is the quantity under testing.

\begin{rem}
The synthesized features are not only an artifact of our new methodology but also admit intuitive interpretations. Consider the case where \(\Sigma_X\) is known to be \(I_p\). The synthesized features \(z_i\) and \(w_i\) represent the relevant and the irrelevant information with respect to the  null, respectively. To see this, suppose that  the true distribution of the data is known. With the population expectations, we can  identify the parameters in the restructured regression (\ref{eq: restructured model}): \(E(z_i y_i)=Ez_i^2 (a^\top \beta_*)\) and \(Ew_i y_i=Ew_{i}w_{i}^\top \beta_{*}\). Notice that the latter equation contains no information regarding \(a^\top \beta_*\) because it can be shown that \(a\) is orthogonal to columns in \(Ew_i w_i^\top\). In other words, knowing \(Ew_{i}w_{i}^\top \beta_{*}\) does not lead to knowing  \(a^\top \beta_*\). Therefore, \(a^\top \beta_*\) is identified with the distribution of \((y_i,z_i)\) and \(w_i\) does not contain information about the null hypothesis.
\end{rem}

It is not hard to verify that, by the construction of the transformed features,
$E[w_{i}z_{i}]=0$. It follows that $E\left[z_{i}(y_{i}-z_{i}g_{0})\right]=E\left[z_{i}\left(\varepsilon_{i}+w_{i}^\top \beta_{*}+z_{i}(a^\top\beta_{*}-g_{0})\right)\right]=E\left[z_{i}^{2}(a^\top\beta_{*}-g_{0})\right].$ Observe that the last expression is $0$ if and only if  the null hypothesis \eqref{eq: null hypo} holds. 
As a result, testing \(H_{0}\) in  \eqref{eq: null hypo}  is equivalent to testing   the following moment condition:
\begin{equation}
H_{0}:E \left[z_{1}(y_{1}-z_{1}g_{0})\right]=0.\label{eq:corr}
\end{equation}
To test the above condition,
we propose a studentized test statistic,  $T_n (g_0)$, taking the form 
 \begin{equation}\label{eq:T_n}
T_{n}(g_{0}):=\frac{n^{-1/2}\sum_{i=1}^{n}l_{i}(g_{0})}{\sqrt{n^{-1}\sum_{i=1}^{n}l_{i}(g_{0})^{2}}},
\end{equation}
with $l_{i}(g_{0})=z_{i}(y_{i}-z_{i}g_{0})$. For a test of \(H_{0}\) with nominal size \(\alpha \in (0,1)\), we reject \(H_0\) if
$$
 |T_n(g_0)|>\Phi(1-\alpha/2).
 $$
 
 The methodology proposed above is novel in a number of aspects. Unlike Wald or Score or Likelihood  principles,  centered around  a consistent estimator of $\beta_*$, our methodology allows for  extremely fast implementation and  does not  estimate the unknown parameter $\beta_*$. The novel methodology consists of two-stages.
 At the first stage,  our procedure  establishes a data-driven feature decomposition    based on the structure of the null hypothesis directly. At the second stage, only ``a moment condition'' of the restructured regression is tested. 
 It is critical to observe that restructured regression
by itself is not sufficient to guarantee  valid inference.
The  novel properties of the proposed method are based on the built-in, i.e., designed  orthogonality of the synthesized features $z_i$ and $w_i$. As such it 
enables us to construct  a test statistic  that  does not contain the unknown parameter \(\beta _{*}\), thereby allowing our methodology to handle dense (and thus possibly non-estimable) $\beta_{*}$.  
 Moreover, no assumption is imposed
on the loadings $a$ either. As we will see in the next section, these
properties under known $\Sigma_{X}$ propagate  to the case of
the unknown $\Sigma_{X}$ and underline all further developments.


%

\begin{assumption}
\label{assu: lyapunov CLT} Let the following hold: 
 (i) there exists a  positive constant $C$ such that $E|z_{i}\sigma_{z}^{-1}|^{8}\leq C$,
$E\varepsilon_{i}^{8}\leq C$ and $E|w_{i}^\top \beta_{*}|^{8}<C$ with $C<\infty$.  Moreover,
 (ii) there exists a constant $c\in(0,\infty)$, such that $\sigma_{\varepsilon}\geq c$.
Lastly,
 (iii) there exist constants $D_{1},D_{2}>0$ such that the eigenvalues
of $\Sigma_{X}$ lie in $[D_{1},D_{2}]$. 
\end{assumption}
The stated conditions in Assumption \ref{assu: lyapunov CLT} are very weak and 
intuitive. 
Assumption \ref{assu: lyapunov CLT}(i) requires components in the
restructured regression (\ref{eq: restructured model}) to have bounded
eighth moments.  Assumption \ref{assu: lyapunov CLT}(ii) rules out
the noiseless regression setting in the original model (\ref{eq: original model}).
Assumption \ref{assu: lyapunov CLT}(iii) is very weak in that it
only imposes well-designed covariance matrix of  the features $x_{i}$ (see \cite{bickel2009simultaneous}).

Notice that Assumption \ref{assu: lyapunov CLT} does not require
any condition regarding the sparsity of $\beta_{*}$. Even in the
case of sparse $a$, existing work, such as the debiasing method,
heavily relies on the sparsity of $\beta_{*}$. Results regarding
dense $a$ are very limited even for sparse \(\beta_*\). \citet{cai2015confidence} impose the
condition of $\max_{j\in{\rm supp}(a)}|a_{j}|/\min_{j\in{\rm supp}(a)}|a_{j}|=O(1)$; however, such a condition is quite hard to satisfy if $a$ is drawn
from a continuous distribution whose support contains zero. In contrast,
our results do not require any condition on $a$ and, hence, bridge the gap in the existing literature on high-dimensional inference.  

\begin{thm}
\label{thm: known variance X}Consider the model in (\ref{eq: original model})
and the definition of $z_{i}$ and $w_{i}$ as  in (\ref{eq: feature decomposition}).
Suppose that Assumption \ref{assu: lyapunov CLT} holds. Under $H_{0}$
in (\ref{eq: null hypo}),  we have that (1) the test statistic $T_n$, \eqref{eq:T_n}, satisfies 
$
T_{n}(g_{0}) \rightarrow^{d}\mathcal{N}(0,1)$ as $n,p\rightarrow \infty$ and that (2)\begin{displaymath}
\lim_{n,p\to \infty}P\Bigl(|T_{n}(g_{0})|>\Phi^{-1}(1-\alpha/2)\Bigl)=\alpha.
\end{displaymath}
\end{thm}  
Theorem \ref{thm: known variance X}  gives an asymptotic approximation for the null distribution of the test statistic $T_n(g_0)$ under
general sparsity  structure.
The   result  of Theorem \ref{thm: known variance X}  has two striking features. The first is that it holds, no matter  the size or sparsity of the loading vector $a$. The second is that the proposed test guarantees Type I error control when $p \geq n$ and $p,n \to \infty$  no matter of the sparsity of $\beta _* $ and without the knowledge of the  noise level $\sigma_\varepsilon$; in particular, it allows $\|\beta_*\|_0 =p$. Therefore, our test is fully adaptive, in the sense that its validity does not depend on   in  the sparse/dense level of either the model parameter \(\beta_*\) or the hypothesis loading \(a\).
We also show that our test can detect deviations from the null that
are larger than $O(\|a\|_{2}/\sqrt n)$ while allowing   $\beta_*$ to be non-sparse and $p \geq n$. \begin{thm}
\label{thm: known variance X power}Under the conditions of Theorem \ref{thm: known variance X},
suppose that $a^\top\beta_{*}=g_{0}+h_{n}$ and  $\sqrt{n}|h_{n}|/\|a\|_{2}\rightarrow\infty$. Then, for any $\alpha\in(0,1)$. 
$$ \lim_{n,p\to \infty}P\Bigl(|T_{n}(g_{0})|>\Phi^{-1}(1-\alpha/2)\Bigl)=1.$$
  \end{thm}
\begin{rem}
Theorem \ref{thm: known variance X power} also suggests that we can expect  the length of the confidence interval for $a^\top\beta_{*}$ (obtained by inverting the proposed test)
to be of  the order of  $O(\|a\|_{2}/\sqrt n)$  regardless of the sparsity of $\beta_{*}$
or $a$.  To the best of our knowledge,
it is the first result to explicitly allow  non-sparse and simultaneously high-dimensional parameters $\beta^*$ or vector loadings $a$.
It is also closely connected with the existing results
for the case of sparse parameters $\beta^*$.  \citet{cai2015confidence},
state that under Gaussianity and sparsity
in both $\beta_{*}$ and $a$ together with known $\Sigma_{X}$ and
$\sigma_{\varepsilon}$, the optimal expected length of confidence
intervals for $a^\top\beta_{*}$ is of the order $O( \|a\|_{2}/\sqrt n)$ (see Theorem 7 therein). Observe that our procedure achieves the same optimality without the knowledge of $\sigma_{\varepsilon}$ and allowing dense vectors $\beta_*$.
\end{rem}
We do not formally claim that this is the optimal rate for dense $\beta_{*}$,
but we can consider   an obvious benchmark. Let $\bar {\beta}$ be
an estimator that attains an efficiency similar to (ordinary least square) OLS in low dimensions, i.e., 
$\bar{\beta}$ is  distributed as $\mathcal N(\beta_{*},\Omega_X\sigma_{\varepsilon}^{2}/n)$.
Then  $a^\top\bar{\beta}$  follows $\mathcal N(a^\top\beta_{*},a^\top\Omega_Xa\sigma_{\varepsilon}^{2}/n)$ distribution.
Since $\Omega_{X}$ has eigenvalues bounded away from infinity, the standard
deviation of $a^\top \bar{\beta}$ is of the order $\|a\|_{2}/\sqrt n$.
Such an estimator might not be feasible in practice, but could serve
as a benchmark for dense $\beta_{*}$. A rigorous study of the efficiency issue is likely to yield results that are quite  different from current literature since existing results, e.g., \cite{cai2015confidence}, do not naturally extend to dense problems. For example, consider the case of $\|a\|_0=\|\beta_{*}\|_0=p$, naively extending Theorem 8 of  \cite{cai2015confidence} would conclude that the minimax expected length of a confidence interval for $a^{\top}\beta_{*}$ is of the order $\|a\|_{\infty}p\sqrt{(\log p)/n}$; however, this rate is larger than the rate $\|a\|_{2}/\sqrt{n}$, which is bounded above by $\|a\|_{\infty}\sqrt{p/n}$.  Lastly, according to Theorem \ref{thm: known variance X power} our proposed test achieves the same rate at the benchmark $\bar \beta$.

\section{\label{sec: extension} Testing \(H_0:\ a^\top \beta_*=g_0\) without prior knowledge of  $\Sigma_{X}$}

The approach proposed in this section tackles the high-dimensional  inference problem in a very general setting. The focus is the more realistic scenario in which the covariance matrix $\Sigma_X$ and  the variance of the model \eqref{eq: original model} are both unknown. We synthesize new features, create a new reference model and explore the correlations therein  in order to design  a suitable inferential procedure that is stable without sparsity assumption.

\subsection{Feature synthetization and restructured regression}

In order to design inference when  \(\Sigma_X\) unknown, we take on a new perspective and build upon the methodology of Section \ref{sec:Methodology}. Consider feature synthetization of Section \ref{sec:Methodology} where $\Sigma_X$ is naively treated as  $I_p$,
\begin{equation}\label{eq: z and w with unknown SigmaX}
z_{i}=\left( \frac{a}{a^{\top} a} \right)^\top x_{i} \in \mathbb{R} \quad{\rm and}\quad w_{i}=\left(I_{p}-aa^\top/(a^\top a)\right)x_{i}\in \mathbb{R}^p.
\end{equation}   Although the decomposition   $x_{i}=az_{i}+w_{i}$ still holds, features
   $z_{i}$ and $w_{i}$ might be correlated (because $\Sigma_{X} \neq I_p$). 
If such correlation is estimated successfully, we can  use certain decoupling method to eliminate the impact of dense parameters while allowing exponentially growing dimensions.

 The first challenge is that   directly estimating the correlation
between $z_{i}$ and $w_{i}$  (as defined)  is not achievable (as  the restricted eigenvalue (RE)  condition \cite{bickel2009simultaneous} on \(W=(w_1,\cdots,w_n)^\top\) is violated).
To address this problem, we propose to {\it stabilize} the feature vector $w_{i}$ and define {\it stabilized}  features $\tilde w_i$. We stabilize the features in such a way that  the RE condition on the stabilized design $\tilde W=(\tilde w_{1},\cdots,\tilde {w}_{n})^\top $ is satisfied with high probability.
Since $I_{p}-aa^\top/(a^\top a)$ is a projection matrix, we can find  $U_{a}\in\mathbb{R}^{p\times(p-1)}$  an orthogonal matrix such that 
$$U_{a}^\top U_{a}=I_{p-1} \qquad \mbox {and}
 \qquad I_{p}-aa^\top/(a^\top a)=U_{a}U_{a}^\top . $$
 Then 
$$W\beta_{*}=X(I_{p}-aa^\top/(a^\top a))\beta_{*}=XU_{a}U_{a}^\top \beta_{*}=\tilde{W}\pi_{*},$$
where 
$$\tilde{W}=  WU_a \qquad \mbox{and} \qquad \pi_{*}=U_{a}^\top \beta_{*}.$$ Since $y_{i}=z_{i}\cdot(a^\top\beta_{*})+w_{i}^\top \beta_{*}+\varepsilon_{i}$,
we have the {\it stabilized model }
\begin{equation}\label{eq:consolidated}
y_{i}=z_{i}\cdot(a^\top\beta_{*})+\tilde{w}_{i}^\top \pi_{*}+\varepsilon_{i}.
\end{equation}
The model is balanced in the sense that 
$E\tilde{W}^\top \tilde{W}/n=U_{a}^\top \Sigma_{X}U_{a}\in\mathbb{R}^{(p-1)\times(p-1) }$ with eigenvalues   bounded away from zero and infinity. Therefore, RE condition on \(\tilde{W}\) holds under weak conditions; see  \cite{rudelson2013reconstruction}.


\begin{rem}
\label{rem: sparse example}The synthesized feature $w_{i}\in\mathbb{R}^{p}$
is consolidated into $\tilde{w}_{i}\in\mathbb{R}^{p-1}$, in that $\tilde{w}_{i}$
has a smaller dimensionality and can be used to recover $w_{i}$ via
$w_{i}=U_{a}\tilde{w}_{i}$. In this sense, $\tilde{w}_{i}$ contains
all the information in $w_{i}$. As an example, consider the case with $a$ being the first column of
$I_{p}$. In this case, it is not hard to verify that $z_{i}=x_{i,1}$, $w_{i}=(0,x_{i,2},\cdots,x_{i,p})^\top \in\mathbb{R}^{p}$,
$U_{a}=(0,I_{p-1})^\top \in\mathbb{R}^{p\times(p-1)}$ and thus $\tilde{w}_{i}=U_{a}^\top w_{i}=(x_{i,2},\cdots,x_{i,p})^\top \in\mathbb{R}^{p-1}$.

\end{rem}
We now introduce an additional model to account for the dependence
between the {\it synthesized feature} $z_{i}$ and the {\it stabilized  feature}
$\tilde{w}_{i}$: 
\begin{equation}
z_{i}=\tilde{w}_{i}^\top\gamma_{*}+u_{i},\label{eq: graphical model}
\end{equation}
where $\gamma_{*}\in\mathbb{R}^{p-1}$ is an unknown parameter and $u_{i}$
is independent of $\tilde{w}_{i}$ with $Eu_{i}=0$ and $Eu_{i}^{2}=\sigma_{u}^{2}$.

In this article, we will assume that $\gamma_{*}$ is sparse, in order  to decouple the
dependence between $z_{i}$ and $\tilde{w}_{i}$ with the unknown $\Sigma_{X}$. %
In fact, sparse $\gamma_{*}$ is a generalization of the sparsity
condition on the precision matrix $\Omega_{X}$,
a regularity condition typically imposed in the literature; see \citet{van2014asymptotically},
\citet{belloni2014inference,Belloni2015} and \citet{ning2014general}. Recall
the example in Remark \ref{rem: sparse example}. Since $x_{i,1}=z_{i}=\tilde{w}_{i}^\top \gamma_{*}+u_{i}=x_{i,-1}^\top \gamma_{*}+u_{i}$,
it is not hard to show that the first row of $\Omega_{X}$ is $(\sigma_{u}^{-2},-\sigma_{u}^{-2}\gamma_{*}^\top )$.
Hence, the sparsity of $\gamma_{*}$ is equivalent to the sparsity
in the first row of $\Omega_{X}$. The sparsity of $\gamma_{*}$ can be justified for dense $a$ as well. Consider the case of $\Sigma_X=c I_{p}$ for some $c>0$;   a prototypical model in compressive sensing corresponds to $c=1$  \citep{nickl2013confidence}. In this case, one can easily show that $z_i$ and $\tilde{w}_{i}$ are uncorrelated, meaning that $\gamma_{*}=0$ for any $a$. The synthesized features also admit intuitive interpretations in this case:   the feature $z_{i}$ contains useful information in testing the null hypothesis $a^\top \beta_*=g_0$, while the consolidated $\tilde{w}_{i}$ contain information not useful for inference.

Now, we are ready to construct the moment condition of interest.
 Observe that under $H_{0}$ in (\ref{eq: null hypo}), $y_{i}-z_{i}g_{0}-\tilde{w}_{i}^\top \pi_{*}=\varepsilon_{i}$
is uncorrelated with $z_{i}-\tilde{w}_{i}^\top \gamma_{*}=u_{i}$. If $H_{0}$
is false, then $y_{i}-z_{i}g_{0}-\tilde{w}_{i}^\top \pi_{*}=\varepsilon_{i}+z_{i}(\theta_{*}-g_{0})=\varepsilon_{i}+\tilde{w}_{i}^\top \gamma_{*}(\theta_{*}-g_{0})+u_{i}(\theta_{*}-g_{0})$
has non-zero correlation with $u_{i}=z_{i}-\tilde{w}_{i}^\top \gamma_{*}$. 
Hence, the initial null hypothesis, \eqref{eq: null hypo} is equivalent to the following null hypothesis  
\begin{equation}
H_0: E\Bigl[\left(z_{1}-\tilde{w}_{1}^\top \gamma_{*}\right) \left(y_{1}-z_{1}g_{0}-\tilde{w}_{1}^\top \pi_{*}\right)\Bigl]=0.\label{eq: moment condition}
\end{equation}
Directly testing this moment condition is not feasible, due to the unknown values of parameters $\gamma_{*}$ and $\pi_{*}$. As a result, we first provide estimates for these unknown parameters and consider the test statistic given by the studentized  statistics.

We make a few  remarks about the above  proposed methodology.  As
mentioned above, the existing literature on high-dimensional inference adopts the approach of relying on an (almost) unbiased estimate of the model parameter to distinguish the null and alternative hypotheses. 
The existing methods largely differ by the means of constructing the unbiased estimate and/or its asymptotic variance. 
Many use  an approximation of a  one-step Newton method \citep{zhang2014confidence,van2014asymptotically, javanmard2014confidence}   to achieve consistency in estimation of possibly all $p$ parameters. In order to test $a^{\top}\beta_{*}$ in this framework, one need to show that the debiased estimator for $\beta_{*}$ can be used to construct an asymptotically unbiased and normal estimator for $a^{\top}\beta_{*}$; to the best of our knowledge, a formal theoretical justification is yet to be established even under sparse $\beta_{*}$. Other than the debiasing technique, some proposals center around Neyman's score orthogonalization ideas \citep{belloni2014inference,Belloni2015,Chernozhukov2015,ning2014general}. It is worth pointing out that such a method requires a clear separation of parameter under testing and the nuisance parameter. In the original problem, the model parameter is $\beta_{*}$ and the quantity under testing is $a^{\top}\beta_{*}$; hence, it is not clear how to define the nuisance parameter since the $a^{\top}\beta_{*}$ is not just one entry (or a subset) of the parameter vector $\beta_{*}$.  Lastly, the work of \cite{cai2015confidence} propose  a minimax optimal test  that allows for dense loadings vector $a$, however in the dense case it provides a conservative error bounds and requires the knowledge of the sparsity size $s$.


 However, our proposal deviates from the above methodologies in a few aspects. Firstly, we design a test statistic irrespective of a consistency of high-dimensional estimators for the model parameter; hence, any refitting or one-step approximations are unnecessary. Secondly, we aim to orthogonalize design features (rather than model parameters) by directly taking into  account the structure of the null hypothesis (represented by $a$ and $g_{0}$). In this way we achieve full adaptivity to the hypothesis testing problem of interest. Thirdly, we reformulate the original parametric hypothesis into a moment condition of which we provide  adaptive estimators. The moment condition itself is not a simple first-order optimality identification (related to Z-estimators), but rather a moment that utilizes the special feature orthogonalization and fusion. Hence, even in setting where the existing work applies, our proposed method provides an alternative. However, apart from existing work, our proposed method applies much more broadly.

\subsection{Adaptive estimation of the unknown quantities}

 In this subsection, we start with a brief introduction of the Dantzig selector, which is the basis of our estimators. Then we introduce the intuition and steps of our estimator as well as  implementation details.

\subsubsection{Dantzig selector review}

Numerous studies have been conducted in regards to  the consistent estimation of high-dimensional parameters   in linear models. The canonical examples of successful estimators represent Lasso and Dantzig selector, defined as $\hat \beta_l$ and $\hat \beta_d$ below,
\begin{equation}\label{eq: DS original}
\hat \beta_l = \arg\min_{\beta \in \mathbb{R}^p} \left\{ \|Y - X \beta  \|_2^2 + \lambda_l \|\beta\|_{1} \right\} , \qquad 
\begin{array}{ccrcl}
\hat{\beta}_d&= &\underset{\beta \in \mathbb{R}^{p}}{\arg\min} \ \|\beta\|_{1}&& \\ 
  & s.t  \qquad  &  \left \|n^{-1} X^\top (Y- X \beta) \right \|_{\infty}& \leq&\lambda_d .
\end{array}
\end{equation}
Although Lasso and Dantzig selector are defined in different times,  \cite{bickel2009simultaneous}  established equivalence between the two estimators  under the conditions of  moderate design correlations and model sparsity, $\|\beta_*\|_0 \ll n$.  Between these two estimator, the Dantzig selector, $\hat \beta_d$, offers easy implementation through linear programming techniques.
 Moreover, the constraint in the Dantzig selector can be interpreted as a relaxation of the least squares normal equations, $X^\top Y = X^\top X \beta$. However, the performance of both estimators is tightly connected to the choice of their respective tuning parameters $\lambda_l$ and $\lambda_d$, i.e. the size of such relaxation.
 Several empirical and
theoretical studies emphasized that  tuning parameters should be chosen proportionally to the noise standard deviation $\sigma_\varepsilon$, i.e. $\lambda_d=\lambda_d(\sigma_\varepsilon)=\sigma_\varepsilon \sqrt{(\log p)/n}$. In such settings one can guarantee  $\|\hat \beta_l - \beta_*\|_1 =O(\|\beta_{*}\|_{0}\sqrt{(\log p)/n})$. Unfortunately,
in most applications, the variance of the noise is unavailable. It is therefore vital to design statistical procedures that estimate unknown parameters  together with the size of model variance in
a joint fashion. 
 This topic received special attention, cf. \cite{giraud2012} and the references therein. Most popular $\sigma$-adaptive
procedures, the square-root Lasso  \citep{Belloni01122011}, the  scaled Lasso \citep{sun2012scaled} and the self-tuned Dantzig selector \citep{gautier2013pivotal,2010arXiv1012.1297B} can be
seen as maximum a posteriori   estimators with a particular choice of prior distribution.
 However they do not provide estimates that  are reasonable in non-sparse and high-dimensional models -- after all in such settings it is impossible to consistently estimate the model parameters (see for more details \cite{2016arXiv160303474C} and \cite{raskutti2011minimax}).
  The aim of the present section 
  is to present an alternative to these methods, which are closely related, but presents some advantages
in terms of implementation and a more transparent theoretical analysis in not necessarily sparse models; the main benefit is that our estimates are well controlled in certain sense.

\subsubsection{Modified  Dantzig selector: adaptive to signal-to-noise ratio}

We start with the estimator for \(\pi_*\), a parameter that is high-dimensional and yet not necessarily sparse. We extend the Dantzig selector above to conform to the testing problem that we have to perform. We begin by splitting the tuning parameter into a constant independent of the variance of the noise and introduce a parameter $\rho$, a  square root of the noise to response ratio  as an unknown in the optimization problem. At
the population level, $\rho$ is intended to represent  $\sigma_{\varepsilon}/\sqrt{E(y_{1}-z_{1}g_{0})^{2}}$
and $\rho_{0}$ is a lower bound for this ratio. One might attempt
to use scaled Lasso by \citet{sun2012scaled} or self-tuning dantzig
selector proposed by \citet{gautier2013pivotal}, but for non-sparse
$\pi_{*}$, these methods cannot ensure that the estimated noise variance
is bounded away from zero whenever the vector \(\pi_*\) is a dense vector (a case of special interest here).

For \(Z=(z_1,\cdots,z_n)^\top \) and \(Y=(y_1,\cdots,y_n)^\top \) defined in (\ref{eq: z and w with unknown SigmaX}), we  introduce  the following version of Dantzig selector  of $\pi_*$  
\begin{equation}\label{eq: dantzig pi}
\begin{array}{ccrcl}
(\hat{\pi},\hat{\rho}) &= & \underset{(\pi,\rho)\in\mathbb{R}^{p-1}\times \mathbb{R} }{\arg\min}\|\pi\|_{1}& &\\
  & s.t  \qquad  & \left \| \tilde{W}^\top (Y-Zg_{0}-\tilde{W}\pi) \right \|_{\infty}&\leq&\eta \ \rho \   \sqrt{n}\|Y-Zg_{0}\|_{2} \\
    & & \left(Y-Zg_{0}\right)^\top \left(Y-Zg_{0}-\tilde{W}\pi \right)&\geq& \rho_{0}\ \rho \  \|Y-Zg_{0}\|_{2}^{2}/2 \\
    & &   \rho  &\in& [\rho_0,1],
\end{array}
\end{equation}
where $\eta\asymp\sqrt{n^{-1}\log p}$   and $\rho_{0}\in(0,1)$ are  scale-free
tuning parameters.

The estimator (\ref{eq: dantzig pi}) is different from  (\ref{eq: DS original}) in two ways. First, the estimator (\ref{eq: dantzig pi}) simultaneously estimates $\pi_{*}$ and $\rho$. We introduce a $\rho_{0}$ the lower bound for $\rho$ as a tuning parameter. Second, the estimator (\ref{eq: dantzig pi}) has an additional constraint, which essentially serves as an upper bound for $\rho$. The intuition of this bound is the following. When $\pi$ is replaced by the true $\pi_*$ and the null hypothesis holds, this constraint (scaled by $1/n$) becomes $\pi_{*}^{\top}\tilde{W}^{\top}\varepsilon/n+\varepsilon^{\top}\varepsilon/n \geq \rho_{0}\rho\|\tilde{W}\pi_{*}+\varepsilon\|_{2}^2/n $. By the law of large numbers, this means that $o_P(1)+\sigma_{\varepsilon}^2\geq\rho_{0}\rho E(y_{1}-z_{1}g_{0})^2$, which is satisfied if $\rho=\sigma_{\varepsilon}/\sqrt{E(y_{1}-z_{1}g_{0})^{2}}$ and $\rho>\rho_{0}$.

The vector $\varepsilon \ = Y-Zg_{0}-\tilde{W}\pi_* $  is a residual vector of the stabilized  model \eqref{eq:consolidated} under the null hypothesis $H_0$. The first constraint  on the residual vector  imposes that for each $i$, much like the Dantzig selector, $\hat \beta_l$, maximal correlation $\|\tilde{W}^{\top}\varepsilon /n\|_{\infty}$ is not larger than the  noise level $\eta   \sigma_{\varepsilon} $. Yet, in contrast to $\hat \beta_l$, our estimator treats $\rho$ as an unknown quantity and estimates it simultaneously with $\pi_*$. Moreover, we introduce the second constraint to stabilize estimation of the moment of interest \eqref{eq: moment condition} in the presence of non-sparse vectors $\pi_*$. Under the null hypothesis, this constraint prevents choice of $\rho$ that is too large; namely, it constraints $ \rho \leq C \left(Y-Zg_{0}\right)^\top \varepsilon/ \left\|Y-Zg_{0}\right\|_2^2$ for a finite constant $C>0$. In sparse settings, this additional constraint is redundant, so we remove it from our estimator of $\gamma_*$ defined below (a vector that is assumed to be sparse).
Hence, we consider the following  estimator,$\hat \gamma$ \begin{equation}\label{eq: dantzig gamma}
\begin{array}{ccrcl}
\hat{\gamma}&= &\underset{\gamma \in \mathbb{R}^{p-1}}{\arg\min} \ \|\gamma\|_{1}&& \\ 
  & s.t  \qquad  &  \left \|n^{-1} \tilde W^\top (Z-\tilde{W}\gamma) \right \|_{\infty}& \leq&\lambda n^{-1/2}\|Z\|_{2} 
\end{array}
\end{equation} 
where $\lambda\asymp\sqrt{n^{-1}\log p}$ is a scale-free tuning parameter and $n^{-1/2}\|Z\|_{2}$ serves as an upper bound of the unknown $\sigma_u$ in the model \eqref{eq: graphical model}. It is worth pointing out that the defined estimators change with a change in the hypothesis testing problem \eqref{eq: null hypo} through the new, synthesized and stabilized feature vectors $\tilde W$ and $Z$ together with $g_0$. We present a few examples in Section 4.

\subsubsection{Implementation}
The  optimization
problem in (\ref{eq: dantzig pi}), a generalization of the Dantzig selector \citep{candes2007dantzig}, can be recast as a linear program; the computational burden of our method is comparable to the Dantzig selector.
  Define scalars $d_{1}=\rho_{0}\|Y-Zg_{0}\|_{2}^{2}/2$,
 	$d_{2}=\|Y-Zg_{0}\|_{2}^{2}$, vectors $D_{1}=\tilde{W}^{\top}(Y-Zg_{0})\in\mathbb{R}^{p-1}$
 	and $D_{2}=\sqrt{n}\eta\|Y-Zg_{0}\|_{2}\mathbf{1}_{p-1}$ and matrix
 	$D_{3}=\tilde{W}^{\top}\tilde{W}\in\mathbb{R}^{(p-1)\times(p-1)}$.
	
 	Then, (\ref{eq: dantzig pi}) is equivalent to the following linear program
\begin{equation}
	 	\begin{array}{cccc}
 		\min_{(c,\pi,\rho)\in\mathbb{R}^{p-1}\times\mathbb{R}^{p-1}\times\mathbb{R}} &   \mathbf{1}_{p-1}^{\top}c&\\
 		s.t. &   -c\leq&\pi&\leq c\\
 		&    \rho_{0}\leq&\rho&\leq1\\
 		&   &d_{1}\rho+D_{1}^{\top}\pi&\leq d_{2}\\
 		&    -D_{2}\rho\leq &D_{1}-D_{3}\pi&\leq D_{2}\rho,
 	\end{array}
	\end{equation}
 	where the optimization variables are $c\in\mathbb{R}^{p-1}$, $\pi\in\mathbb{R}^{p-1}$
 	and $\rho\in\mathbb{R}$. 
	 For application purposes  we propose to choose the following choices of the tuning parameters: $\rho_{0}=0.01$ and $\eta=\sqrt{2\log(p)/n}$. They are universal choices and we show in simulations that they provide good results.

\subsection{Test Statistic}
With defined estimators of $\gamma_*$
and $\pi_*$, we are ready to define a sample analog of the moment condition \ref{eq: moment condition}.
Under our proposed method, a test of nominal size \(\alpha \in (0,1)\) rejects $H_{0}$ in (\ref{eq: null hypo}) if $|S_{n}|>\Phi^{-1}(1-\alpha/2)$,
where 
\begin{equation}\label{eq:sn}
S_n = \sqrt{n}\frac{ (Z-\tilde{W}\hat{\gamma})^\top (Y-Zg_{0}-\tilde{W}\hat{\pi})}{ \|Z-\tilde{W}\hat{\gamma}\|_{2}   \|Y-Zg_{0}-\tilde{W}\hat{\pi}\|_{2}}.
\end{equation}
Other estimators of the first moment \eqref{eq: moment condition} are certainly possible, however we focus and analyze the natural case above; we leave future efficiency studies for future work since it is not apparent that any other choice is preferred.
Moreover, the self-normalizing statistic above is directly dependent on the hypothesis of interest and is a function of synthesized features. Compared with the existing  approaches where the normalization   factor
is a consistent estimator of the asymptotic variance, our self-normalized approach adopts
an inconsistent estimator as the normalization factor, which in a sense corresponds to
``inefficient Studentizing'' (cf. \cite{RSSB:RSSB737}).
However,  we establish that the asymptotic distribution of the resulting statistic is pivotal and its
percentiles can be obtained from the normal distribution.

In constructing   estimates of $\gamma_{*}$ and $\pi_{*}$, we  do not impose any assumption regarding the sparsity of $\pi_{*}$
or $\beta_{*}$. Notice that, except for the case of sparse $a$,
it is in general unreasonable to expect sparsity in $\pi_{*}$, even
if $\beta_{*}$ is sparse. Although we use estimates
for both $\gamma_{*}$ and $\pi_{*}$ denoted by $\hat{\gamma}$ and
$\hat{\pi}$, respectively, we only require  $l_1$ consistency properties for
$\hat{\gamma}$; in fact, $\hat{\pi}$ only serves to satisfy our
decoupling argument in the proof and does not need to be consistent. We now briefly explain this point.  The constraints imposed in the estimator (\ref{eq: dantzig pi}) guarantee that for the test statistic $S_n$, the term $n^{-1/2}(Z-\tilde{W}\hat{\gamma})^\top (Y-Zg_{0}-\tilde{W}\hat{\pi})$ can be approximated by a product of two independent terms, i.e. $n^{-1/2}(Z-\tilde{W}\gamma_{*})^{\top}(Y-Zg_{0}-\tilde{W}\hat{\pi})$. Then, the only requirement needed is to guarantee that the second term in the last expression does not grow to fast (it does not need to converge to zero) which in turn is provided by the constraints of the optimization problem \eqref{eq: dantzig pi}. 


\subsection{Theoretical properties}
 
 In deriving the theoretical properties of our test, we impose the following assumption.
\begin{assumption}
\label{assu: regularity condition} Let   (i) $x_{i}$  and $\varepsilon_{i} $   have Gaussian distributions, $\mathcal{N}(0,\Sigma_X)$ and $\mathcal{N}(0,\sigma_{\varepsilon}^2)$, respectively. Moreover, 
 assume (ii) that there exist constants $c_{1},c_{2}>0$, such that $\sigma_{\varepsilon}$
and the  eigenvalues of $\Sigma_{X}$ lie in $[c_{1},c_{2}]$. Lastly, let
(iii) there exist   constants $c_{3},c_4\in(0,1)$, such that $\sigma_{u}^{2}/\sigma_{z}^{2}\geq c_{3}$
and $\sigma_{\varepsilon}^{2}/\sigma_{y}^{2}\geq c_{4}$.
\end{assumption} 

Assumption \ref{assu: regularity condition}(i) is only imposed to simplify
the proof. 
In high-dimensional literature Gaussian design is a very common assumption (e.g. \cite{javanmard2014hypothesis,cai2015confidence}). The same results, at the expense of  more complicated proofs,  can be derived for sub-Gaussian designs and errors.   Assumption \ref{assu: regularity condition}(ii) is very standard in high-dimensional literature (see \cite {
bickel2009simultaneous,
ning2014general,
van2014asymptotically} for more details). 

%
 
Assumption \ref{assu: regularity condition}(iii) imposes nondegeneracy of signal-to-noise
ratios for models (\ref{eq: original model}) and (\ref{eq: graphical model}). Since
$\|a\|_{2}$ is allowed to tend to infinity, $\sigma_{z}^{2}=a^\top\Sigma_{X}a/(a^\top a)^{2}$
can tend to zero and thus it is too restrictive to assume that $\sigma_{u}$
is bounded away from zero. 
Hence, Assumption \ref{assu: regularity condition}(iii) is a relaxation, as it only rules
out the uninteresting case of asymptotic noiselessness.
\begin{rem}
The sparsity condition is imposed on neither $a$ nor $\beta_{*}$.
Theorem \ref{thm: unknown variance size X} below says that we can
conduct valid inference of a non-sparse linear combination of a non-sparse
high-dimensional parameter without knowing $\Sigma_{X}$. To the best
of our knowledge, this is the first result that allows for such generality. \end{rem}
\begin{thm}
\label{thm: unknown variance size X}Let Assumption \ref{assu: regularity condition}
hold. Consider estimators \eqref{eq: dantzig pi} and \eqref{eq: dantzig gamma} with suitable choice of
tuning parameters: $\eta,\lambda\asymp\sqrt{n^{-1}\log p}$, $\rho_{0}^{-1}=O(1)$
and $\rho_{0}\leq[1+c_{2}c_{1}^{-1}(c_{3}^{-1}-1)]^{-1/2}$. Suppose
that $\|\gamma_{*}\|_{0}=o(\sqrt{n}/\log p)$. Then, under $H_{0}$
in (\ref{eq: null hypo}), optimization
problems \eqref{eq: dantzig pi} and \eqref{eq: dantzig gamma} are feasible with probability
approaching one and 
\[
\lim_{n,p\to \infty}P\left(|S_{n}|>\Phi^{-1}(1-\alpha/2)\right)=\alpha\qquad\forall\alpha\in(0,1),
\]
where $S_{n}$ is defined in Equation \eqref{eq:sn}. 
\end{thm}

Theorem \ref{thm: unknown variance size X} establishes  that the proposed test is asymptotically exact regardless of how sparse the model  parameter or the loading vector are. In that sense, the result is unique in the existing literature as it covers cases of $\beta$ sparse and $a$ sparse (SS), $\beta$ sparse and $a $ dense (SD) , $\beta$ dense and $a $ sparse (DS) and especially $\beta$ dense and $a$ dense (DD). The (SS) case appears in a number of existing works (see \cite{belloni2014inference,van2014asymptotically,javanmard2014hypothesis,ning2014general}), case (SD) appears in \cite{cai2015confidence}.   Whenever (SS) case holds, our result above matches the above mentioned work see Theorem \ref{thm: power unknown SigmaX}. In the special setting of (SD) our result generalizes the one of \cite{cai2015confidence} as Theorem \ref{thm: unknown variance size X} does not impose any restriction on the size of the loading vector $a$. The last two cases of (DS) and (DD) present an extremely challenging cases in which inference based on estimation (much like Wald or Rao or Likelihood principles) fails due to the inherit limit of detection -- work of \cite{2016arXiv160303474C} provides details of impossibility of estimation in such settings. However, despite these challenges our method is able to provide asymptotically valid inference as we have developed inference based on a specifically designed moment condition (and not a parameter estimation alone).

The result in Theorem \ref{thm: unknown variance size X} is based on the assumption that  \(\hat{\pi}_* \) is a possibly inconsistent estimator of the parameter vector $\pi_*$, i.e. the full model is  dense with all non-zero entries. In the following, we will show that  if  the model is a sparse model, the proposed test \eqref{eq:sn} maintains strong power properties. To facilitate the mathematical derivations, we consider the  local alternatives of the form
\begin{equation}
H_{1,n}:\ a^\top\beta_{*}=g_{0}+n^{-1/2}(a^\top\Omega_Xa)^{1/2}\sigma_{\varepsilon}d,\label{eq: local alternative}
\end{equation}
where $d\in\mathbb{R}$ is a fixed constant.   The following result shows that the proposed test achieves certain optimality in detecting alternatives $H_{1,n}$.

%
%

\begin{thm}
\label{thm: power unknown SigmaX}Consider  $z_{i}$ and $w_{i}$ defined in (\ref{eq: z and w with unknown SigmaX}).
Let Assumption \ref{assu: regularity condition} hold and consider
the choice of tuning parameters, as in Theorem \ref{thm: unknown variance size X}.
Suppose that $\|\gamma_{*}\|_{0}\vee\|\beta_{*}\|_{0}\vee\|a\|_{0}=o(\sqrt{n}/\log p)$.
Then, under $H_{1,n}$ in (\ref{eq: local alternative}), optimization
problems \eqref{eq: dantzig pi} and \eqref{eq: dantzig gamma} are feasible with probability
approaching one and 
\[
\lim_{n,p \to \infty}P\left(|S_{n}|>\Phi^{-1}(1-\alpha/2)\right)=\Psi_{\alpha}(d)\qquad\forall\alpha\in(0,1),
\]
where $\Psi_{\alpha}(d):=\Phi\left(-\Phi^{-1}(1-\alpha/2)+d\right)+\Phi\left(-\Phi^{-1}(1-\alpha/2)-d\right)$.
\end{thm}

To better understand the optimality of the result above, consider 
 the estimator (possibly
infeasible) discussed at the end of Section \ref{sec:Methodology}:
let $\bar{\beta}$ denote an estimator satisfying $\sqrt{n}(\bar{\beta}-\beta_{*})\sim \mathcal N(0,\Omega_X\sigma_{\varepsilon}^{2})$.
Notice that, for the low-dimensional components of $\beta_{*}$, $\bar{\beta}$
achieves semi-parametric efficiency; see \citet{robinson1988root}. Therefore, for sparse $a$, $a^\top\bar{\beta}$ is a semi-parametrically
efficient estimator for $a^\top\beta_{*}$. Notice that $\sqrt{n}(a^\top\bar{\beta}-a^\top\beta_{*})\sim \mathcal N(0,a^\top\Omega_Xa\sigma_{\varepsilon}^{2})$. Based on such efficient estimator, 
one might consider an ``oracle'' test: for a test of nominal size \(\alpha\),   reject the null \(H_0:\ a^\top \beta _*=g_0\) if and only if
\begin{displaymath}
\frac{\sqrt{n}|a^{\top}\bar{\beta}-g_{0}|}{(a^{\top}\Omega_{X}a)^{1/2}\sigma_{\varepsilon}}>\Phi^{-1}(1-\alpha/2).
\end{displaymath} 
It is easy to verify
that the power of this ``oracle'' test of nominal size $\alpha$ against the local alternatives \(H_{1,n}\) (\ref{eq: local alternative}) is  asymptotically equal to  $\Psi_{\alpha}(d)$. Therefore, Theorem \ref{thm: power unknown SigmaX} says that our test asymptotically achieves the same power as the ``oracle'' test  under sparse \(a\) and \(\beta_*\), i.e. it is as efficient as the ``oracle'' test. 

Moreover, in light of recent inferential results in the high-dimensional sparse models, the rate of Theorem 4 can also be shown to be optimal. As existing results apply only to the case of $a=e_j$ for a coordinate vector $e_j$, $1 \leq j \leq p$, we discuss the relations of our work in this specific settings. We note that the tests based on VBRD and BCH are asymptotically equivalent to this ``oracle'' test and hence have the same asymptotic local power; the power of  Wald or Score inferential methods (see Theorem 2.2 in \cite{van2014asymptotically}, Theorem 1 in \cite{belloni2014inference} or Theorem 4.7 in \cite{ning2014general}) and that of \cite{javanmard2014hypothesis} (see Theorem 2.3 therein) is asymptotically
equal to and converges to	$\Psi_\alpha(d)$, respectively. This in turn, implies that the proposed method is  semi-
parametrically efficient and asymptotically minimax.
 For vectors $a$ that have more than one non-zero coordinate, we can only compare our work with that of \cite{cai2015confidence}, where we observe that the result of Theorems 1 and 3 therein   matches those of   Theorem 4  covering the case of extremely sparse beta and potentially dense vectors $a$. However, observe that the confidence intervals developed therein require specific knowledge of the sparsity of the parameter $\beta_*$, $\| \beta_*\|_0$, a quantity rarely known in practice. 
Unlike their method, our method can be directly implemented without the knowledge of the sparsity of $\beta_*$ and yet achieves the same optimality guarantees.

\section{Applications to non-sparse  high-dimensional models} \label{sec:examples}

This section is devoted to three concrete applications of the general  methodological  results
developed in Sections \ref{sec:Methodology} and  \ref{sec: extension} -- hence, showcasing the  wide impact of the developed theories.

\subsection{Testing pairwise homogeneity}

The previous section deals with situations in which  each coordinate of the parameters is allowed to vary  independently and  any subset of the coordinates can be non-zero simultaneously. This condition will not be satisfied if we are interested in testing pairwise homogeneity in the linear model (group effect), that is, if we are interested in testing the hypothesis   
 \[
 H_0: \beta_{*,k} = \beta_{*,j}
 \]
 for $k,j \in \{1,2,\dots, p\}$   while also allowing $\beta$ to be a dense and high-dimensional vector. To the best of our knowledge, such tests were not designed in the existing literature.
 The proposed methodology   easily extends to this case, where the loading vector $a$ takes the form 
 $a = (0,\dots, 0,1, 0,\dots,0, -1 , 0 , \dots, 0)^\top $, with the location of the $1$'s at the $j$-th and $k$-th coordinate, respectively.
  Without loss of generality, we assume that \(k=1\) and \(j=2\). Then it is not hard to show that  \( z_i=(x_{i,1}-x_{i,2})/2\) and \(\tilde{w}_i=((x_{i,1}+x_{i,2})/\sqrt{2},x_{i,3},\cdots,x_{i,p})^\top \in \mathbb{R}^{p-1}\). 
 The proposed methodology for this problem simplifies, then, to finding $\hat \pi$ and $\hat \rho$ that satisfy 
 \begin{equation}\label{eq: dantzig pi1}
\begin{array}{cclcl}
(\hat{\pi},\hat{\rho}) &= & \underset{(\pi,\rho)\in\mathbb{R}^{p-1}\times \mathbb{R}_+}{\arg\min}\|\pi\|_{1}& &\\
    & s.t  \qquad & \tilde{W} =  [(X_1+X_2)/\sqrt{2},X_3,\cdots,X_p] &&\\
       &  &\|\tilde{W}^\top (Y-\tilde{W}\pi)\|_\infty \leq \eta \rho \sqrt{n} \|Y\|_2 && \\

    & & Y^\top \left(Y-\tilde{W}\pi \right) \geq  \rho_{0}\ \rho \  \|Y\|_2^{2}/2 && \\
    & &   \rho  \in [\rho_0,1] && \\
\end{array}
\end{equation}
and $\hat \gamma$ that satisfies
  \begin{equation}\label{eq: dantzig gamma}
\begin{array}{cclcl}
\hat{\gamma}&= &\underset{\gamma \in \mathbb{R}^{p-1}}{\arg\min} \ \|\gamma\|_{1}&& \\ 
  & s.t  \qquad  &\tilde{W} =  [(X_1+X_2)/\sqrt{2},X_3,\cdots,X_p] & & \\
    & & \|\tilde{W}^\top (X_1-X_2-2\tilde{W}\gamma)\|_\infty \leq \lambda \sqrt{n} \|X_1-X_2\|_2 & &,
\end{array}
\end{equation} 
  for $\lambda,\eta \asymp \sqrt{n^{-1}\log p}$ .

Consequently, we reject $H_{0}: \beta_{*,1}\ = \beta_{*,2}$  if $|S_{n}|>\Phi^{-1}(1-\alpha/2)$,
where 
\begin{equation}\label{eq:sna}
S_n = \sqrt{n}\frac{ \left(X_{1}- X_2-2\tilde{W}\hat{\gamma} \right)^\top  \left(Y-\tilde{W}\hat{\pi} \right) }{ \left\|X_{1} - X_2-2\tilde{W}\hat{\gamma}\right\|_{2}   \left\|Y -\tilde{W}\hat{\pi}\right\|_{2}}.
\end{equation}

\subsection{Inference of conditional mean}\label{subsec: conditional mean}

Our methodology can also be used for the inference regarding the average value of the response i.e. regarding the conditional mean of the regression model. Suppose that the object of interest is \(E(y_i \mid \zeta_i)\), where \(y_i \in \mathbb{R}\) and \(\zeta_i \in \mathbb{R}^{k}\). For a given value \(d\in \mathbb R ^{k}\) and $g_0 \in \mathbb{R}$, the focus is to test \begin{displaymath}
H_0:\ E(y_i \mid \zeta_i=d)=g_0.
\end{displaymath}

Assuming that for some given dictionary of transformations of $\{\phi_j (\cdot)\}_{j=1}^p$, the conditional mean function admits the representation: $E(y_i \mid \zeta_i)=\sum_{j=1}^p \beta_{*,j} \phi_j (\zeta_i)$ for some  vector \(\beta_*=(\beta_{*,1},\cdots,\beta_{*,k})^\top \in \mathbb R ^p\). Then the conditional mean model can be written as 
\begin{equation} \label{eq: representation}
y_i=x_i ^\top \beta_* +\varepsilon_i,
\end{equation}
where $x_i=(\phi_1(\zeta_i),\cdots,\phi_p(\zeta_i))^\top \in \mathbb R ^p$  and $E(\varepsilon_i\mid x_i)=0$. In turn, the confidence intervals for the regression mean can be designed simply by inverting the test statistics 
\[
S_n = \sqrt{n}\frac{ (Z-\tilde{W}\hat{\gamma})^\top (Y-Zg_{0}-\tilde{W}\hat{\pi})}{ \|Z-\tilde{W}\hat{\gamma}\|_{2}   \|Y-Zg_{0}-\tilde{W}\hat{\pi}\|_{2}}
\]
designed for the 
 inference problem  \begin{displaymath}
H_0:\ a^\top \beta_*=g_0,
\end{displaymath} 
where $a=(\phi_1 (d),\cdots,\phi_p(d))^\top \in \mathbb R ^p$ and  $U_{a}U_{a}^\top=\left(I_{p}-aa^\top/\sum_{j=1}^p \phi_j^2(d)\right)$ with
$$z_{i}=  \frac{\sum_{j=1}^p \phi_j(d) \phi_j(\zeta_i)}{ \sum_{j=1}^p \phi_j^2(d)} , \qquad \mbox{and} \qquad  \tilde w_{ij}= \sum_{l=1}^p \{ U_a \}_{lj}  \phi_l(\zeta_i),  \qquad 1\leq j\leq p-1. $$
Notice that we do not assume that the vector \(\beta_*\) is sparse and we allow for $p \gg n$. Therefore, representing the conditional mean function in terms of a large number of transformations of \(\zeta_i\), while simultaneously allowing all to be non-zero, does not lose much in generality.  Additionally, it is worth mentioning that inference for such models has not been addressed in the existing literature: most of the existing work is strictly focused around sparse or sparse additive models. With the general model considered here, one can consider tests regarding treatment effects (when viewed as the conditional mean) and allow for fully dense models and loading vectors, i.e.  the treatment being a dense combination of many variables. Existing work, such as \cite{belloni2014inference}, only allows the treatment to be a single variable.  



\subsection{Decomposition of conditional mean}
 
In practice, the researcher might be interested in how much a certain group of features contribute to the conditional mean. Let $\mathcal{G}\subseteq\{1,...,p\}$. The goal is to conduct inference on linear functionals of $\{\beta_{*,j}\}_{j\in\mathcal{G}}$, i.e., $\sum_{j\in\mathcal{G}} c_{j}\beta_{*,j}$ for some known $\{c_{j}\}_{j\in\mathcal{G}}$.

For example, consider  the notations from Section \ref{subsec: conditional mean}. Let $\zeta_{i}=(\zeta_{i,1},...,\zeta_{i,k})^{\top }$ and suppose that one is interested in the impact of $\zeta_{i,1}$ on the conditional mean for $\zeta=d$. This is equivalent to quantifying $\sum_{j\in \mathcal{G}_1} \phi_{j}(d)\beta_{*,j} $, where the set contains all the indexes $j$ such that the first entry of $\zeta_{i}$ has non-zero effect on $\phi_{j}(\zeta_{i})$, i.e., $\mathcal{G}_{1}=\{j:  \phi_{j}(\zeta) {\rm \ is\ not\ constant\ in}\ \zeta_{1}  \}$. If $\phi_{j}(\cdot)$'s are transformations of individual entries of $\{\zeta_{i,j}\}_{j=1}^{k}$, then $\mathcal{G}_{1}$ corresponds to transformations of $\zeta_{i,1}$. For another example, suppose that all the $p$ features are genes. The domain scientist (biologist, doctor, geneticist, etc) might be interested in  how much a group of genes contributes to the expected value of the response variable.

Without loss of generality, we assume that $\mathcal{G}=\{1,...,H\}$ and $c=(c_{1},...,c_{H})^\top \in \mathbb{R}^H$. Let $U_{c}\in\mathbb{R}^{H\times(H-1)}$ satisfy $I_{H}-cc^{\top}/(c^\top c)=U_{c}U_{c}^{\top}$ and $U_{c}^\top U_{c}=I_{H-1}$. Then the synthesized features can be constructed by 
$
z_{i}=\|c\|_{2}^{-2}\sum_{j=1}^{H}c_{j}x_{i,j}\quad {\rm and}\quad \tilde{w}_{i}= \left( \sum_{l=1}^{H} (U_{c})_{l,1}x_{i,l},
\cdots,
\sum_{l=1}^{H} (U_{c})_{l,H-1}x_{i,l},
x_{i,H},
\cdots,
x_{i,p}
\right)^\top\in\mathbb{R}^{p-1},
$
where $(U_{c})_{l,j}$ denotes the $(l,j)$ entry of the matrix $U_{c}$.
For example, whenever $H=3$ and $c_{j}=1$ for all $j=1,2,3$, then 
 \[
U_{c}=\begin{pmatrix}-\sqrt{3/2} &-1/\sqrt{2}  \\
0 & \sqrt{2} \\
\sqrt{3/2} &-1/\sqrt{2}
\end{pmatrix}
\]
and the procedure for testing $\beta_{*,1}+\beta_{*,2}+\beta_{*,3}=g_0$ would be as follows. We define 

 \begin{equation}\label{eq: dantzig pi2}
 \begin{array}{cclcl}
 (\hat{\pi},\hat{\rho}) &= & \underset{(\pi,\rho)\in\mathbb{R}^{p-1}\times \mathbb{R}_+}{\arg\min}\|\pi\|_{1}& &\\
 & s.t  \qquad & \tilde{W} = \left[ \sqrt{\frac{3}{2}}(X_3-X_1),\  -\frac{1}{\sqrt{2}}(X_1 -2 X_2 +X_3),\ X_4,\cdots,X_p \right] &&\\
 &  &\|\tilde{W}^\top [Y-(X_1 +X_2 + X_3)g_{0}/3-\tilde{W}\pi]\|_\infty \leq \eta \rho \sqrt{n} \|Y-(X_1 +X_2 + X_3)g_{0}/3\|_2 && \\
 
 & & (Y-(X_1 +X_2 + X_3)g_{0}/3)^\top \left(Y-(X_1 +X_2 + X_3)g_{0}/3-\tilde{W}\pi \right) &&\\
 & & \qquad \geq  \rho_{0}\ \rho \  \|Y-(X_1 +X_2 + X_3)g_{0}/3\|_2^{2}/2 && \\
 & &   \rho  \in [\rho_0,1] && \\
 \end{array}
 \end{equation}
 and $\hat \gamma$ that satisfies
 \begin{equation}\label{eq: dantzig gamma2}
 \begin{array}{cclcl}
 \hat{\gamma}&= &\underset{\gamma \in \mathbb{R}^{p-1}}{\arg\min} \ \|\gamma\|_{1}&& \\ 
 & s.t  \qquad  & \tilde{W}= \left[ \sqrt{\frac{3}{2}}(X_3-X_1),\  -\frac{1}{\sqrt{2}}(X_1 -2 X_2 +X_3),\ X_4,\cdots,X_p \right] & & \\
 & & \|\tilde{W}^\top \left((X_1 +X_2 + X_3)g_{0}-3\tilde{W}\gamma\right)\|_\infty \leq \lambda \sqrt{n} g_{0}\|X_1+X_2+X_3\|_2 & &,
 \end{array}
 \end{equation} 
 for $\lambda,\eta \asymp \sqrt{n^{-1}\log p}$ .

For a test of nominal size $\alpha$, we reject $H_{0}: \beta_{*,1}+\beta_{*,2}+\beta_{*,3}=g_0$  if $|S_{n}|>\Phi^{-1}(1-\alpha/2)$,
where 
\begin{equation}\label{eq:sna1}
S_n = \sqrt{n}\frac{ \left((X_1 +X_2 + X_3)g_{0}-3\tilde{W}\hat{\gamma} \right)^\top  \left(Y-(X_1 +X_2 + X_3)g_{0}/3-\tilde{W}\hat{\pi} \right) }{ \left\|(X_1 +X_2 + X_3)g_{0}-3\tilde{W}\hat{\gamma}\right\|_{2}   \left\|Y-(X_1 +X_2 + X_3)g_{0}/3-\tilde{W}\hat{\pi} \right\|_{2}}.
\end{equation}

\section{\label{sec:Numerical-results}Numerical results}
 
 In this section we study the finite sample performance of the proposed methodology for both known $\Sigma_X$ and unknown $\Sigma_X$. We explicitly consider dense loadings $a$ and dense parameter vectors $\beta_*$ as well as more common sparse settings.

\subsection{Monte Carlo experiments} \label{sec:MC}

Consider the model (\ref{eq: original model}) with  the model error following standard normal distribution. In all the simulations, we set $n=100$ and $p=500$ and the nominal size of all the tests is 5\%. The rejection
probabilities are based on 500 repetitions.  The null hypothesis we test is $H_{0}:\ a^\top \beta_{*}=g_{0}$, where $g_{0}=a^\top \beta_{*}+h$ and $h$ is allowed to vary in order to capture both Type I and Type II error rates.

\subsubsection{Setup}
We consider  in total four regimes on the structure of the model and the null hypothesis -- sparse and dense regimes for $\beta_{*}$ as well as sparse and dense  regimes for  the loading vector $a$.
\begin{itemize}
	\item[(i)] In the  {\it Sparse parameter regime} we consider the parameter   structure with $\beta_{*}=(0.8,0.8,0,...,0)^{\top}$.
	\item[(ii)]  In the  {\it Dense parameter regime} we consider   the parameter   structure with  $\beta_{*}=\frac{3}{\sqrt{p}}(1,1,...,1)^{\top}$.
	\item[(iii)]  In the {\it Sparse loading regime} we consider the loading vector  $a=(0,1,0,...,0)^{\top}$.
		\item[(iv)] In the  {\it Dense loading regime} we consider the loading vector   $a=(1,1,...,1)^{\top}$.
\end{itemize}
 Observe that (iii) is an extreme sparse-loading case. We consider this special case in order to compare existing inferential methods, like VBRD and BCH. However, our method can be implement for various number of non-zero elements, whereas the existing one cannot. 

We present results for  three different designs settings including sparse, dense, Gaussian and non-Gaussian settings.

\begin{itemize}
	\item[] {\it Example 1.} Here we consider the standard Toeplitz design where the rows of $X$ are drawn as an i.i.d random draws from a multivariate Gaussian distribution $\mathcal{N}(0,\Sigma_X) $, with covariance matrix $(\Sigma_{X})_{i,j}=0.4^{|i-j|}$.
	\item[] {\it Example 2.}  In this case we consider a non-sparse design matrix with equal correlations among the features. Namely,  rows of $X$ are  i.i.d draws  from  the multivariate Gaussian distribution $\mathcal{N}(0,\Sigma_X) $, where  $(\Sigma_{X})_{i,j}$ is 1 for $i=j$ and is $0.4$ for $i\neq j$. Observe that this case is particularly hard for most inferential methods as all features are interdependent and $\Omega_{X}$ is not sparse.

	\item[] {\it Example 3.} In this example we consider a highly non-Gaussian design that also has strong dependence structure. We consider the setting of \cite{fan2010sure}. We repeat the details here for the convenience of the reader. Let $x$ be a typical row of $X$. For $j \in \{1,...,15\}$, $x_{j}=(\xi+c \xi_{j})/\sqrt{1+c^2}$, where $\xi$ and $\{\xi_{j}\}_{j=1}^{15}$ are i.i.d $\mathcal{N}(0,1)$ and $c$ is chosen such that ${\rm corr}(x_{1},x_{2})=0.4$. For $j \in \{16,...,[p/3]\}$, $x_{j}$ is i.i.d $\mathcal{N}(0,1)$. For $j\in \{[p/3]+1,...,[2p/3]\}$, $x_{j}$ is i.i.d from a  double exponential distributions with location parameter zero and scale parameter one. For $j\in \{[2p/3]+1,...,p\}$, $x_{j}$ is i.i.d from the half-half mixture of $\mathcal{N}(-1,1)$ and $\mathcal{N}(1,0.5)$. Observe that in this case $2/3$ of the features follow non-Gaussian distributions. Thus,  in this  case it  is extremely difficult to even obtain consistent estimation of the model parameters. 
\end{itemize}

\subsubsection{Implementation details}
We compare the proposed tests with VBRD and BCH;  methods proposed in \cite{cai2015confidence} contain constants whose values could be very conservative in finite samples. Our tests with known and unknown $\Sigma_{X}$ are implemented as discussed in Sections 
\ref{sec:Methodology} and \ref{sec: extension}, respectively. 

The VBRD method is implemented for both dense and sparse loadings as follows. We first compute  the debiased estimator $\hat{\beta}_{\rm debias}$ and the nodewise Lasso estimator $\hat{\Omega}_{\rm Lasso}$ for the precision matrix $\Sigma_{X}$ as in VBRD. Then  test is to reject $H_{0}$ if and only if $$\sqrt{n}|a^{\top}\hat{\beta}_{\rm debias}-g_{0}|/\sqrt{a^{\top}\hat{\Omega}_{\rm Lasso}\hat{\Sigma}_{X}\hat{\Omega}_{\rm Lasso}^{\top} a \sigma_{\varepsilon}^2}>\Phi^{-1}(1-0.05/2).$$

The BCH method is only implemented for the sparse loadings. We compute the generic post-double-selection estimator for the second entry of $\beta$ as in Equation (2.8) of BCH and compute the standard error as in Theorem 2 therein. Then a usual t-test is conducted. It is not clear how BCH can be extended to handle any loading vector  $a$ different from an extremely sparse case (see (iii) above): first, for any other loading structure it is not defined how to gather selected features of what would be multiple simultaneous equations; second, naively extending the original BCH to the problem of dense $a$ ($\|a\|_0=p$) means running an OLS regression of the response against all the features, which is not feasible for $p>n$. 

\subsubsection{Results}

 We start with the size properties of competing tests. For this purpose, we examine the distributions of the test statistics under the null hypothesis by  comparing empirical distributions of the tests with the theoretical benchmark of standard normal random variable. For simplicity of presentation, we only consider the Toeplitz design.  For the testing problem with sparse $\beta_*$ and sparse $a$, our tests, VBRD and BCH exihibit the validity guaranteed by the theory; in Figure \ref{fig: sparse_a_sparse_beta}, the histograms of the test statistics are close to $\mathcal{N}(0,1)$ with large p-values of the Kolmogorov-Smirnov (KS) tests. For all the other problems, our tests outperform existing methods. As shown in Figure \ref{fig: dense_a_sparse_beta}, the histogram of VBRD test visually is still close to the standard normal distribution but the KS test suggests discernible discrepancies between the two distributions. In Figure \ref{fig: sparse_a_dense_beta}, 
we see that lack of sparsity in $\beta_{*}$ causes serious problems in Type I error for both VBRD and BCH. Inference under dense $\beta_*$ and dense $a$ turns out to be the most challenging problem for existing methods; in Figure \ref{fig: dense_a_dense_beta}, we see quite noticeable difference between the histogram of VBRD test and $\mathcal{N}(0,1)$. In contrast, the distribution of the test statistics of the proposed methods closely match $\mathcal{N}(0,1)$ in all the scenarios, as established in  Theorems \ref{thm: known variance X} and \ref{thm: unknown variance size X}. 
The Type I errors, reported in Table \ref{tab:Size1}, confirm the above findings: existing methods can suffer greatly from lack of sparsity in $\beta_*$ and/or $a$ in terms of validity -- observed Type I error of BCH or VBRD can easily reach 40\%.


\begin{figure}
	\caption{\label{fig: sparse_a_sparse_beta} Distribution of the test statistics  under the null hypothesis $H_0:   \beta_{*,2}=0.8$ (in blue) and the standard normal distribution $\mathcal{N}(0,1)$ (in red) with $n=100$ and $p=500$. In this example we consider sparse $\beta$ and sparse $a$ setting and compare the distribution under the null of our tests (with and without known variance) in the top row and two competing methods VBRD and BCH in the bottom row. We report p-values of the Kolmogorov-Smirnov test statistics in the subtitles.}

	\begin{centering}
		\includegraphics[scale=0.5]{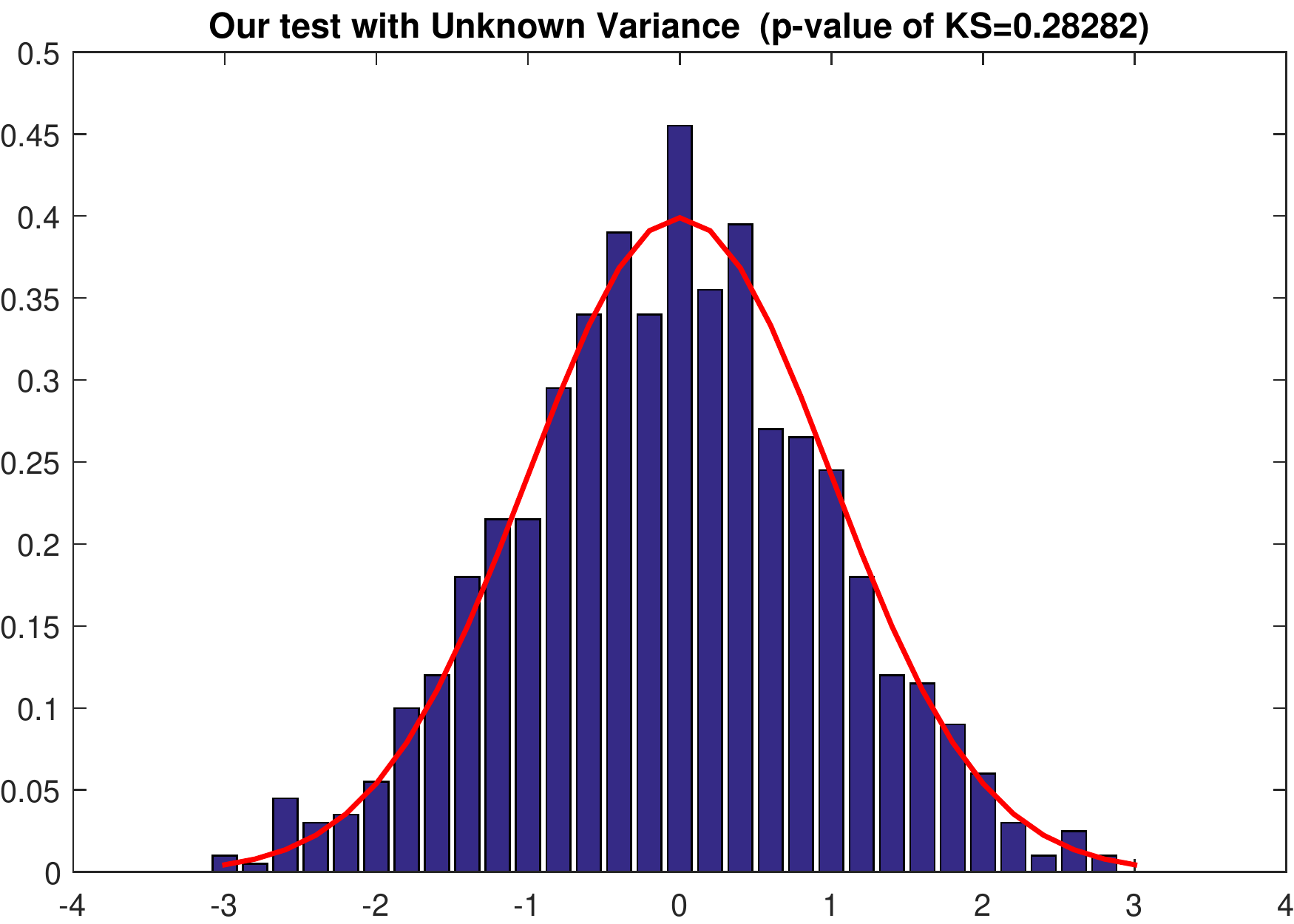}\includegraphics[scale=0.5]{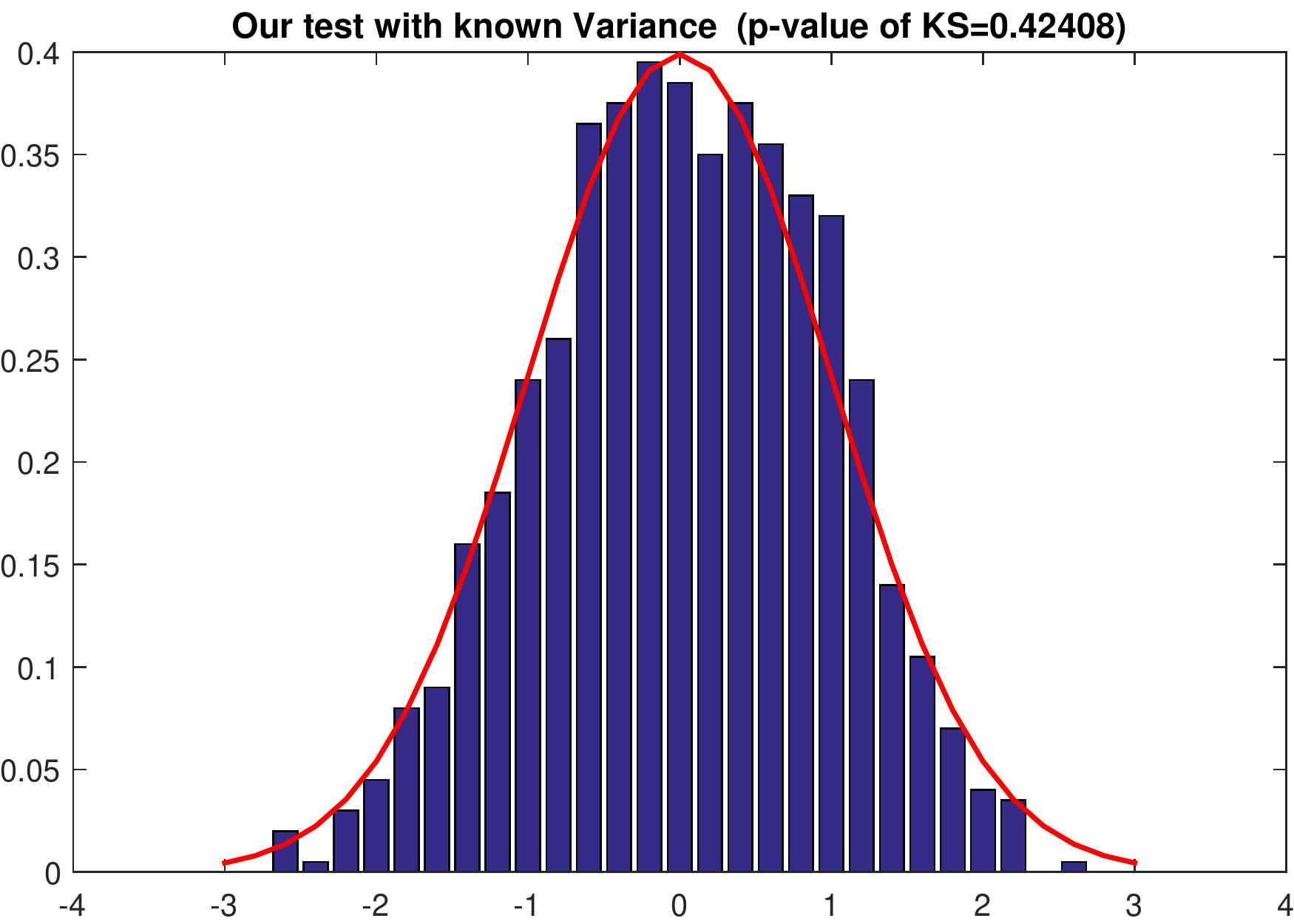}
		\par\end{centering}
 
	\centering{}\includegraphics[scale=0.5]{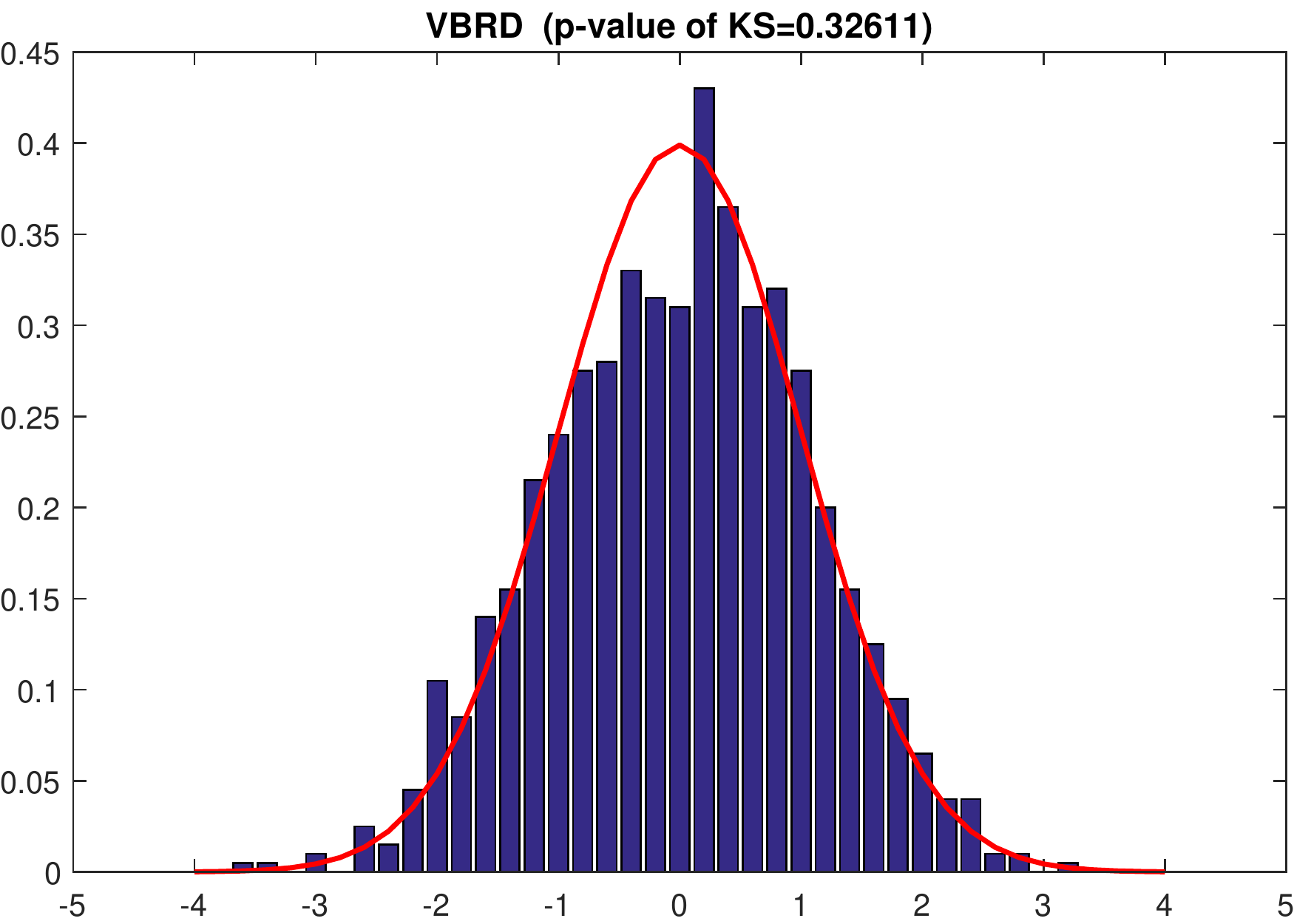}\includegraphics[scale=0.5]{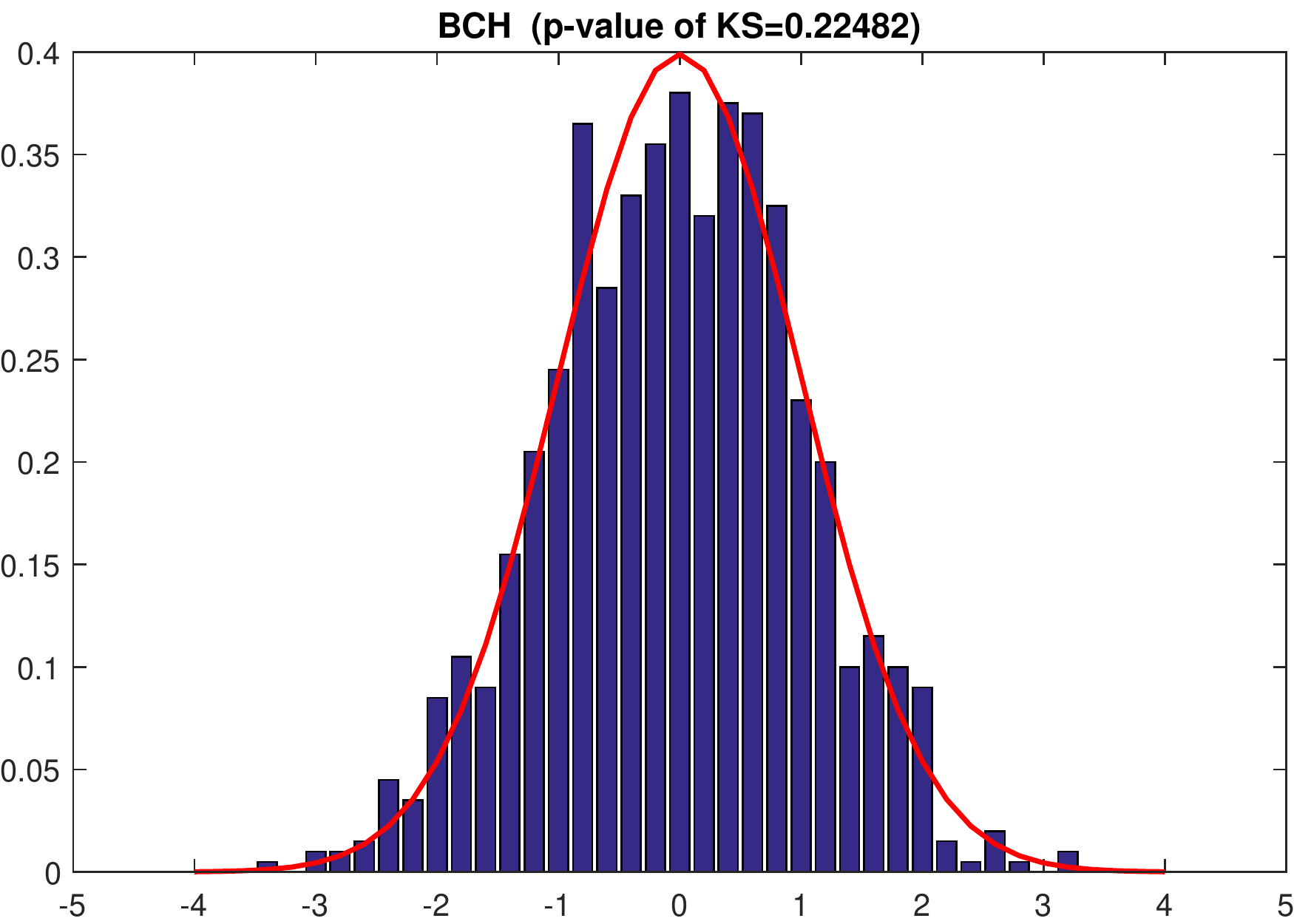}
	\floatfoot{\small   Note that tuning parameters for all the methods are chosen according to their ``oracle'' theoretical values. Error and design are normally distributed with Toeplitz correlation structure with $\rho=0.4$. The histograms are computed based on   $500$ simulation runs.}
\end{figure}

\begin{figure}
	\caption{\label{fig: dense_a_sparse_beta}
	Distribution of the test statistics  under the null hypothesis $H_0: \sum_{j=1}^p a_j \beta_{*,j}=1.6$ (in blue) and the standard normal distribution $\mathcal{N}(0,1)$ (in red) with $n=100$ and $p=500$. In this example we consider sparse $\beta$ and dense $a$ setting and compare the distribution under the null of our tests (with and without known variance) in the top row and two competing methods VBRD and BCH in the bottom row. We report p-values of the Kolmogorov-Smirnov test statistics in the subtitles. }

	\bigskip{}

	\begin{centering}
		\includegraphics[scale=0.5]{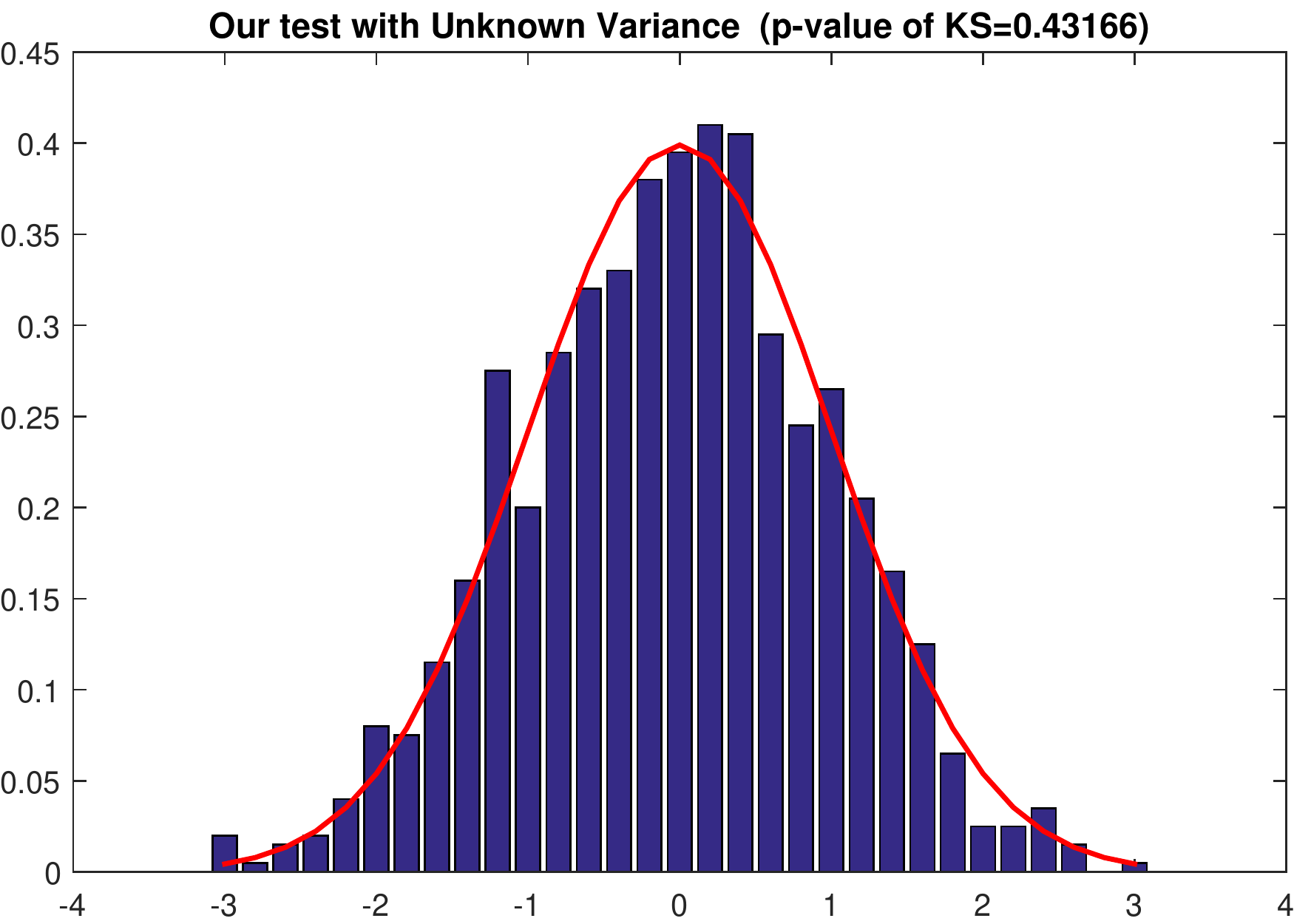}\includegraphics[scale=0.5]{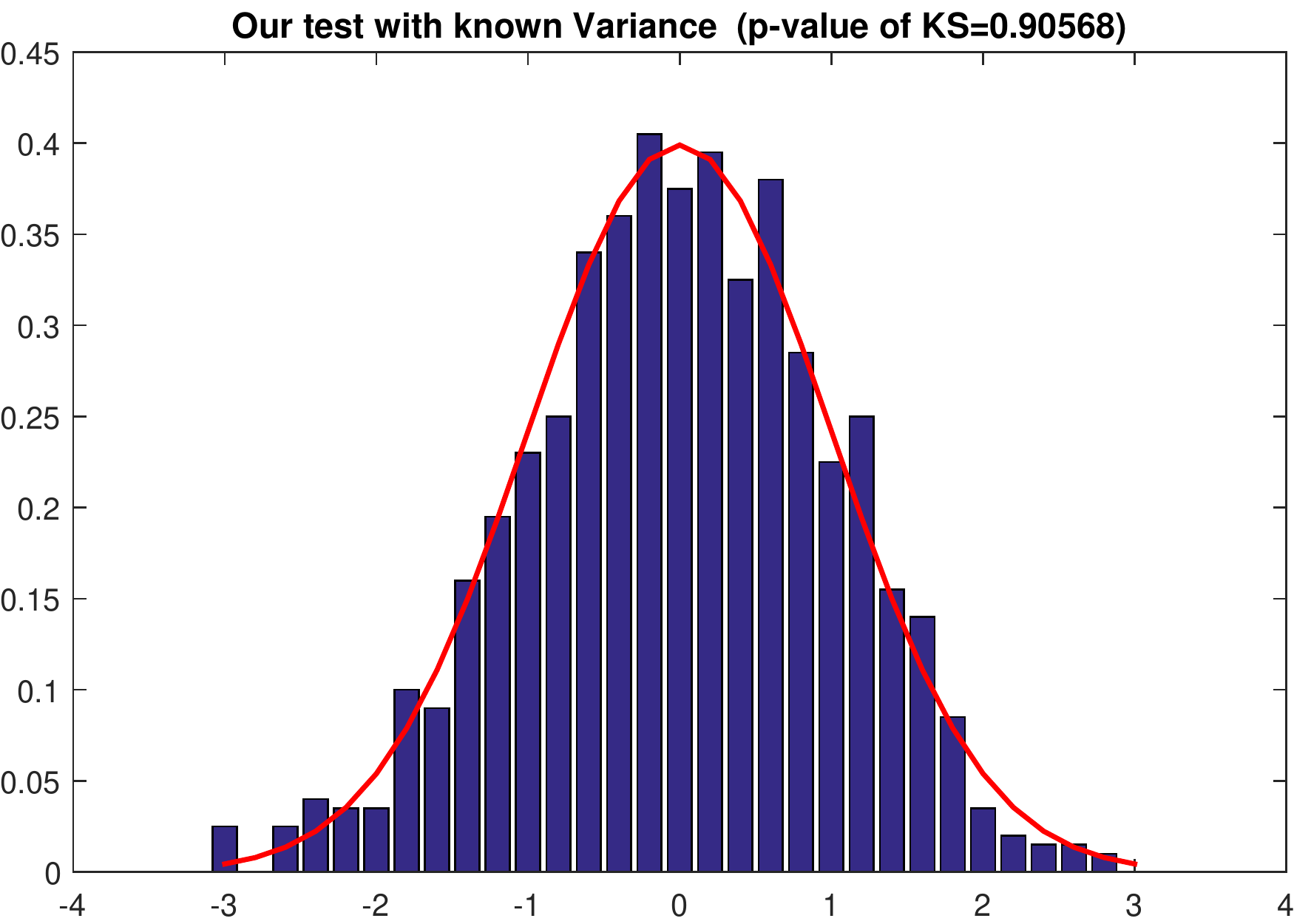}
		\par\end{centering}
	
	\bigskip{}

	\centering{}\includegraphics[scale=0.5]{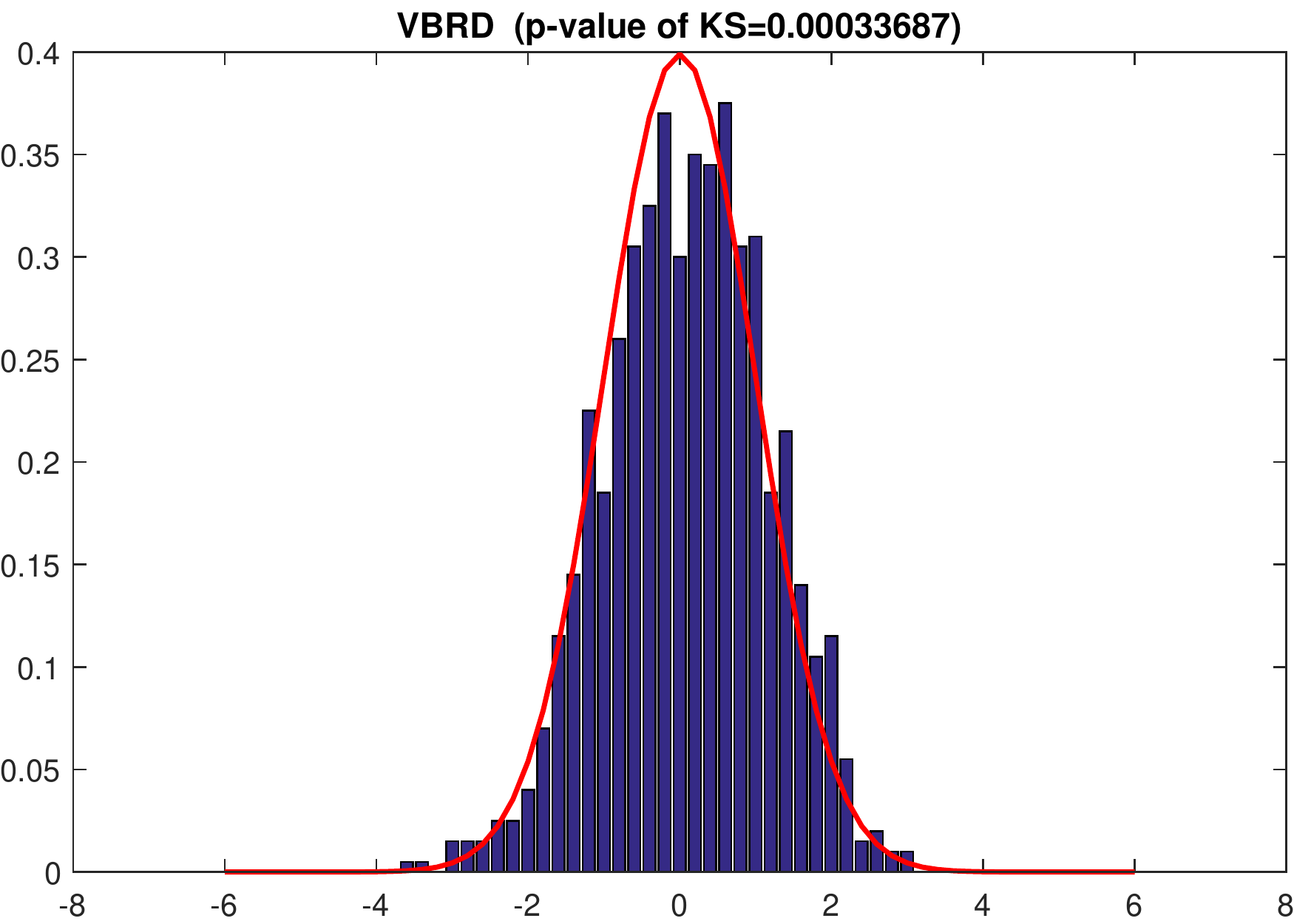}
	\floatfoot{\small   Note that tuning parameters for all the methods are chosen according to their ``oracle'' theoretical values. Error and design are normally distributed with Toeplitz correlation structure with $\rho=0.4$. The histograms are computed based on   $500$ simulation runs.}
\end{figure}

\begin{figure}
	\caption{\label{fig: sparse_a_dense_beta}
	Distribution of the test statistics  under the null hypothesis $H_0:   \beta_{*,2}=3/\sqrt{p}$ (in blue) and the standard normal distribution $\mathcal{N}(0,1)$ (in red) with $n=100$ and $p=500$. In this example we consider dense $\beta$ and sparse $a$ setting and compare the distribution under the null of our tests (with and without known variance) in the top row and two competing methods VBRD and BCH in the bottom row. We report p-values of the Kolmogorov-Smirnov test statistics in the subtitles. }

	\bigskip{}

	\begin{centering}
		\includegraphics[scale=0.5]{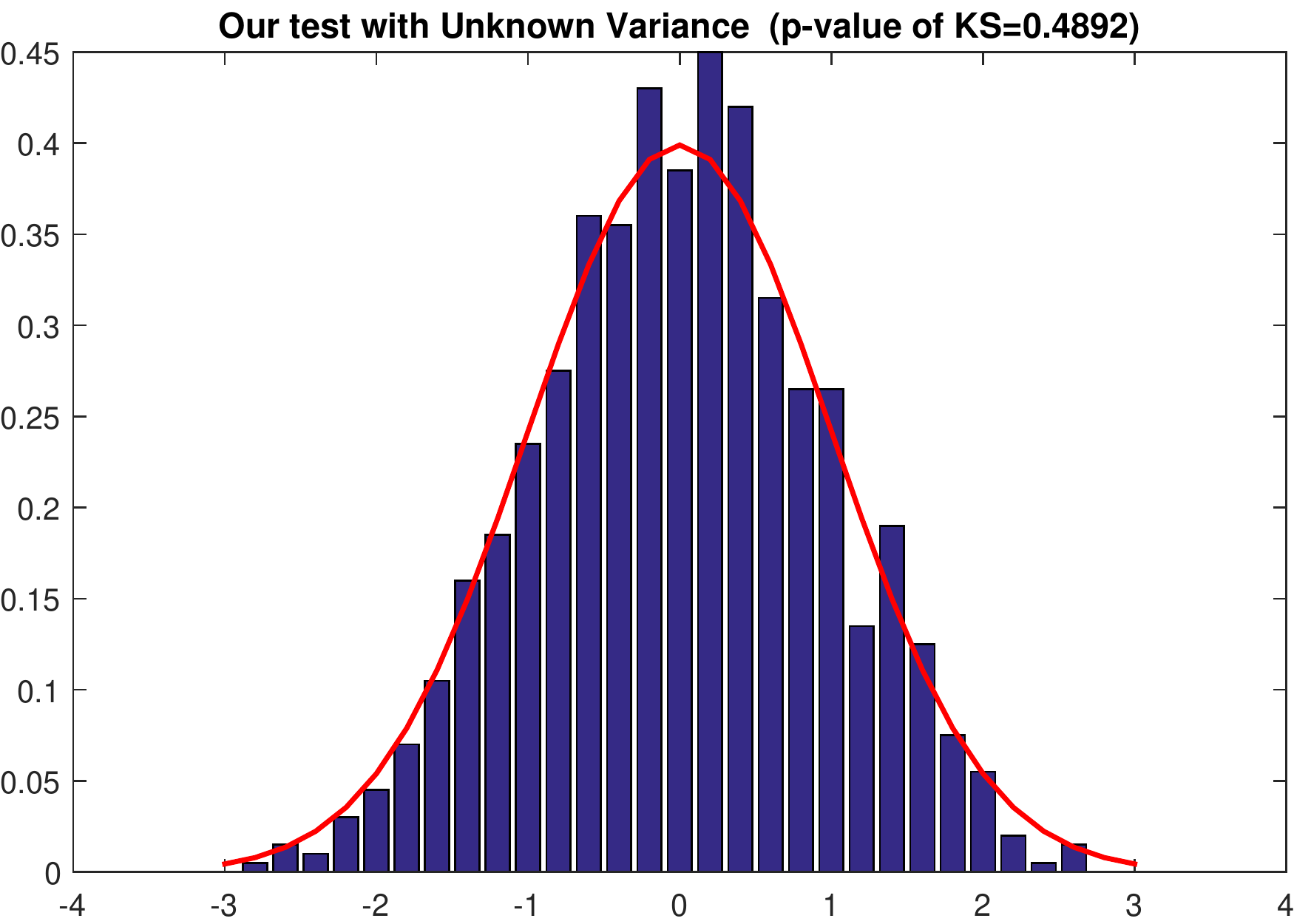}\includegraphics[scale=0.5]{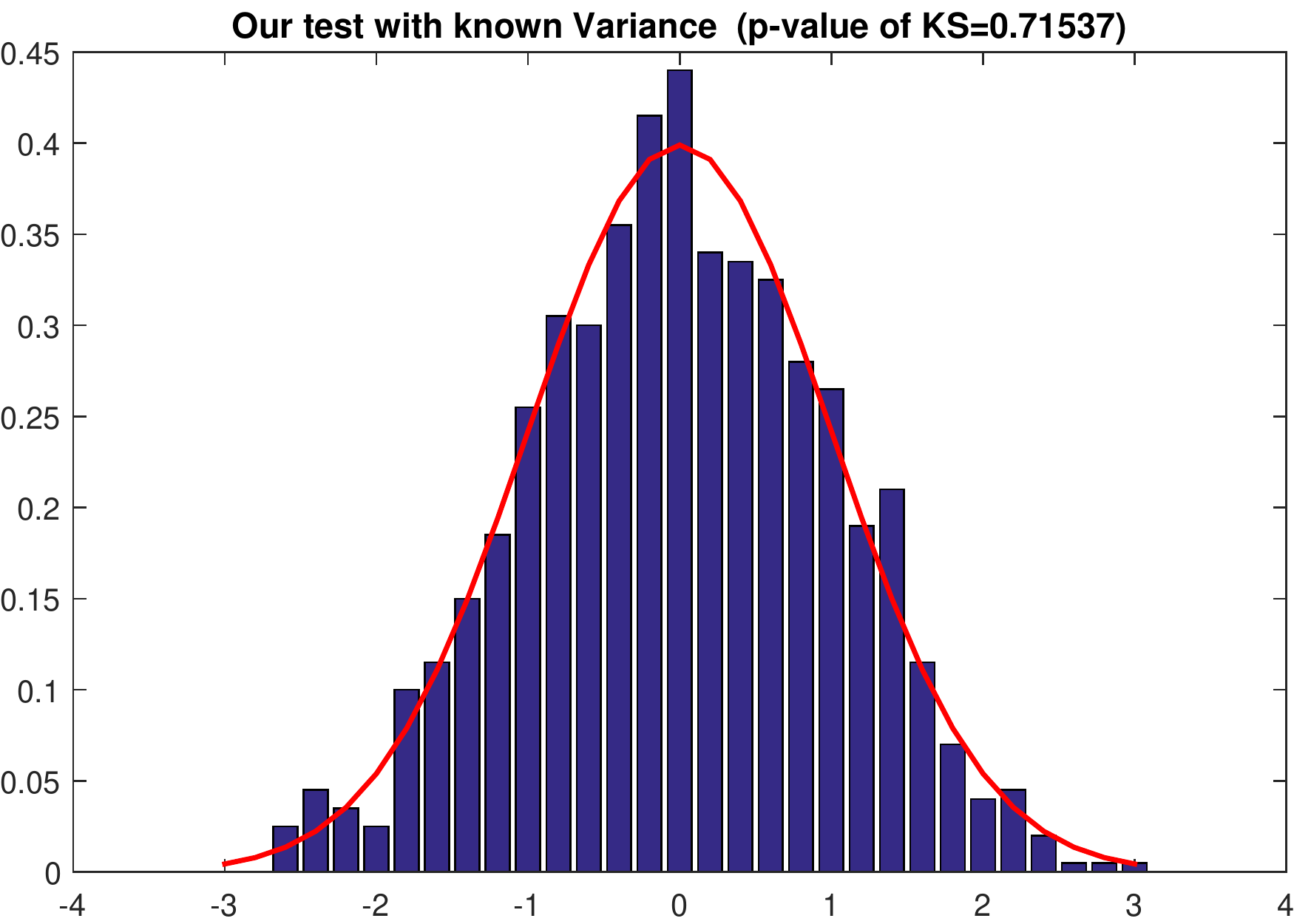}
		\par\end{centering}
	
	\bigskip{}

	\centering{}\includegraphics[scale=0.5]{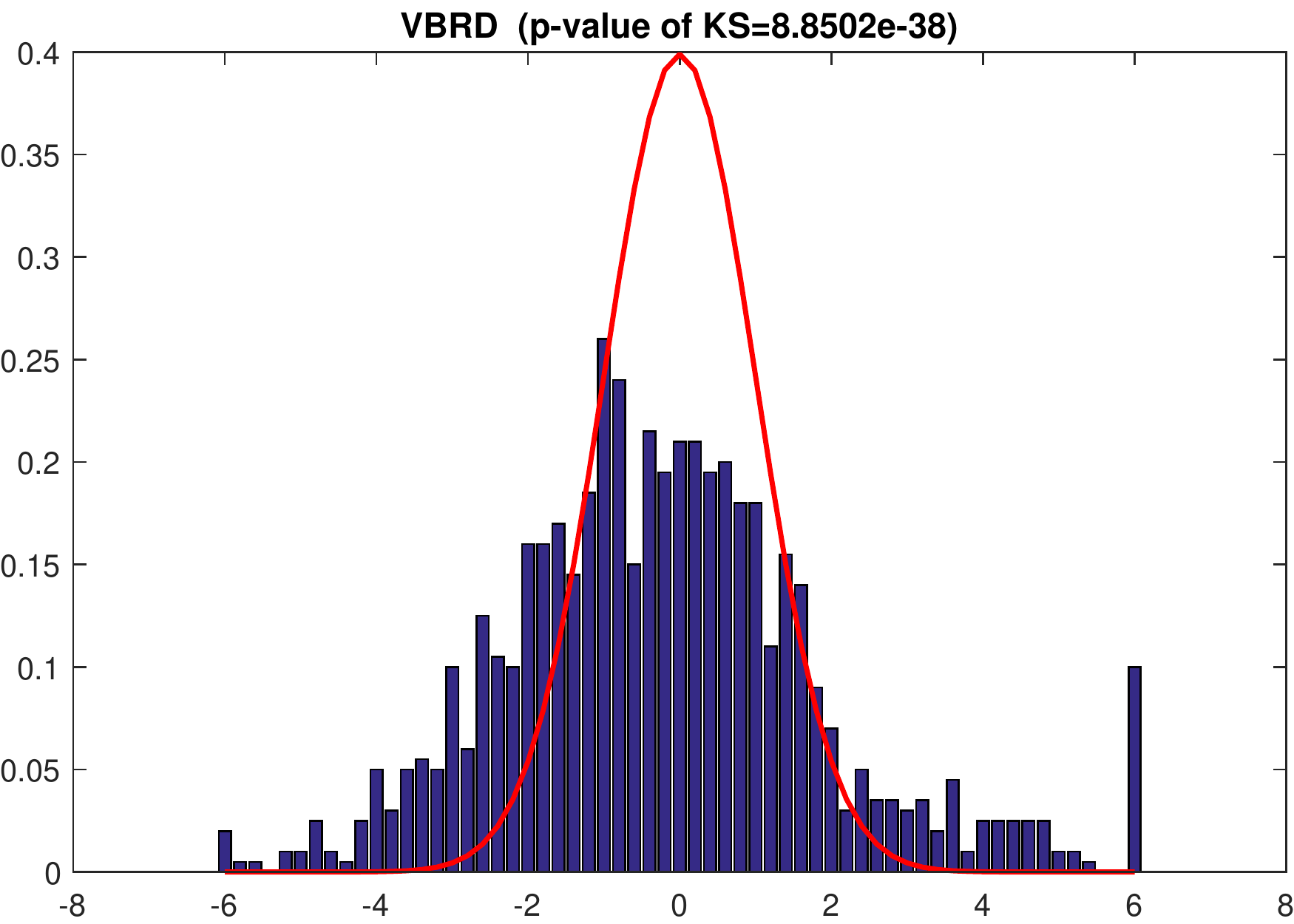}\includegraphics[scale=0.5]{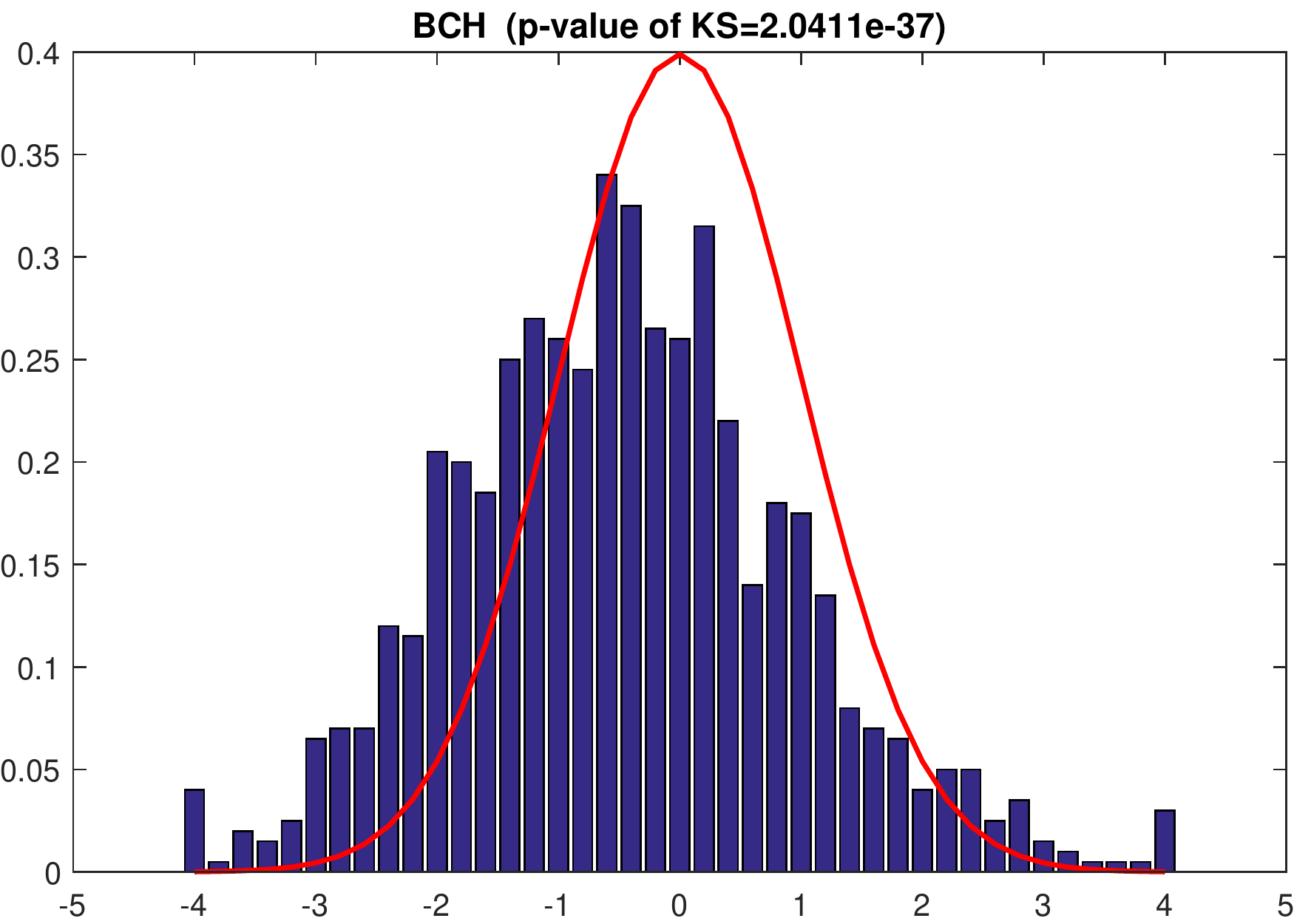}
		\floatfoot{\small   Note that tuning parameters for all the methods are chosen according to their ``oracle'' theoretical values. Error and design are normally distributed with Toeplitz correlation structure with $\rho=0.4$. The histograms are computed based on   $500$ simulation runs.}
\end{figure}

\begin{figure}
	\caption{\label{fig: dense_a_dense_beta}
	Distribution of the test statistics  under the null hypothesis $H_0:   \sum_{j=1}^p \beta_{*,j}=3 \sqrt{p}$ (in blue) and the standard normal distribution $\mathcal{N}(0,1)$ (in red) with $n=100$ and $p=500$. In this example we consider dense $\beta$ and dense $a$ setting and compare the distribution under the null of our tests (with and without known variance) in the top row  and two competing methods VBRD and BCH in the bottom row. We report p-values of the Kolmogorov-Smirnov test statistics in the subtitles.}

	\bigskip{}

	\begin{centering}
		\includegraphics[scale=0.5]{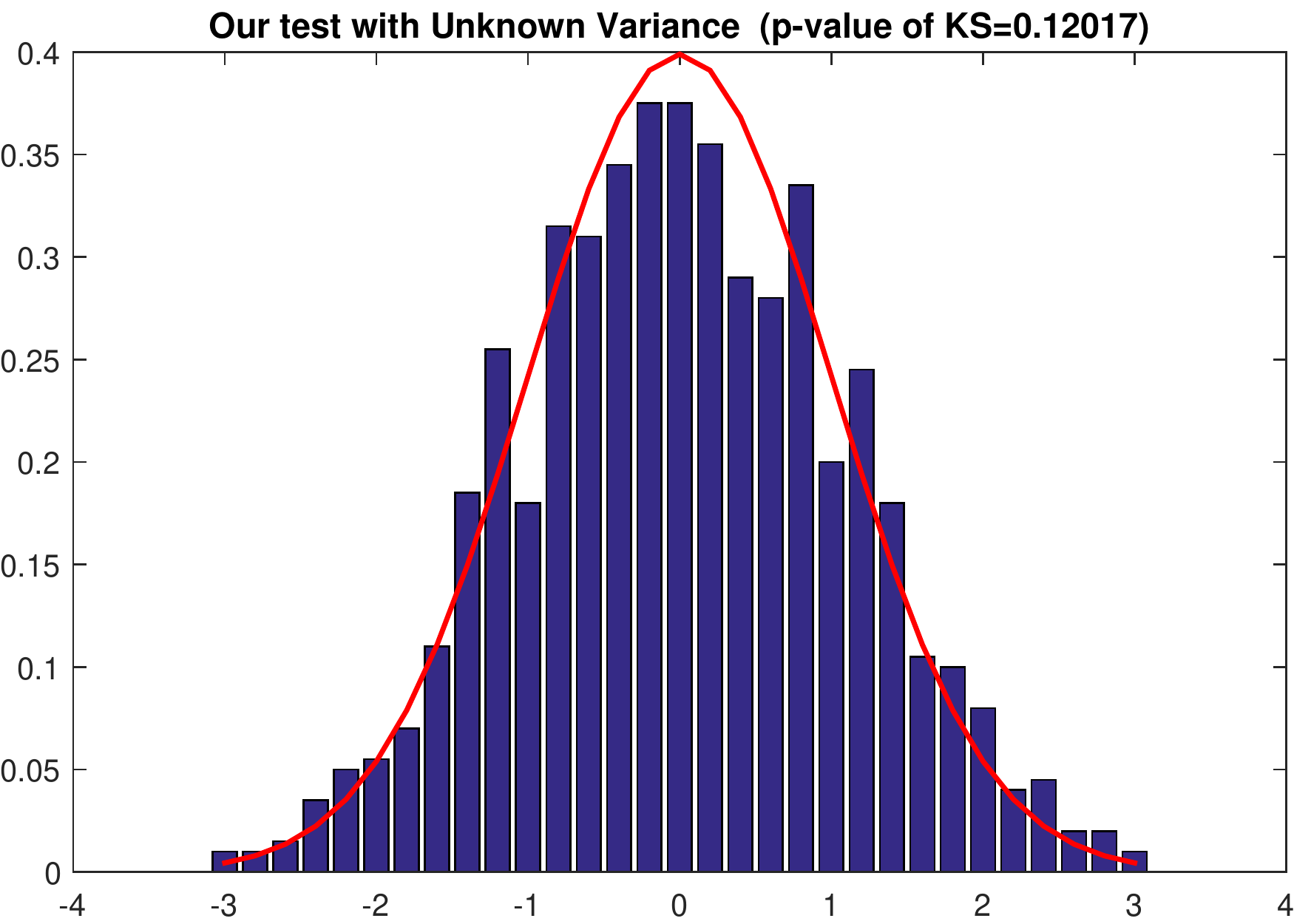}\includegraphics[scale=0.5]{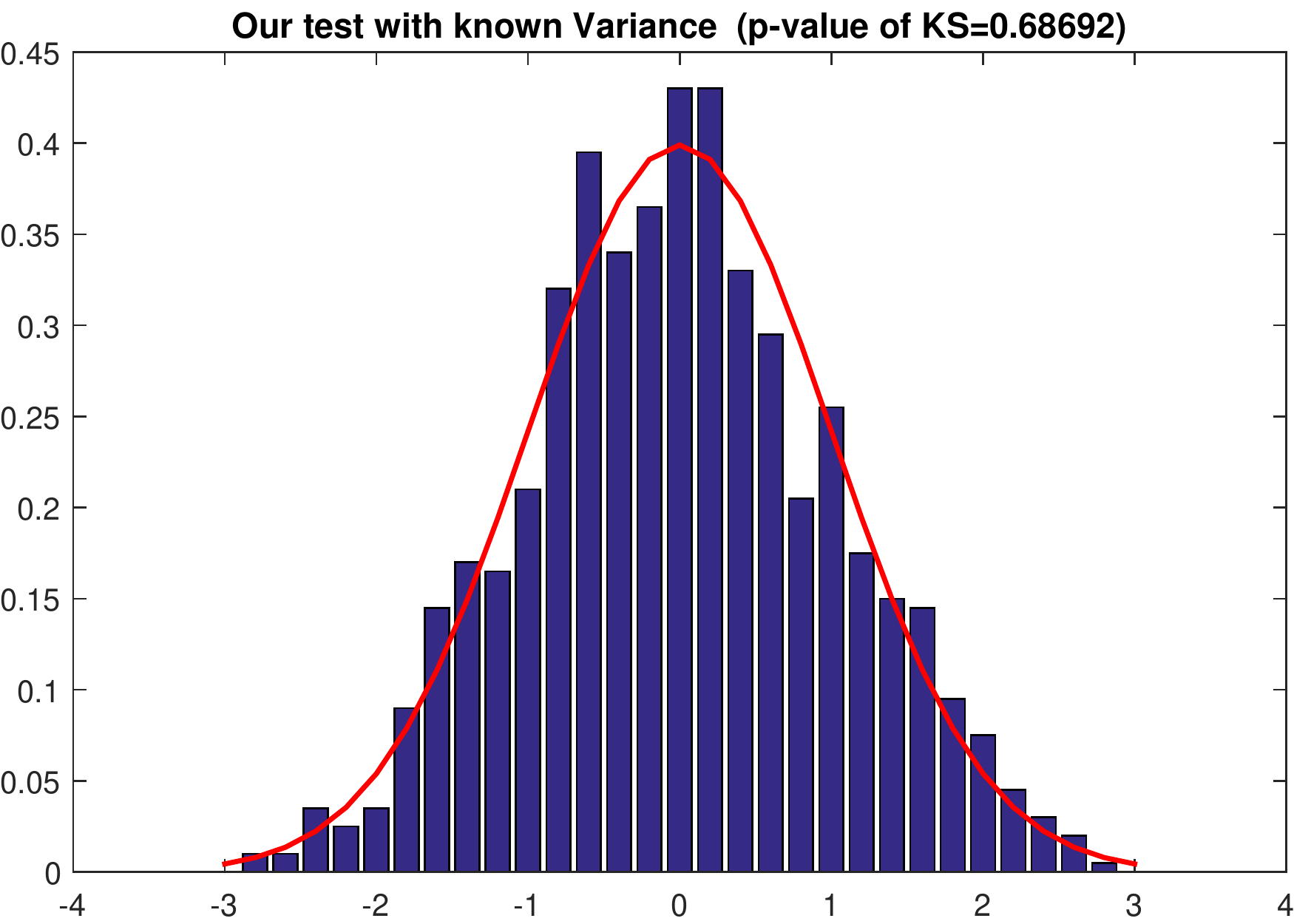}
		\par\end{centering}
	
	\bigskip{}

	\centering{}\includegraphics[scale=0.5]{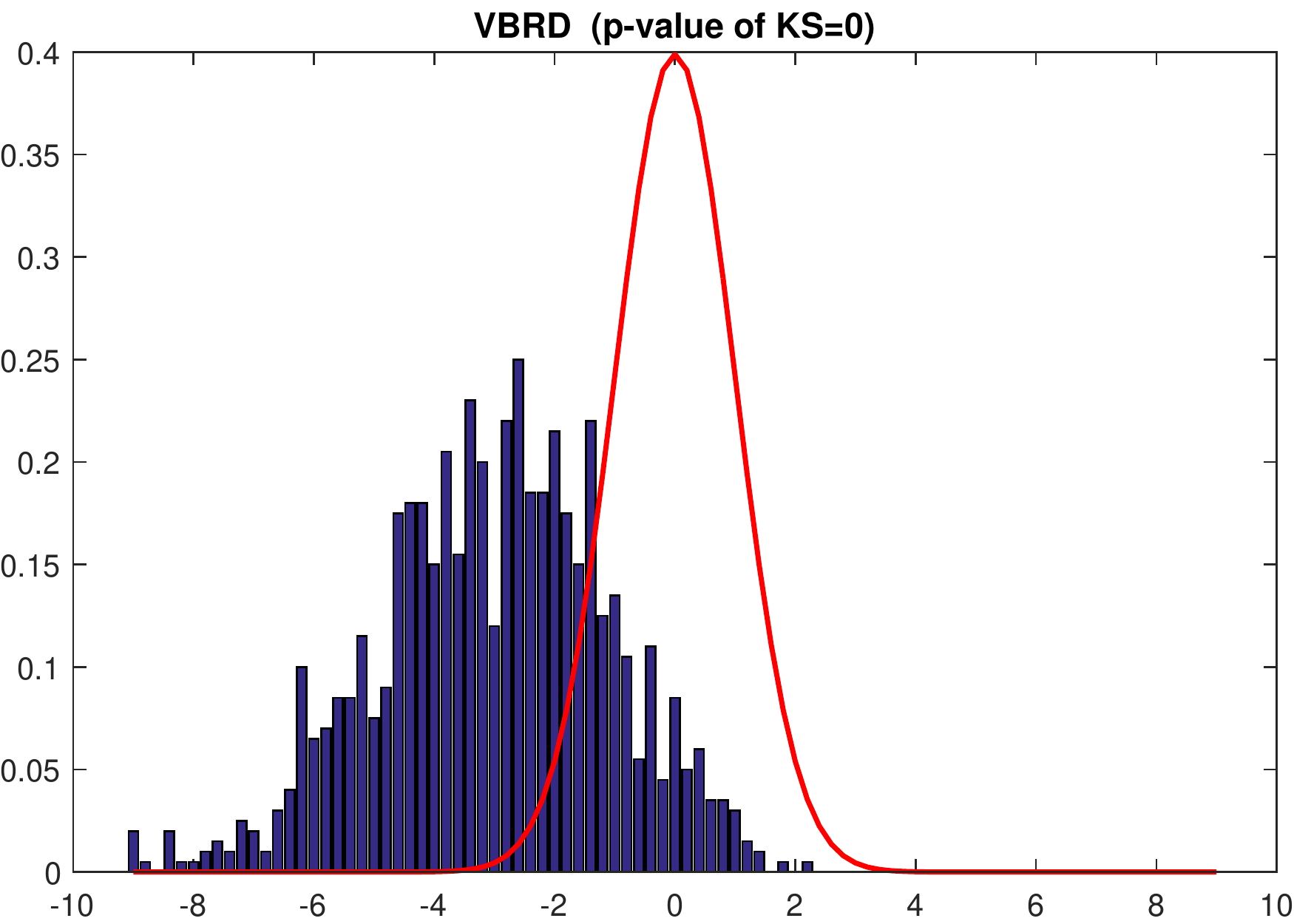}
	\floatfoot{\small   Note that tuning parameters for all the methods are chosen according to their ``oracle'' theoretical values. Error and design are normally distributed with Toeplitz correlation structure with $\rho=0.4$. The histograms are computed based on   $500$ simulation runs.}
\end{figure}

\begin{table}
	\caption{\label{tab:Size1} Mean Size Properties over $500$ repetitions of the 5\% level proposed tests together with  VBRD and BCH. In the table, NA symbol indicates that the method cannot be implemented ``as is''.}
	\centering{}%
	\begin{tabular}{c||cccc}
		&  &  &  & \tabularnewline		
		Hypothesis Setting& Unknown $\Sigma_{X}$ & Known $\Sigma_{X}$ & VBRD & BCH\tabularnewline
		\hline\hline
		Sparse $\beta$ and Sparse $a$ & 7.4\% & 5.6\% & 8.2\% & 6.6\% \tabularnewline
		Sparse $\beta$ and Dense $a$ & 4.4\% & 4.8\% & 7.4\% & NA\tabularnewline
		Dense $\beta$ and Sparse $a$ & 3.6\% & 4.4\% & 33.4\% & 27.2\% \tabularnewline
		Dense $\beta$ and Dense $a$ & 5.6\% & 3.0\% & 67.2\% & NA\tabularnewline
	\end{tabular}
\end{table}

%
%
%
%
%
%
%
%
%
%
%
%
%

We also contrast  the power properties of the proposed tests with respect to the existing methods.  Results are collected in Figures \ref{fig:Power_Toeplitz}, \ref{fig:Power_equal_corr} and \ref{fig:Power_FanSong}, where we plot the power curves of competing methods  for design Examples 1, 2 and 3 described above  with hypothesis setting of (i)-(iv).
The overall message is clear from these figures: our tests and existing methods are quite similar for sparse $\beta_*$ and sparse $a$, whereas our tests behave nominally for other problems with preserving both low Type I error rates and Type II error rates. The biggest advantages are seen for dense vectors $\beta_*$ with other methods behaving in a manner close to random guessing.  In addition to the advantages in Type I error, our methods also display certain power advantages.   In the case of equal-correlation setting we observe  that our methods consistently reach faster power than BCH method even in the case of all sparse setting. Observe that the precision matrix in this setting is not sparse and our methods are still well-behaved. 
In the case of dense models, VBRD method completely breaks down with Type I or Type II error being close to 1. 
 For non-Gaussian design we see that VBRD may not be a nominal test any more  regardless of the model sparsity.  BCH  behaves more stably in this case but fails to apply  for the hypothesis settings (ii) and (iv) as described at the beginning of the Section. 
 In conclusion, we observe that our methods are stable across vastly different designs and model setting whereas existing methods fail to control either Type I error rate or Type II error rate. Hence the proposed methodology offers a robust and more widely applicable alternative to the existing inferential procedures, achieving better error control in difficult setting and not loosing much in the simples cases.

\begin{figure}
	
	\caption{\label{fig:Power_Toeplitz} Power curves of competing methods across different hypothesis $a^\top \beta_*=g_0$ settings. Design settings  follows Example 1 with $n=100$ and $p=500$. The alternative hypothesis takes the form of $a^\top \beta_*=g_0+h$ with $h$ presented on the x-axes. The y-axes contains the average rejection probability over $500$ repetition. Therefore, $h=0$ corresponds to Type-I error and the remaining ones the Type II error. ``Known variance'' denotes the method as is introduced in Section 2 whereas, ``unknown variance'' denotes the method introduced in Section 3. VBRD and BCH refer to the methods proposed in  \cite{van2014asymptotically} and  \cite{belloni2014inference}, respectively.}

	\begin{centering}
		\noindent\begin{minipage}[t]{1\columnwidth}%
			\begin{singlespace}
				\begin{center}
					(Sparse $\beta$ and Sparse $a$)\qquad{}\qquad{}\qquad{}\qquad{}\qquad{}(Sparse
					$\beta$ and Dense $a$)
					\par\end{center}
			\end{singlespace}
			\begin{center}
				\includegraphics[scale=0.5]{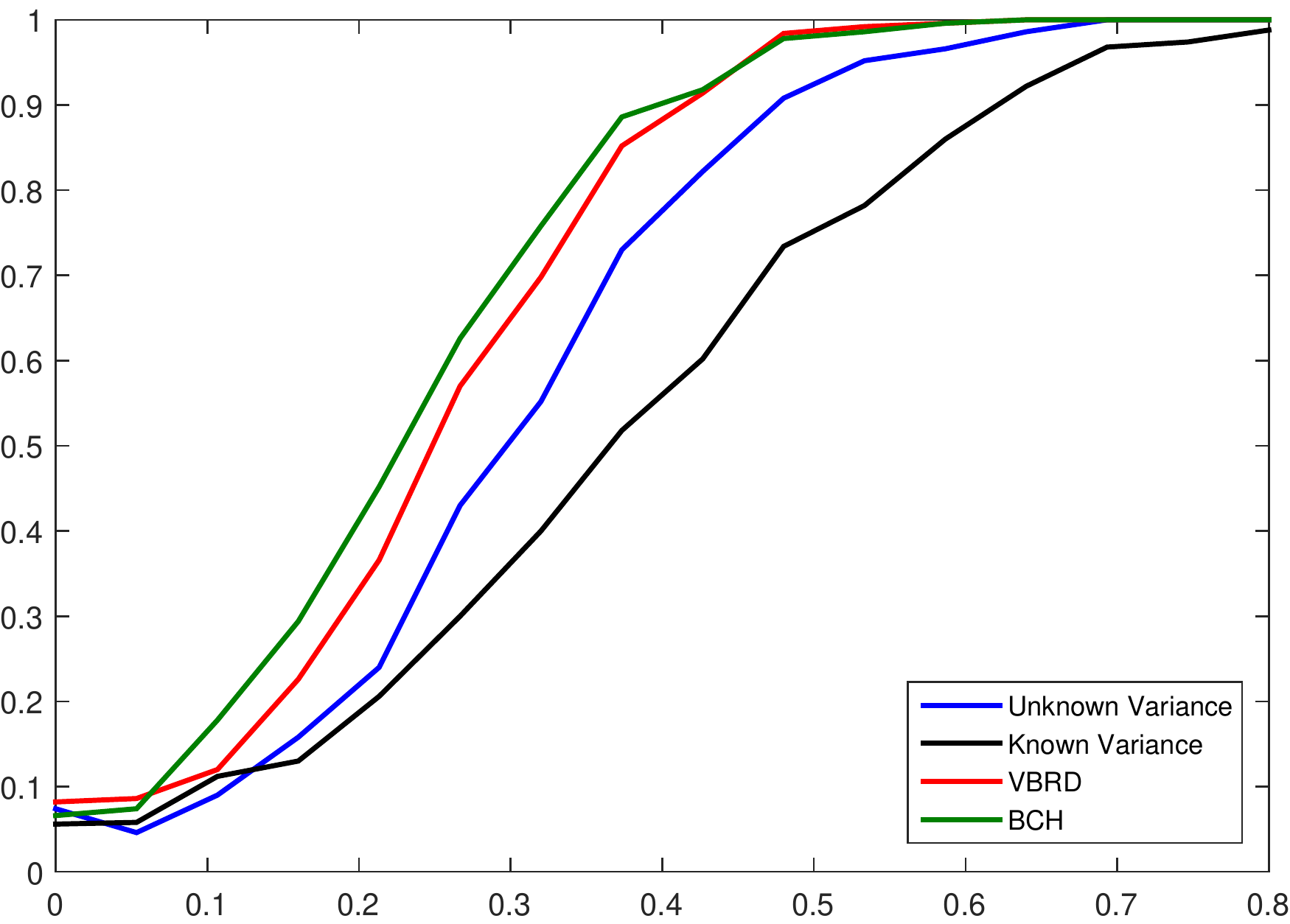}\includegraphics[scale=0.5]{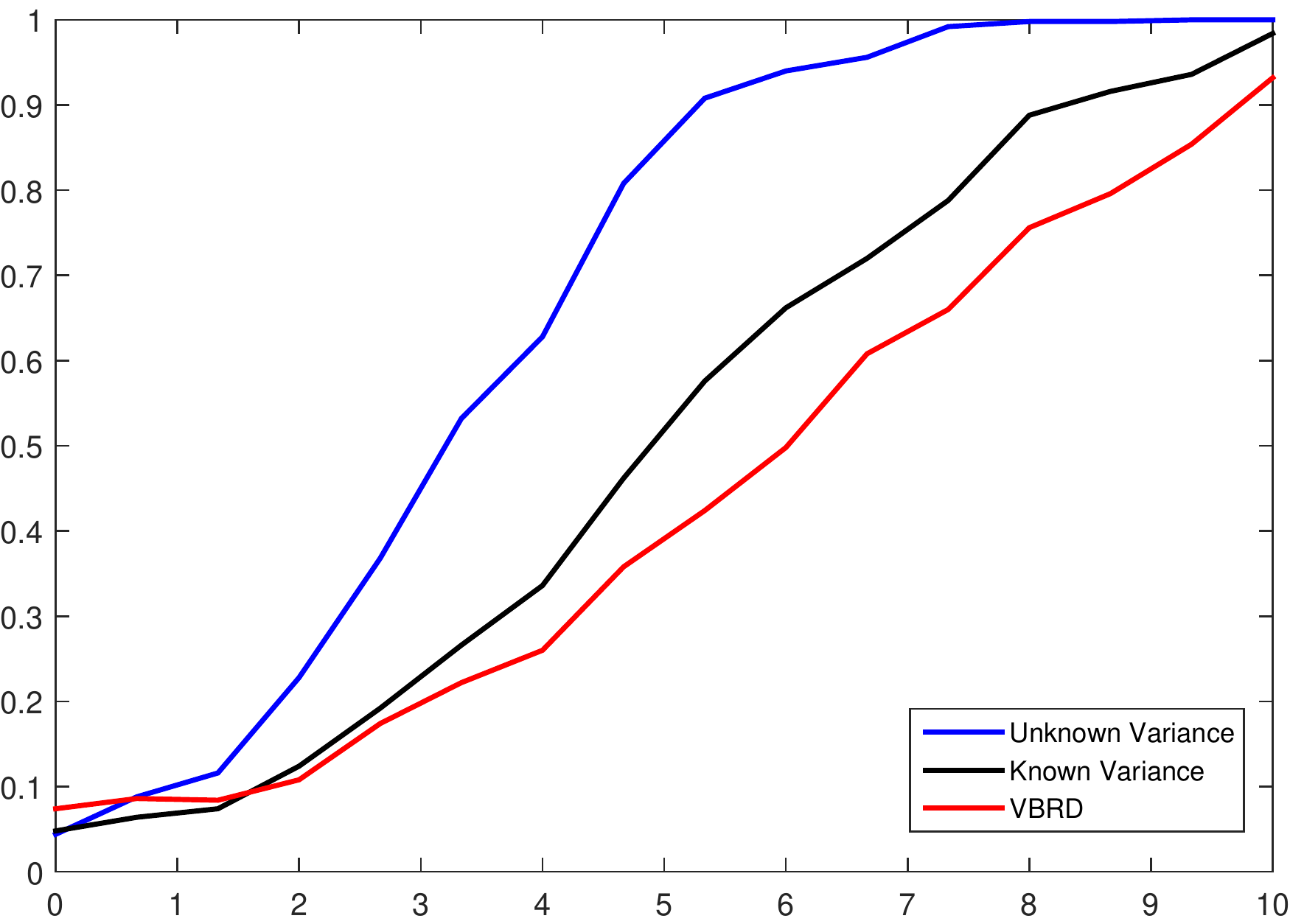}
				\par\end{center}%
		\end{minipage}
		\par\end{centering}
	\begin{centering}
		\bigskip{}
		\par\end{centering}
	\begin{centering}
		\bigskip{}
		\par\end{centering}
	\begin{centering}
		\bigskip{}
		\par\end{centering}
	\begin{centering}
		\noindent\begin{minipage}[t]{1\columnwidth}%
			\begin{center}
				(Dense $\beta$ and Sparse $a$)\qquad{}\qquad{}\qquad{}\qquad{}\qquad{}(Dense
				$\beta$ and Dense $a$)
				\par\end{center}
			\begin{center}
				\includegraphics[scale=0.5]{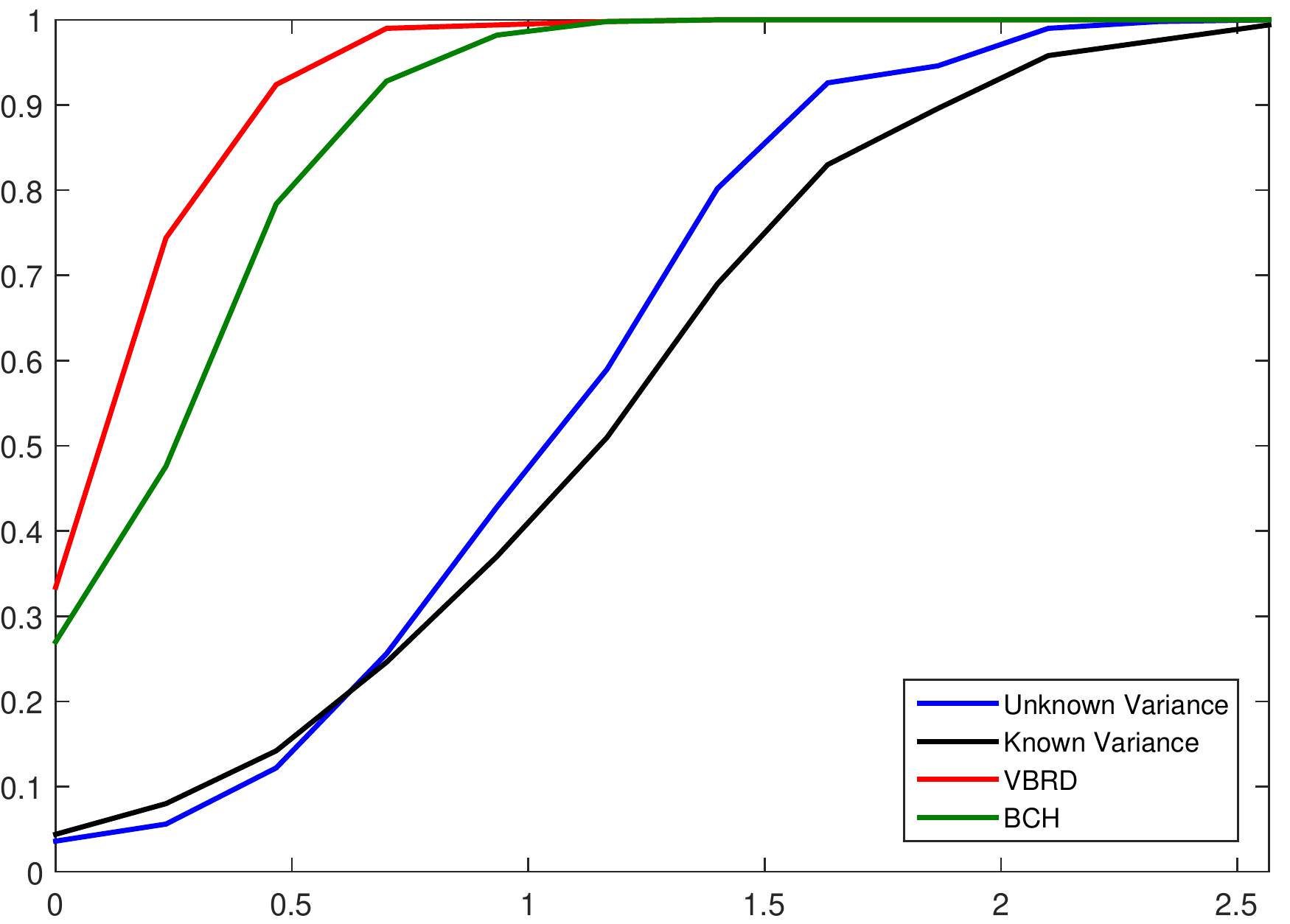}\includegraphics[scale=0.5]{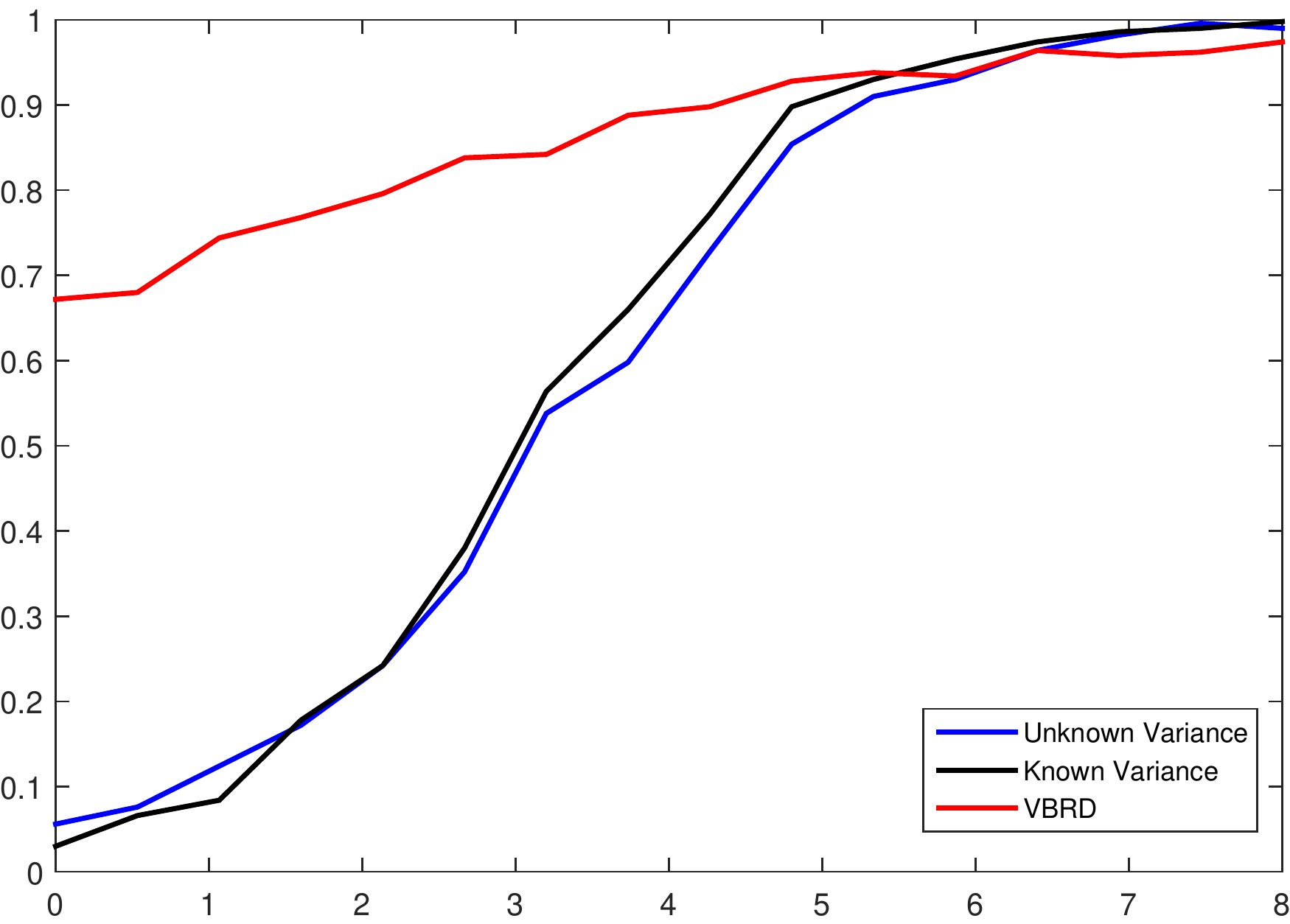}
				\par\end{center}%
		\end{minipage}
		\par\end{centering}
		\floatfoot{ Note that tuning parameters for all the methods are chosen according to their ``oracle'' theoretical values. If a method could not be implemented as is proposed in its respective paper it wasn't included in the graph. }
\end{figure}

\begin{figure}
	
	\caption{\label{fig:Power_equal_corr}Power curves of competing methods across different hypothesis $a^\top \beta_*=g_0$ settings. Design settings  follows Example 2 with $n=100$ and $p=500$. The alternative hypothesis takes the form of $a^\top \beta_*=g_0+h$ with $h$ presented on the x-axes. The y-axes contains the average rejection probability over $500$ repetition. Therefore, $h=0$ corresponds to Type-I error and the remaining ones the  Type II error. ``Known variance'' denotes the method as is introduced in Section 2 whereas, ``unknown variance'' denotes the method introduced in Section 3. VBRD and BCH refer to the methods proposed in  \cite{van2014asymptotically} and  \cite{belloni2014inference}, respectively.}
	
	\bigskip{}
	\bigskip{}
	
	\begin{centering}
		\noindent\begin{minipage}[t]{1\columnwidth}%
			\begin{singlespace}
				\begin{center}
					(Sparse $\beta$ and Sparse $a$)\qquad{}\qquad{}\qquad{}\qquad{}\qquad{}(Sparse
					$\beta$ and Dense $a$)
					\par\end{center}
			\end{singlespace}
			\begin{center}
				\includegraphics[scale=0.5]{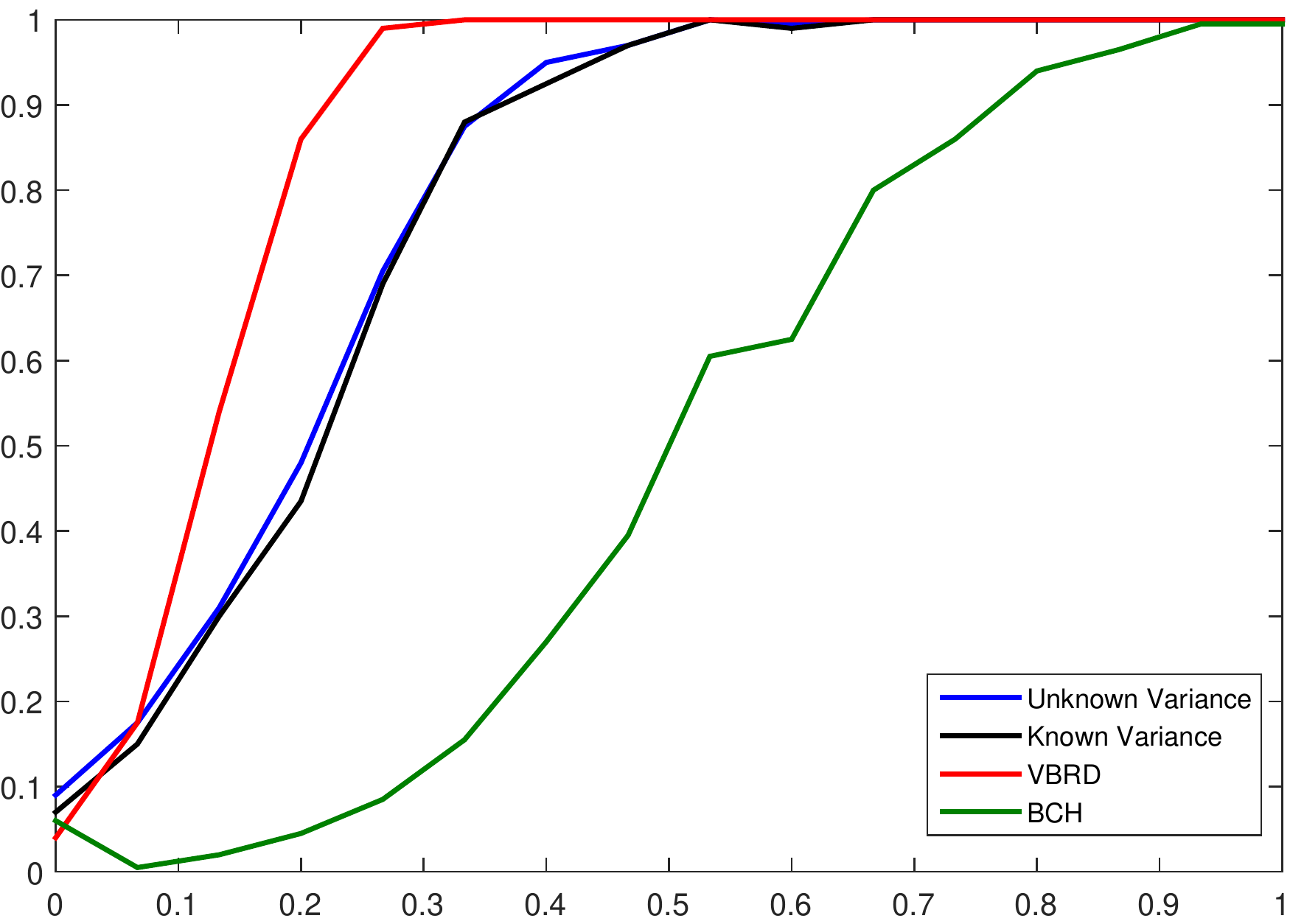}\includegraphics[scale=0.5]{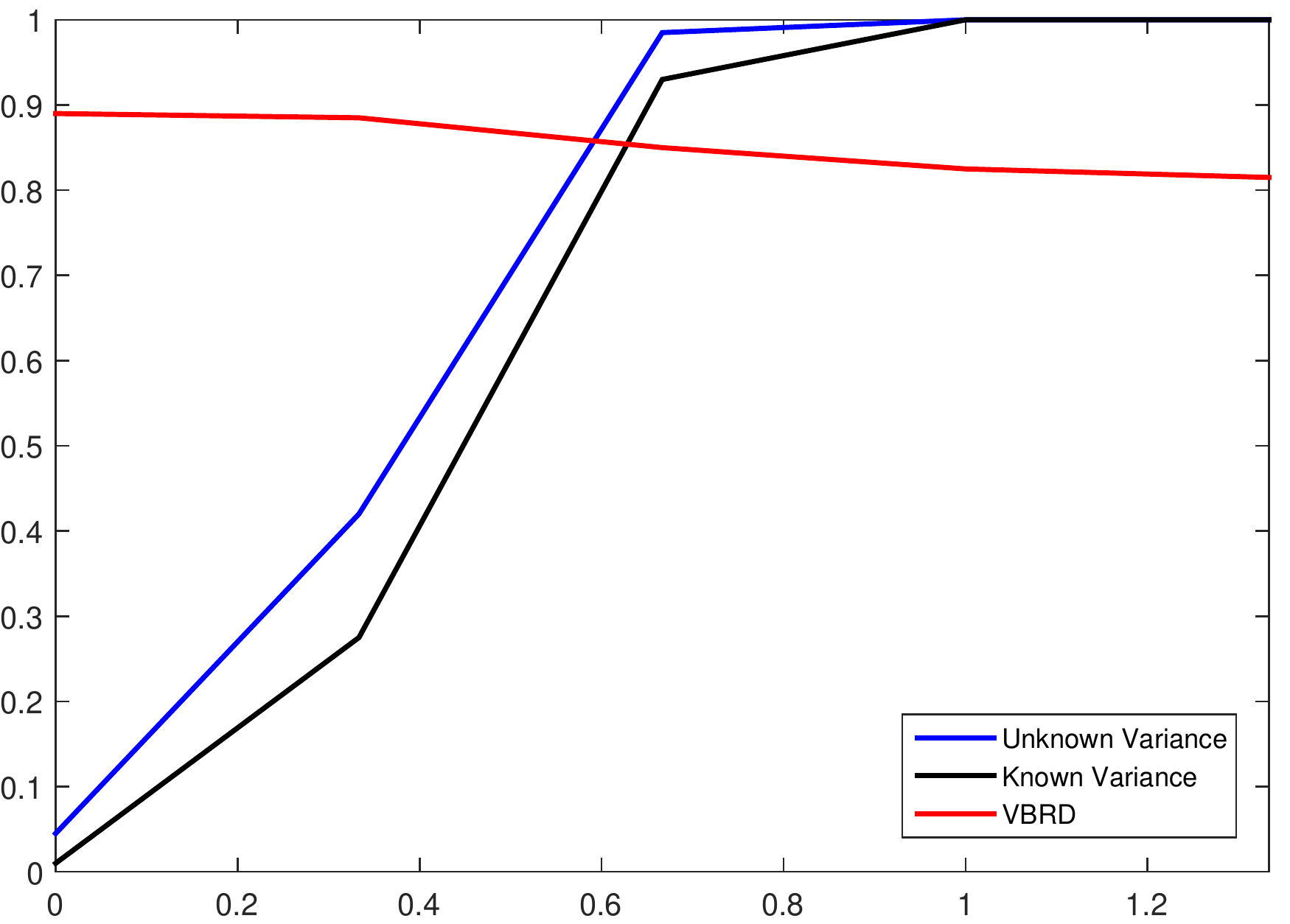}
				\par\end{center}%
		\end{minipage}
		\par\end{centering}
	\begin{centering}
		\bigskip{}
		\par\end{centering}
	\begin{centering}
		\bigskip{}
		\par\end{centering}
	\begin{centering}
		\bigskip{}
		\par\end{centering}
	\begin{centering}
		\noindent\begin{minipage}[t]{1\columnwidth}%
			\begin{center}
				(Dense $\beta$ and Sparse $a$)\qquad{}\qquad{}\qquad{}\qquad{}\qquad{}(Dense
				$\beta$ and Dense $a$)
				\par\end{center}
			\begin{center}
				\includegraphics[scale=0.5]{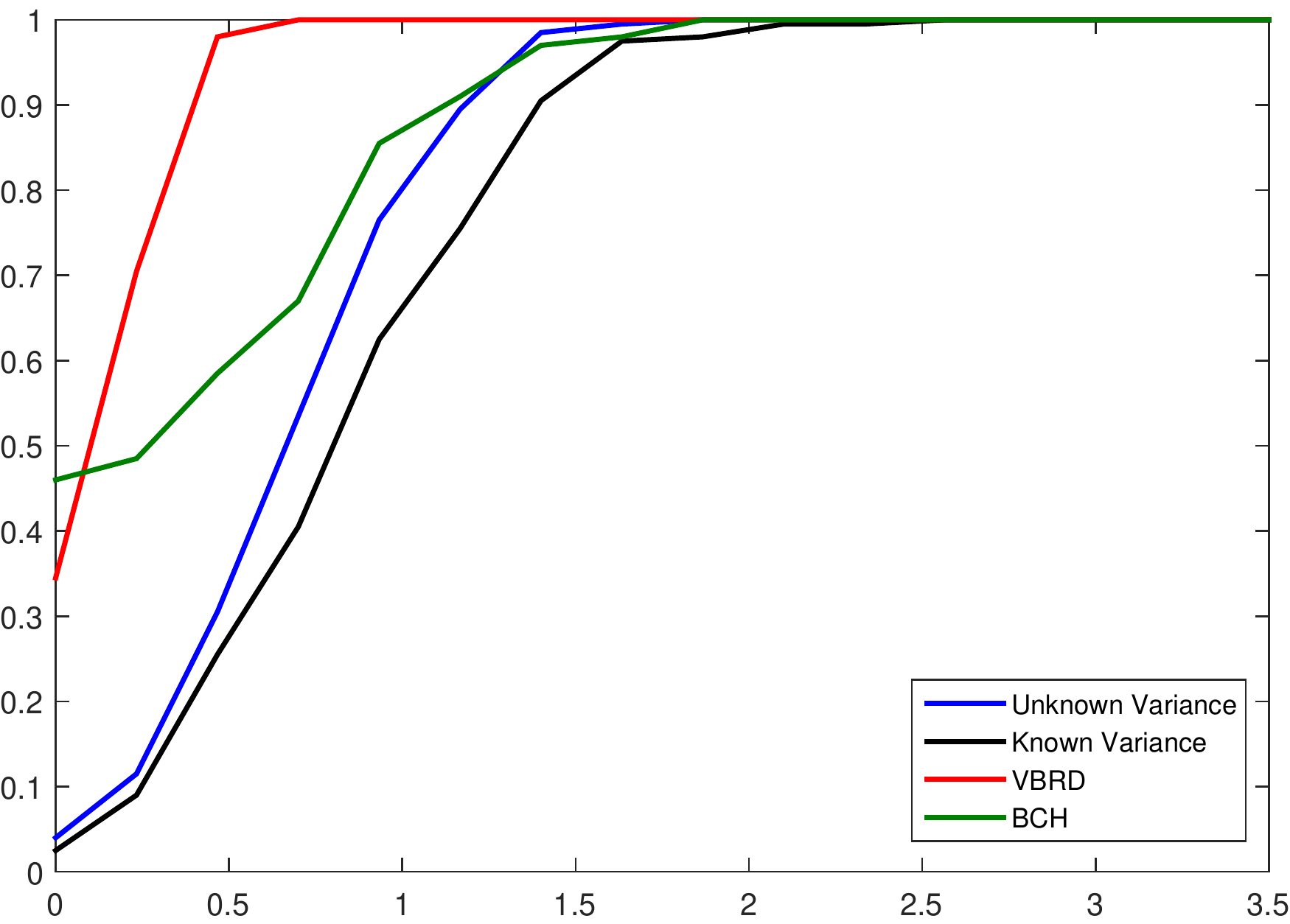}\includegraphics[scale=0.5]{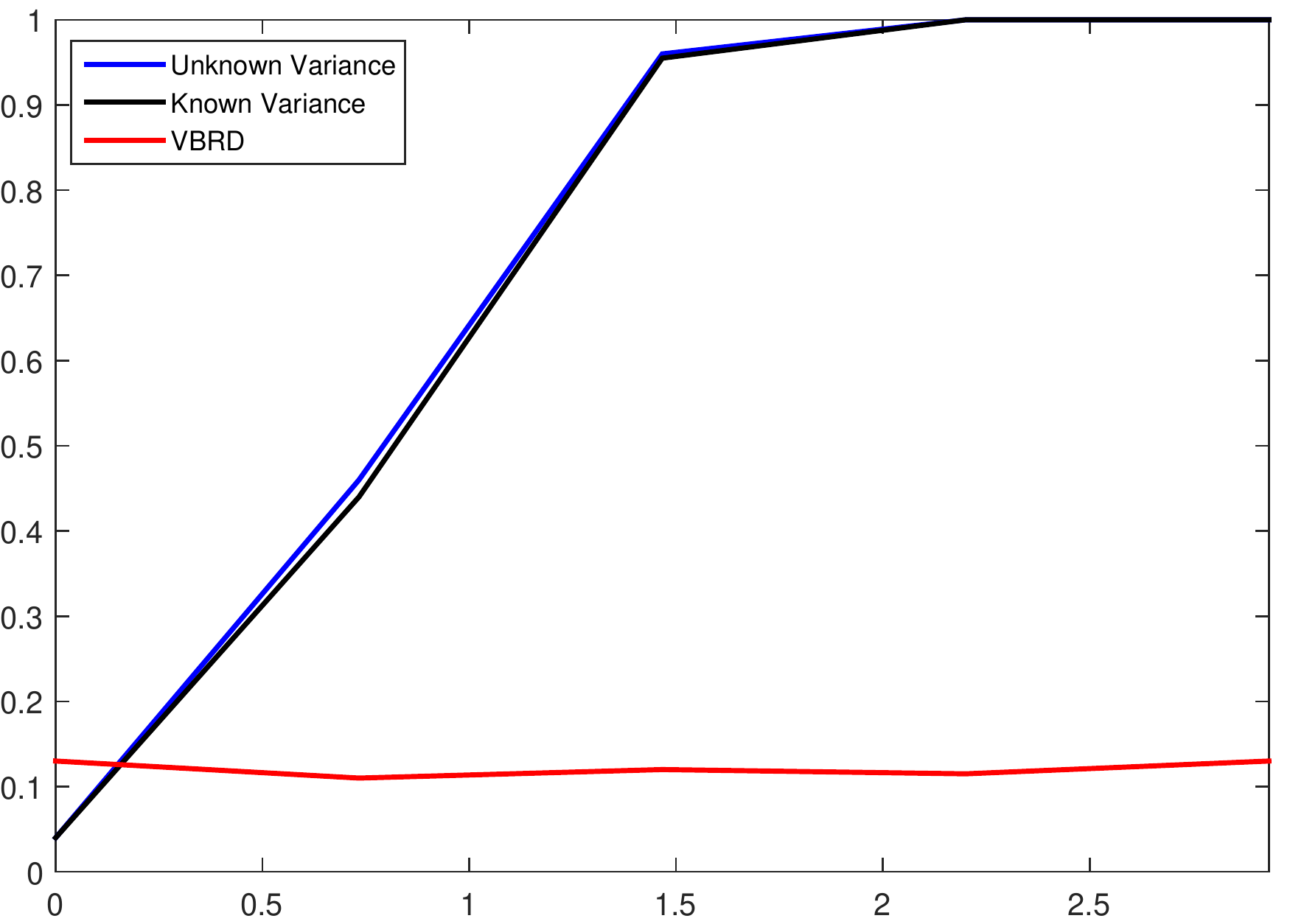}
				\par\end{center}%
		\end{minipage}
		\par\end{centering}
		\floatfoot{ Note that tuning parameters for all the methods are chosen according to their ``oracle'' theoretical values. If a method could not be implemented as is proposed in its respective paper it wasn't included in the graph. }
\end{figure}

\begin{figure}
	
	\caption{\label{fig:Power_FanSong}Power curves of competing methods across different hypothesis $a^\top \beta_*=g_0$ settings. Design settings  follows Example 3 with $n=100$ and $p=500$. The alternative hypothesis takes the form of $a^\top \beta_*=g_0+h$ with $h$ presented on the x-axes. The y-axes contains the average rejection probability over $500$ repetition. Therefore, $h=0$ corresponds to Type-I error and the remaining ones the Type II error. ``Known variance'' denotes the method as is introduced in Section 2 whereas, ``unknown variance'' denotes the method introduced in Section 3. VBRD and BCH refer to the methods proposed in  \cite{van2014asymptotically} and  \cite{belloni2014inference}, respectively.}
	
	\bigskip{}
	\bigskip{}
	
	\begin{centering}
		\noindent\begin{minipage}[t]{1\columnwidth}%
			\begin{singlespace}
				\begin{center}
					(Sparse $\beta$ and Sparse $a$)\qquad{}\qquad{}\qquad{}\qquad{}\qquad{}(Sparse
					$\beta$ and Dense $a$)
					\par\end{center}
			\end{singlespace}
			\begin{center}
				\includegraphics[scale=0.5]{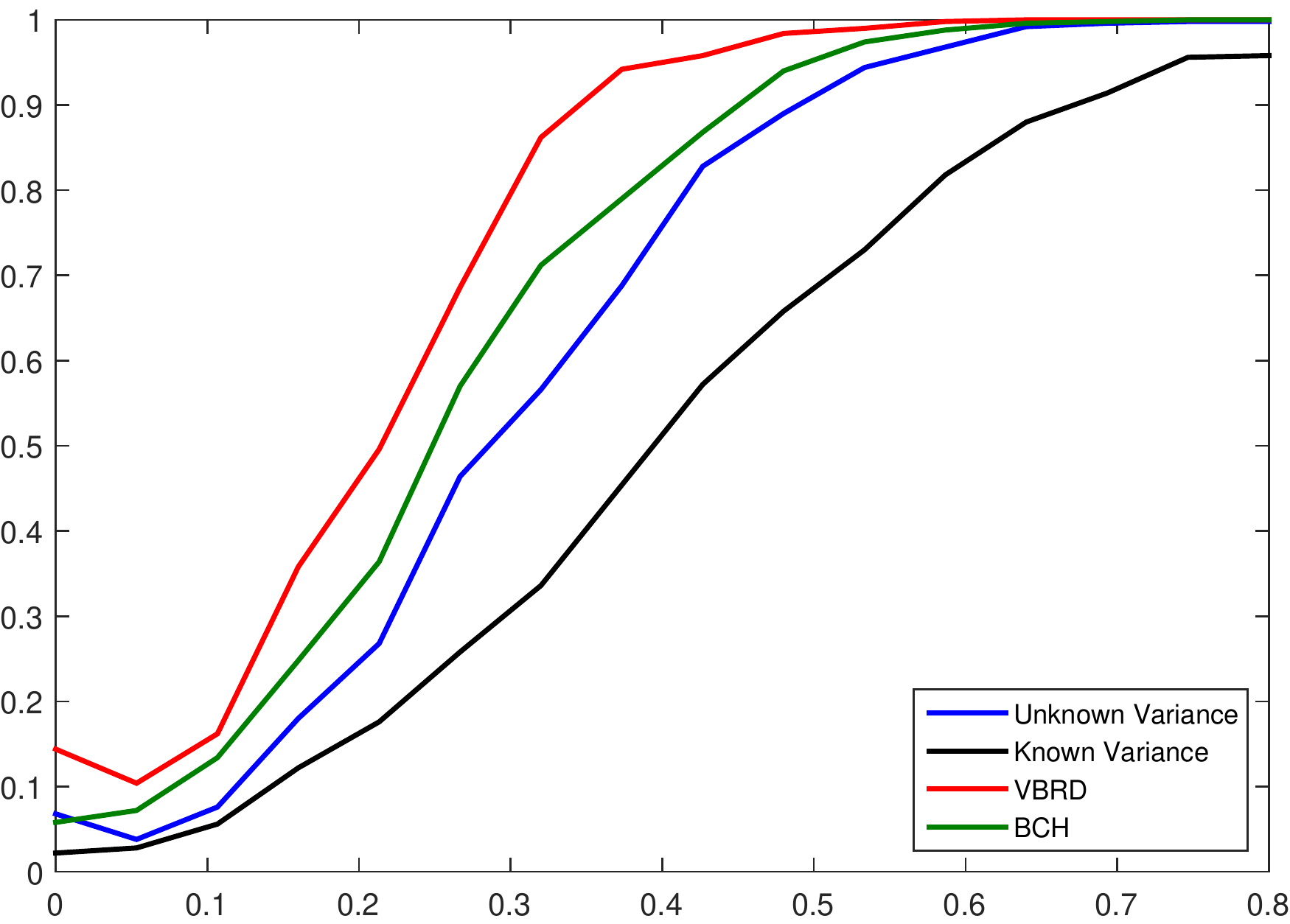}\includegraphics[scale=0.5]{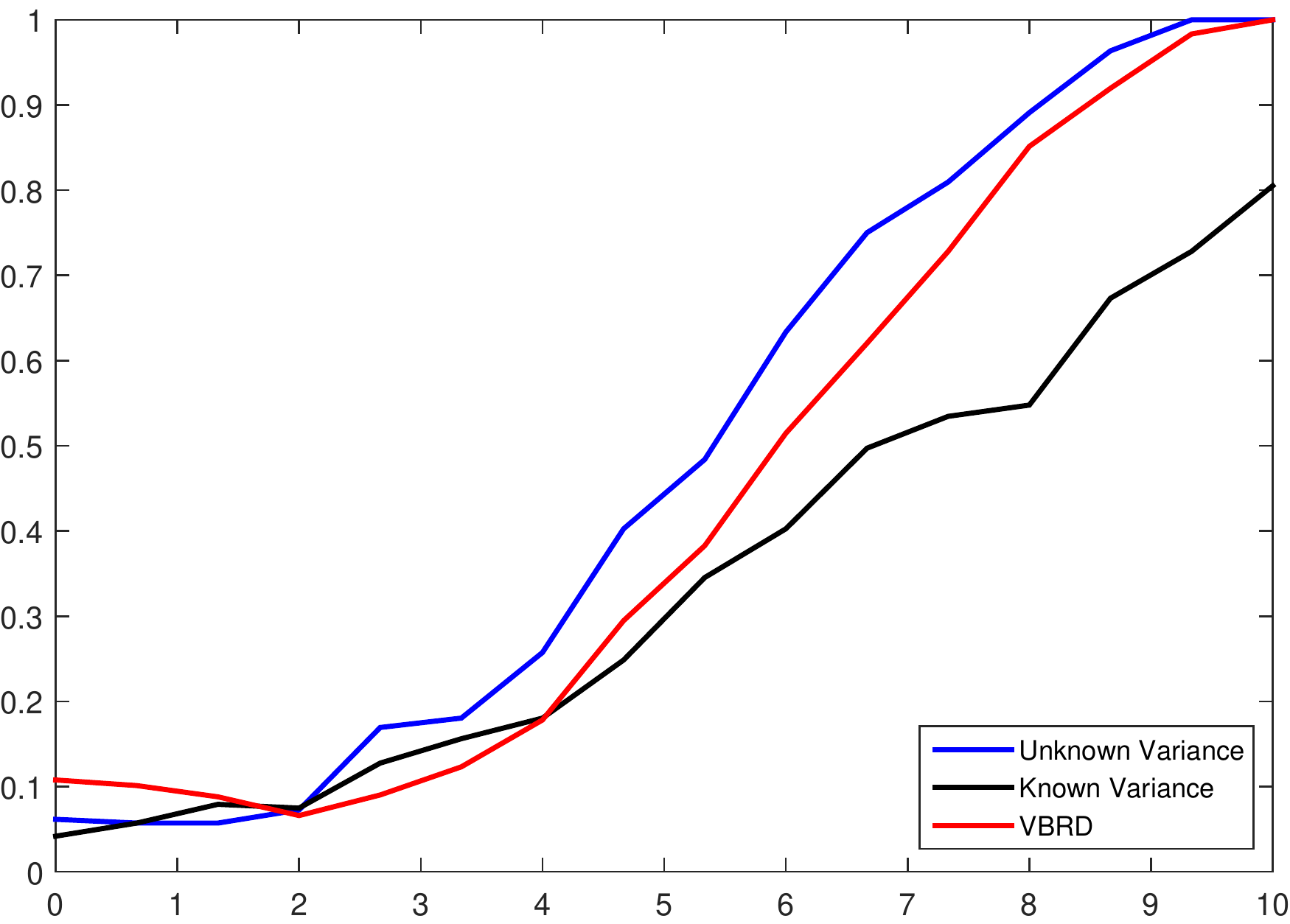}
				\par\end{center}%
		\end{minipage}
		\par\end{centering}
	\begin{centering}
		\bigskip{}
		\par\end{centering}
	\begin{centering}
		\bigskip{}
		\par\end{centering}
	\begin{centering}
		\bigskip{}
		\par\end{centering}
	\begin{centering}
		\noindent\begin{minipage}[t]{1\columnwidth}%
			\begin{center}
				(Dense $\beta$ and Sparse $a$)\qquad{}\qquad{}\qquad{}\qquad{}\qquad{}(Dense
				$\beta$ and Dense $a$)
				\par\end{center}
			\begin{center}
				\includegraphics[scale=0.5]{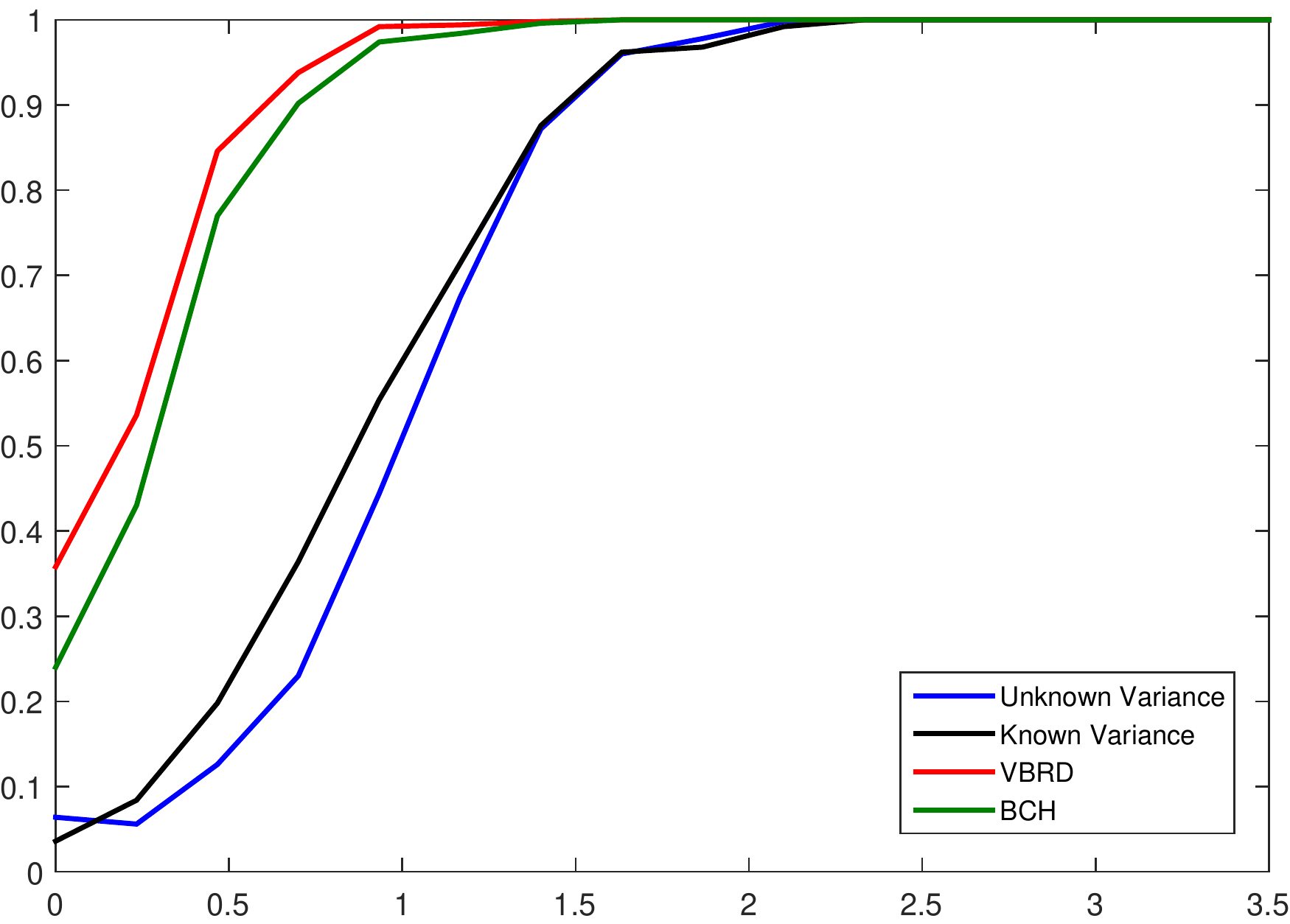}\includegraphics[scale=0.5]{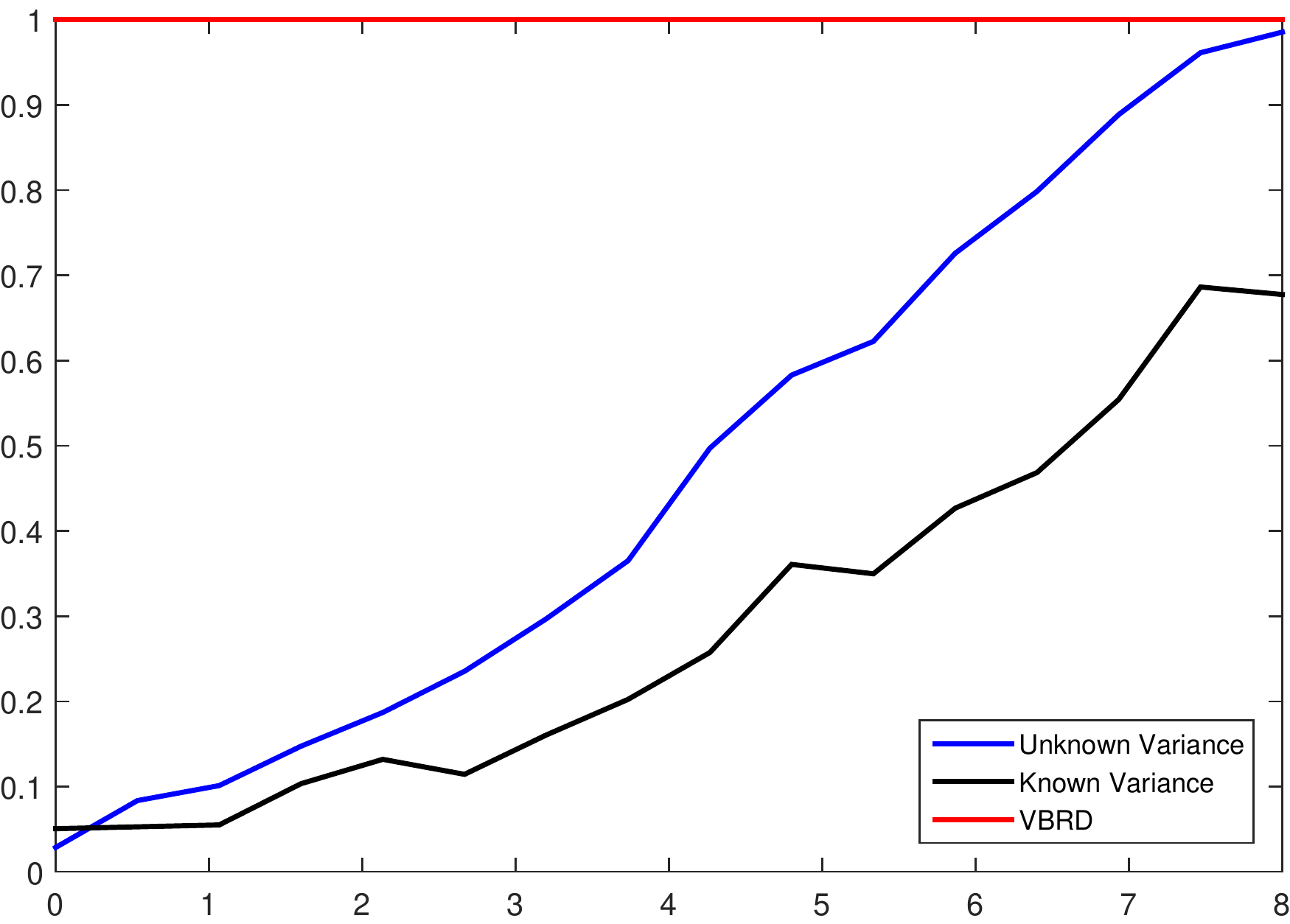}
				\par\end{center}%
		\end{minipage}
		\par\end{centering}
		\floatfoot{ Note that tuning parameters for all the methods are chosen according to their ``oracle'' theoretical values. If a method could not be implemented as is proposed in its respective paper it wasn't included in the graph. }
\end{figure}

\subsection{Real data example: equity risk premia}\label{sec:rd}

We apply the methods developed in Section \ref{sec: extension} to
inference of equity risk premia during different states of the economy.
Some studies have found that the risk premia of stock market returns have
different predictability, depending on whether the macroeconomy is
in recession or expansion; see \citet{Rapach2010}, \citet{henkel2011time}
and \citet{Dangl2012}. One common explanation for this is  time variation
in risk premia; see \citet{henkel2011time}. It is plausible that the stock market is riskier in recessions than in expansions and thus a higher expected
return is demanded by investors, implying that the expected stock returns
 can be predicted by the state of the macroeconomy. In this
section, we revisit this argument by directly conducting inference
on the expected return of the stock market conditional on a large
number of macroeconomic variables. 

Let $y_{t}$ be the excess return of the U.S stock market observed
at time $t$ and $x_{t-1}\in\mathbb{R}^{p}$ be a large number of
macroeconomic variables observed at time $t-1$. Let $s_{t}\in\{0,1\}$
denote the NBER recession indicator; $s_{t}=1$ means that the economy
is in recession at time $t$. We would like to conduct inference on
$E(y_{t}\mid x_{t-1})$ for the two different values of $s_{t-1}$.
Formally, we wish to construct confidence intervals for the following
quantities:
 (a)  $E[E(y_{t}\mid x_{t-1})\mid s_{t-1}=1]$,
 (b)   $E[E(y_{t}\mid x_{t-1})\mid s_{t-1}=0]$  and 
 (c)  $E[E(y_{t}\mid x_{t-1})\mid s_{t-1}=1]-E[E(y_{t}\mid x_{t-1})\mid s_{t-1}=0]$.
 
We impose a linear model on the risk premia: $E(y_{t}\mid x_{t-1})=x_{t-1}^\top \beta_{*}$
for some unknown $\beta_{*}\in\mathbb{R}^{p}$. Hence, the quantities
of interest are: $a_{1}^\top \beta_{*}$, $a_{0}^\top \beta_{*}$ and $(a_{1}-a_{0})^\top \beta_{*}$,
where $a_{j}=E(x_{t-1}\mid s_{t-1}=j)$. The macroeconomic variables
we use are from the dataset constructed by \citet{mccracken2015fred}.
We also include the squared, cubed and fourth power of these variables,
leading to $p=440$ (after removing variables with more than 30 missing
observations). It is possible that $\beta_{*}\in\mathbb{R}^{p}$ is
not a sparse vector because many macroeconomic variables might be
relevant and each might only explain a tiny fraction of the equity
risk premia. Therefore, the methods proposed in this article are particularly
useful because they do not assume the sparsity of $\beta_{*}$. 
\begin{rem}
There have been numerous attempts to include information from many
macroeconomic variables in estimating the equity risk premium. \citet{Rapach2010}
use the model combination approach by taking the simple average of
14 univariate linear models. Although this   approach
manages to reduce the variance in the predictions, it only produces
a single point prediction and does not deliver a confidence interval.
Moreover, under the specification of $E(y_{t}\mid x_{t-1})=x_{t-1}^\top \beta_{*}$,
we should not expect the simple average of predictions by individual
components of $x_{t-1}$ to be close to $x_{t-1}^\top \beta_{*}$, especially with highly correlated regressors. Another
popular approach is to use factor models. This method is widely used
in macroeconomics for predictions; see \citet{stock2002forecasting},
\citet{stock2002macroeconomic} and \citet{mccracken2015fred}. The
idea is to extract  a few principal components (PC's) from $x_{t}$
and to predict $y_{t}$ using these PC's. Although the PC's account
for a large variation in $x_{t-1}$, they are not hard-wired to have
high predictive power for $y_{t}$ unless we assume that the PC's
capture the factors that drive $y_{t}$. In some sense, this factor
approach only uses information in $x_{t-1}$ that is relevant for
predicting variations among different components of $x_{t-1}$; by
contrast, the methodology we propose in this article allows us to use
all the information in $x_{t-1}$. 
\end{rem}
Our dataset has 659 monthly observations starting from 1960. We use
the first 20 years ($n=240$) to train the data and the last $659-n$
months to compute $a_{j}=\sum_{t=n+1}^{659}x_{t}\mathbf{1}\{s_{t}=j\}/\sum_{t=n+1}^{659}\mathbf{1}\{s_{t}=j\}$.
In other words, we investigate the equity risk premia between 1980
and 2014. We conduct inference on the average equity risk premia in different states of the macroeconomy. The 95\% confidence intervals for $a_{1}^\top \beta_{*}$, $a_{0}^\top \beta_{*}$
and $(a_{1}-a_{0})^\top \beta_{*}$ are reported in Table \ref{tab: CI risk premia}.

\begin{table}
\protect\caption{\label{tab: CI risk premia}95\% confidence intervals for equity risk
premia}

\begin{centering}
\begin{tabular}{lcc}
 & Lower bound & Upper bound\tabularnewline
Risk premia in expansion $a_{0}^\top \beta_{*}$:  & 2.79 & 10.94\tabularnewline
Risk premia in recession  $a_{1}^\top \beta_{*}$:  & 6.32 & 36.92\tabularnewline
Risk premia difference $(a_{1}-a_{0})^\top \beta_{*}$:  & 5.13 & 38.30\tabularnewline
\end{tabular}
\par\end{centering}

The values are reported in annualized percentage, i.e., $2.79$ means
$2.79\%$. 
\end{table}

The confidence intervals in Table \ref{tab: CI risk premia}  are  very informative for our purpose. The results presented in Table \ref{tab: CI risk premia},
imply that the risk premia in recessions are higher
than in expansions and that the magnitude of difference is economically
meaningful. These results are consistent with existing literature;
see Table 1 of \citet{henkel2011time}. Figure \ref{fig: CI each time}
plots the confidence intervals for $E(y_{t}\mid x_{t-1})$ at each
$t$. This figure is consistent with the hypothesis that, during the Recessions (e.g., in the early 80's or around 2008),
the risk premia went up substantially. 

\begin{figure}

\protect\caption{\label{fig: CI each time} 95\% confidence interval for the risk premia
at each time period (the blue band) with the grey shades representing the NBER recession periods. }

\begin{centering}
\includegraphics[scale=0.7]{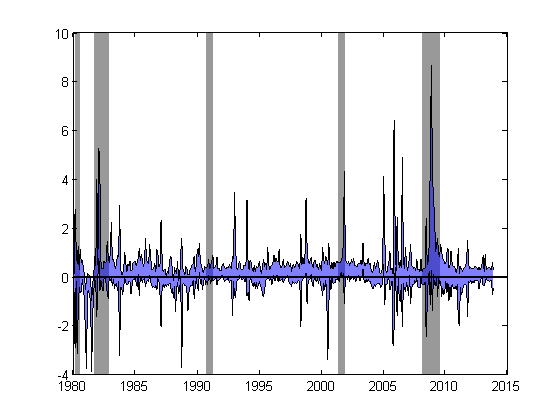}
\par\end{centering}

\end{figure}

\section*{Discussions}

In this article, we develop new methodology for testing hypotheses on $a^\top \beta_*$, where \(a\) is given and \(\beta_*\) is the regression parameter of a high-dimensional linear model. Under the proposed methodology, a new restructured regression  and with features that are synthesized  and augmented, is constructed based on  \(a\) and is used to obtain  moment conditions that are equivalent to the null hypothesis. Estimators  proposed are tailored to the problem at hand and solve constrained high-dimensional optimization problems. The two proposed methods deal with the scenario with known $\Sigma_X$ and the scenario with unknown $\Sigma_{X}$, respectively.  The first can be used when a prior information about correlation among the features exists; a case of independent features, whereas  the second applies more broadly to many scientific examples where feature correlations need to be estimated. 
To solve a high-dimensional inference problem, there exists at least one competing choice. It is based on the ``debiasing'' principles of \cite{zhang2014confidence}. However, the principles laid out therein only apply to strictly sparse linear models. Therefore, we fulfill an important gap in the existing literature by developing methodology that allows fully non-sparse linear models.

Restructuring the model according to the  hypothesis under testing allows for the high-dimensional \(a\) and \(\beta_*\) that are not necessarily sparse.   The synthesized features are customized based on the null hypothesis and are close to being orthogonal. We note that this customization is the key, since the orthogonality per se is not useful. Techniques that only induce feature orthogonality, such as pre-conditioning
by \citet{jia2012preconditioning} and DECO by \citet{wang2016decorrelated}, still cannot be used to  test  $H_0:\ a^\top \beta_*=g_0$ when \(a\) and \(\beta_*\) are dense.

To conclude the article, we would like to discuss here valuable topics for future research. 
The proposed methodology can   be used to  conduct inference of conditional distributions of the response, whenever the distribution function of $\varepsilon$, $Q(\cdot)$ is known or is consistently estimated. Specific example includes construction of prediction intervals for high-dimensional linear models -- a topic of extreme importance.
For  $F_{Y|X}(y,x)=P(y_{n+1}\leq y\mid x_{n+1}=x)$   $F_{Y|X}$  can be parametrized as  $F_{Y|X}(y,x;\beta_{*},Q)=Q(y-x^{\top}\beta_{*})$.
  For a given $x$, we can obtain a confidence set for $x^{\top}\beta_{*}$
: $\hat{I}(1-\alpha,x)$ such that $P(x^{\top}\beta_{*}\in\hat{I}(1-\alpha,x))\rightarrow1-\alpha$, by inverting the tests
proposed in this article.
This leads to a natural confidence set for the $F_{Y|X}(y,x)$: $P(F_{Y|X}(\cdot,x)\in\hat{S}(1-\alpha,x))\rightarrow1-\alpha$,
where 
$$\hat{S}(1-\alpha,x)=\{Q(\cdot-c)\mid c\in\hat{I}(1-\alpha,x)\}.$$ 
If we restrict the model parameters to be sparse, then we can consistently estimate \(\varepsilon_i\) (and thus \(Q(\cdot)\)) and consequently form valid prediction intervals -- a topic of specific importance for practitioners. However, when the model is allowed to be non-sparse and high-dimensional, the question of construction of prediction intervals hasn't been answered and needs special considerations.
Additionally, under this setup, the proposed methods also lead to an inference method for
(possibly nonlinear) functionals of the conditional distribution of
$y_{n+1}$ given $x_{n+1}$. For example, suppose that one  is interested in 
$H(u,x)=\inf\{y\in\mathbb{R}\mid F_{Y|X}(y,x)\geq u\}$.  Following the above proposal, we can
simply take 
$$\hat{\mathcal{H}}(u,x,\alpha)=\{\inf\{y\in\mathbb{R}\mid Q(y-c)\geq u\}\mid c\in\hat{I}(1-\alpha,x)\}$$
as a confidence set for $H(u,x)$. 
\bibliographystyle{apalike}
\bibliography{biblio_state_space}

\newpage
\appendix

\section{Supplementary Materials} \label{ap:a}

In the rest of the article, we use $\lambda_{\min}(\cdot)$ and $\lambda_{\max}(\cdot)$
to denote the minimal and maximal eigenvalues of a matrix, respectively.
For a random variable, let $\|\cdot\|_{L^{r}(P)}$ denote the $L^{r}(P)$-norm,
i.e., $\|z_{i}\|_{L^{r}(P)}=[Ez_{i}^{r}]^{1/r}$. For a vector $x=(x_{1},\cdots,x_{p})^\top \in\mathbb{R}^{p}$,
let $\mathcal{M}(x)$ denotes its support $\{i\mid x_{i}\neq0\}$.

\subsection{Proof Theorems \ref{thm: known variance X} and \ref{thm: known variance X power}}
\begin{proof}[\textbf{Proof of Theorem \ref{thm: known variance X}}]
Under $H_{0}$ in (\ref{eq: null hypo}), $l_{i}(g_{0})=z_{i}(\varepsilon_{i}+w_{i}^\top \beta_{*})$.
Notice that 
$$\sigma_{l}^{2}=El_{i}(g_{0})^{2}=\sigma_{z}^{2}\sigma_{\varepsilon}^{2}+Ez_{i}^{2}(w_{i}^\top \beta_{*})^{2}\geq\sigma_{z}^{2}\sigma_{\varepsilon}^{2}.$$
Hence, $s_{n}^{2}:=\sum_{i=1}^{n}E(l_{i}(g_{0}))^{2}\geq n\sigma_{z}^{2}\sigma_{\varepsilon}^{2}$.
It follows that 
\begin{align*}
\frac{\sum_{i=1}^{n}E|l_{i}(g_{0})|^{3}}{s_{n}^{3}}
&\le\frac{E|z_{i}(\varepsilon_{i}+w_{i}^\top \beta_{*})|^{3}}{n^{1/2}\sigma_{\varepsilon}^{3}\sigma_{z}^{3}}
\\
&\overset{(i)}{\leq}\frac{\sqrt{\|z_{i}\sigma_{z}^{-1}\|_{L^{6}(P)}^{6}\|\varepsilon_{i}+w_{i}^\top \beta_{*}\|_{L^{6}(P)}^{6}}}{n^{1/2}c^{3}}\overset{(ii)}{=}o(1),
\end{align*}
where $(i)$ follows by Holder's inequality and $(ii)$ follows by
Assumption \ref{assu: lyapunov CLT} and Minkowski's inequality $\|\varepsilon_{i}+w_{i}^\top \beta_{*}\|_{L^{6}(P)}\leq\|\varepsilon_{i}\|_{L^{6}(P)}+\|w_{i}^\top \beta_{*}\|_{L^{6}(P)}=O(1)$.
By Lyapunov's CLT (Theorem 11.1.4 of \citet{athreya2006measure}),
$\sum_{i=1}^{n}l_{i}(g_{0})/s_{n}\rightarrow^{d}\mathcal N(0,1)$. 

By Slutsky's lemma, it suffices to show that $s_{n}/\sqrt{n^{-1}\sum_{i=1}^{n}l_{i}(g_{0})^{2}}\rightarrow^{p}1$.
Notice that this is equivalent to the condition 
\begin{equation}
n^{-1}\sum_{i=1}^{n}\left(\frac{l_{i}(g_{0})^{2}}{El_{i}(g_{0})^{2}}-1\right)=o_{P}(1).\label{eq: known variance size eq 1}
\end{equation}
By Markov's inequality, we have that, for any $M>0$, 
\begin{align}
P\left[\left(n^{-1}\sum_{i=1}^{n}\left(\frac{l_{i}(g_{0})^{2}}{El_{i}(g_{0})^{2}}-1\right)\right)^{2}>M\right]
&\leq
 M^{-1}n^{-1}E\left(\frac{l_{i}(g_{0})^{2}}{El_{i}(g_{0})^{2}}-1\right)^{2}\\
&\overset{(i)}{\leq}
2M^{-1}n^{-1}\left[\frac{El_{i}(g_{0})^{4}}{\left[El_{i}(g_{0})^{2}\right]^{2}}+1\right],\label{eq: known variance size eq 2}
\end{align}
where $(i)$ holds by the elementary inequality $(a+b)^{2}\leq2a^{2}+2b^{2}$. 

By Holder's inequality and Assumption \ref{assu: lyapunov CLT}, 
$$El_{i}(g_{0})^{4}\sigma_{z}^{-4}\leq\sqrt{\|z_{i}\sigma_{z}^{-1}\|_{L^{8}(P)}^{8}\|\varepsilon_{i}+w_{i}^\top \beta_{*}\|_{L^{8}(P)}^{8}}<C_{0}$$
for some constant $C_{0}>0$, depending only on $C$. Since $El_{i}(g_{0})^{2}\geq\sigma_{z}^{2}\sigma_{\varepsilon}^{2}\geq\sigma_{z}^{2}c^{2}$,
we have 
$$El_{i}(g_{0})^{4}/\left[El_{i}(g_{0})^{2}\right]^{2}\leq C_{0}c^{-4}<\infty.$$
This, together with (\ref{eq: known variance size eq 2}), implies
(\ref{eq: known variance size eq 1}). The proof is complete. 
\end{proof}

\begin{proof}[\textbf{Proof of Theorem \ref{thm: known variance X power}}]
Since the eigenvalues of $\Sigma_{X}$ are bounded away from zero
and infinity, we have $\sigma_{z}^{2}=Ez_{i}^{2}=b^\top \Sigma_{X}b=(a^\top\Omega_Xa)^{-1}\asymp\|a\|_{2}^{-2}$.
It follows, by $\sqrt{n}|h_{n}|/\|a\|_{2}\rightarrow\infty$, that
\begin{equation}
\sqrt{n}|h_{n}|\sigma_{z}\rightarrow\infty.\label{eq: known var power eq 0}
\end{equation}

It should be noted that when $a^\top\beta_{*}=g_{0}+h_{n}$, we have $l_{i}(g_{0})=z_{i}(\varepsilon_{i}+w_{i}^\top \beta_{*})+z_{i}^{2}h_{n}$.
Also note that (\ref{eq: known variance size eq 2}) in the proof
of Theorem \ref{thm: known variance X} still holds, in that for all $ M>0$,
\begin{equation}
P\left[\left(n^{-1}\sum_{i=1}^{n}\left(\frac{l_{i}(g_{0})^{2}}{El_{i}(g_{0})^{2}}-1\right)\right)^{2}>M\right]\leq2M^{-1}n^{-1}\left[\frac{El_{i}(g_{0})^{4}}{\left[El_{i}(g_{0})^{2}\right]^{2}}+1\right].\label{eq: known var power eq 1}
\end{equation}
Observe that, by Assumption \ref{assu: lyapunov CLT}, 
\begin{align}
\left\Vert \frac{l_{i}(g_{0})}{\sigma_{z}(\sigma_{z}|h_{n}|\vee1)}\right\Vert _{L^{4}(P)}
&\leq
\left\Vert \frac{z_{i}(\varepsilon_{i}+w_{i}^\top \beta_{*})}{\sigma_{z}(\sigma_{z}|h_{n}|\vee1)}\right\Vert _{L^{4}(P)}
+
\left\Vert \frac{z_{i}^{2}h_{n}}{\sigma_{z}(\sigma_{z}|h_{n}|\vee1)}\right\Vert _{L^{4}(P)}
\\
& \leq
\left\Vert \frac{z_{i}(\varepsilon_{i}+w_{i}^\top \beta_{*})}{\sigma_{z}}\right\Vert _{L^{4}(P)}
+\left\Vert z_{i}\sigma_{z}^{-1}\right\Vert _{L^{8}(P)}^{2}\frac{\sigma_{z}|h_{n}|}{\sigma_{z}|h_{n}|\vee1}
\\
&\leq\left\Vert z_{i}\sigma_{z}^{-1}\right\Vert _{L^{8}(P)}^{2}\left\Vert \varepsilon_{i}+w_{i}^\top \beta_{*}\right\Vert _{L^{8}(P)}^{2}+O(1)=O(1).\label{eq: known var power eq 2}
\end{align}

Observe that 
$$El_{i}(g_{0})^{2}=E\left(z_{i}\varepsilon_{i}+z_{i}(w_{i}^\top \beta_{*}+z_{i}h_{n})\right)^{2}=E(z_{i}^{2}\varepsilon_{i}^{2})+E(z_{i}^{2}(w_{i}^\top \beta_{*}+z_{i}h_{n})^{2})\geq\sigma_{z}^{2}\sigma_{\varepsilon}^{2}.$$
Also, we have $El_{i}(g_{0})^{2}\geq[El_{i}(g_{0})]^{2}=\sigma_{z}^{4}h_{n}^{2}$.
Hence, 
$$El_{i}(g_{0})^{2}\geq(\sigma_{z}^{4}h_{n}^{2}\vee\sigma_{z}^{2}\sigma_{\varepsilon}^{2})=\sigma_{z}^{2}(\sigma_{z}^{2}h_{n}^{2}\vee\sigma_{\varepsilon}^{2}).$$
This, together with (\ref{eq: known var power eq 2}) and Assumption
\ref{assu: lyapunov CLT}, implies that 
\begin{equation}
\frac{El_{i}(g_{0})^{4}}{\left[El_{i}(g_{0})^{2}\right]^{2}}\leq\frac{O(1)\left[\sigma_{z}(\sigma_{z}|h_{n}|\vee1)\right]^{4}}{\left[\sigma_{z}^{2}(\sigma_{z}^{2}h_{n}^{2}\vee\sigma_{\varepsilon}^{2})\right]^{2}}\leq O(1)\frac{\sigma_{z}^{4}h_{n}^{4}\vee1}{\sigma_{z}^{4}h_{n}^{4}\vee c^{4}}\leq O(1)(1\vee c^{-4}).\label{eq: known var power eq 3}
\end{equation} 
It follows, by (\ref{eq: known var power eq 1}) and (\ref{eq: known var power eq 3}),
that $n^{-1}\sum_{i=1}^{n}\left(l_{i}(g_{0})^{2}/El_{i}(g_{0})^{2}-1\right)=o_{P}(1)$,
which means that 
\begin{equation}
\frac{n^{-1}\sum_{i=1}^{n}l_{i}(g_{0})^{2}}{El_{i}(g_{0})^{2}}=1+o_{P}(1).\label{eq: known var power eq 4}
\end{equation}
By Markov's inequality, we have that, $\forall M>0$, 
\begin{align*}
P\left[\left(\frac{n^{-1/2}\sum_{i=1}^{n}\left(l_{i}(g_{0})-El_{i}(g_{0})\right)}{\sqrt{El_{i}(g_{0})^{2}}}\right)^{2}>M\right]
&\leq M^{-1}\frac{E\left[l_{i}(g_{0})-El_{i}(g_{0})\right]^{2}}{El_{i}(g_{0})^{2}}\\
&=M^{-1}\frac{El_{i}(g_{0})^{2}-\left[El_{i}(g_{0})\right]^{2}}{El_{i}(g_{0})^{2}}\leq M^{-1}.
\end{align*}
Hence, 
\begin{equation}
\frac{n^{-1/2}\sum_{i=1}^{n}\left(l_{i}(g_{0})-El_{i}(g_{0})\right)}{\sqrt{El_{i}(g_{0})^{2}}}=O_{P}(1).\label{eq: known var power eq 5}
\end{equation}
Lastly, we observe that
\begin{align*}
El_{i}(g_{0})^{2}
&=E\left[z_{i}(\varepsilon_{i}+w_{i}^\top \beta_{*})+z_{i}^{2}h_{n}\right]^{2}
\\
&\overset{(i)}{\leq}2Ez_{i}^{2}(\varepsilon_{i}+w_{i}^\top \beta_{*})^{2}+2Ez_{i}^{4}h_{n}^{2}\\
&\overset{(ii)}{\leq}2\sigma_{z}^{2}\sqrt{E(z_{i}\sigma_{z}^{-1})^{4}E(\varepsilon_{i}+w_{i}^\top \beta_{*})^{4}}+2E(z_{i}\sigma_{z}^{-1})^{4}\sigma_{z}^{4}h_{n}^{2}
\\
&\overset{(iii)}{\leq}O(1)+O(1)\sigma_{z}^{4}h_{n}^{2},
\end{align*}
where $(i)$ follows according to    the elementary inequality $(a+b)^{2}\leq2a^{2}+2b^{2}$,
$(ii)$ follows by Holder's inequality and $(iii)$ is determined  by Assumption
\ref{assu: lyapunov CLT}. Since $El_{i}(g_{0})=\sigma_{z}^{2}h_{n}$,
we have 
\begin{equation}
\left|\frac{n^{-1/2}\sum_{i=1}^{n}El_{i}(g_{0})}{\sqrt{El_{i}(g_{0})^{2}}}\right|\geq\frac{\sqrt{n}\sigma_{z}^{2}|h_{n}|}{\sqrt{O(1)+O(1)\sigma_{z}^{4}h_{n}^{2}}}=\left(\frac{O(1)}{n\sigma_{z}^{2}h_{n}^{2}}+O(n^{-1})\right)^{-1/2}\rightarrow\infty,\label{eq: known var power eq 6}
\end{equation}
where the last step follows by (\ref{eq: known var power eq 0}).
The desired result follows by (\ref{eq: known var power eq 6}), (\ref{eq: known var power eq 5})
and (\ref{eq: known var power eq 4}), together with Slutsky's lemma. 
\end{proof}

\subsection{Proof of Theorem \ref{thm: unknown variance size X}}

In the rest of the article, we recall the definitions from Section \ref{sec: extension}:
$z_{i}=a^\top x_{i}/(a^\top a)$, $w_{i}=(I_{p}-aa^\top/(a^\top a))x_{i}$, $\pi_{*}=U_{a}^\top \beta_{*}$
and $\tilde{w}_{i}=U_{a}^\top w_{i}$. 

We need to derive some auxiliary results before we can prove Theorem
\ref{thm: unknown variance size X}. The proof of the following lemma
is similar to that of Theorem 7.1 of \citet{bickel2009simultaneous}
and   thus is omitted. 
\begin{lem}
\label{lem: DS1}Let $Y\in\mathbb{R}^{n}$ and $X\in\mathbb{R}^{n\times p}$.
Let $\hat{\xi}$ be any vector satisfying $\|n^{-1}X^{\top}(Y-X\hat{\xi})\|_{\infty}\leq\eta$.
Suppose that there exists $\xi_{*}$ such that $\|n^{-1}X^{\top}(Y-X\xi_{*})\|_{\infty}\leq\eta$
and $\|\hat{\xi}\|_{1}\leq\|\xi_{*}\|_{1}$. If $s_{*}=\|\xi_{*}\|_{0}$
and 
\begin{equation} \label{eq: RE condition X}
\min_{J_{0}\subseteq\{1,\cdots,p\},|J_{0}|\leq s_{*}}\min_{\delta\neq0,\|\delta_{J_{0}^{c}}\|_{1}\leq\|\delta_{J_{0}}\|_{1}}\frac{\|X\delta\|_{2}}{\sqrt{n}\|\delta_{J_{0}}\|_{2}}\geq\kappa,
\end{equation}
then $\|\delta\|_{1}\leq8\eta s_{*}\kappa^{-2}$ and $\delta^{\top}X^{\top}X\delta/n\leq16\eta^{2}s_{*}\kappa^{-2}$,
where $\delta=\hat{\xi}-\xi_{*}$. 
\end{lem}

\begin{lem}
\label{lem: feasibility null hypo} Suppose that Assumption \ref{assu: regularity condition}
and $H_{0}$ in (\ref{eq: null hypo}) hold. Consider the optimization
problem \eqref{eq: dantzig pi}. Let $v_{i}=y_{i}-z_{i}g_{0}$,
$\sigma_{v}^{2}=Ev_{i}^{2}$ and $\rho_{*}=\sigma_{\varepsilon}/\sigma_{v}$.
There exists a constant $C>0$, such that for any $\eta,\lambda>C\sqrt{n^{-1}\log p}$,
$\rho_{0}\leq[1+c_{2}c_{1}^{-1}(c_{3}^{-1}-1)]^{-1/2}$, we have
\[
P\left((\pi_{*},\rho_{*})\ {\rm and}\ \gamma_{*}\ {\rm are\ in\ the\ feasible\ region\ in\  }\ \eqref{eq: dantzig pi}\right)\rightarrow1.
\]
\end{lem}
\begin{proof}
Let $V=Y-Zg_{0}$ and notice that under Assumption \ref{assu: regularity condition},
$Z-\tilde{W}\gamma_{*}=u$ and $E\tilde{w}_{i}u_{i}=0$. Since $u_{i}\sigma_{u}^{-1}\sim \mathcal N(0,1)$
is independent of $\tilde{w}_{i}\sim \mathcal N(0,\Sigma_{\tilde{W}})$ with
the eigenvalues of $\Sigma_{\tilde{W}}=E\tilde{W}^\top \tilde{W}/n$ bounded
away from zero and infinity, it follows that there exists a constant
that upper bounds the sub-exponential norm of each entry of $\tilde{w}_{i}u_{i}\sigma_{u}^{-1}$.
To see this, note  that, by the moment generating function of $\mathcal N(0,1)$,
for $t>0$, 
$$E\exp(t\tilde{w}_{i,j}u_{i}\sigma_{u}^{-1})=E[E(\exp(t\tilde{w}_{i,j}u_{i}\sigma_{u}^{-1})\mid\tilde{w}_{i,j})]=E\exp(\tilde{w}_{i,j}^{2}t^{2}/2).$$
Since $\tilde{w}_{i,j}^{2}$ has bounded sub-exponential norm (by
Lemma 5.14 of \citet{vershynin2010introduction}), Lemma 5.15 of \citet{vershynin2010introduction}
implies that for small enough $t$, $E\exp(t\tilde{w}_{i,j}u_{i}\sigma_{u}^{-1})=E\exp(\tilde{w}_{i,j}^{2}t^{2}/2)$
is bounded by some constant. Hence, Equation (5.16) in \citet{vershynin2010introduction}
implies that $\tilde{w}_{i,j}u_{i}\sigma_{u}^{-1}$ has bounded the sub-exponential
norm. 

By Proposition 5.16 of \citet{vershynin2010introduction} and the
union bound, we have that $\forall t_{0}>0$, 
\[
P\left(\|n^{-1}\tilde{W}^\top u\sigma_{u}^{-1}\|_{\infty}>t_{0}\sqrt{n^{-1}\log p}\right)\leq2p\exp\left[-\min\left(\frac{t_{0}^{2}\log p}{K^{2}},\frac{t_{0}\sqrt{n\log p}}{K}\right)\right],
\]
where $K>0$ is a constant depending only on the constants in Assumption
\ref{assu: regularity condition}. Hence, there exists a constant
$M_{1}>0$, such that $P\left(\|n^{-1}\tilde{W}^\top u\|_{\infty}>M_{1}\sigma_{u}\sqrt{n^{-1}\log p}\right)\rightarrow0$.
It follows that
\begin{multline}
P\left(\|n^{-1}\tilde{W}^\top (Z-\tilde{W}\gamma_{*})\|_{\infty}>2c_{3}^{-1/2}M_{1}\sqrt{n^{-1}\log p}n^{-1/2}\|Z\|_{2}\right)\\
\leq P\left(\|n^{-1}\tilde{W}^\top u\|_{\infty}>M_{1}\sigma_{u}\sqrt{n^{-1}\log p}\right)+P\left(2\frac{\sqrt{c_{3}}\sigma_{u}}{\sigma_{z}}\geq\frac{n^{-1/2}\|Z\|_{2}}{\sigma_{z}}\right)\overset{(i)}{=}o(1),\label{eq: feasibility eq 1}
\end{multline}
where $(i)$ follows by $2\sqrt{c_{3}}\sigma_{u}/\sigma_{z}\geq2$
(Assumption \ref{assu: regularity condition}) and $n^{-1/2}\|Z\sigma_{z}^{-1}\|_{2}=1+o_{P}(1)$.
By the Law of Large Numbers : $n^{-1}\|Z\sigma_{z}^{-1}\|_{2}^{2}$ is the average of $n$
independent $\chi^{2}(1)$ random variables. 

Notice that under $H_{0}$ in (\ref{eq: null hypo}), $V-\tilde{W}\pi_{*}=\varepsilon$.
By an analogous argument, there exists a constant $M_{3}>0$ such
that
\begin{equation}
P\left(\|n^{-1}\tilde{W}^\top (V-\tilde{W}\pi_{*})\|_{\infty}>M_{3}\rho_{*}\sqrt{n^{-1}\log p}n^{-1/2}\|V\|_{2}\right)\rightarrow0.\label{eq: feasibility eq 2}
\end{equation}
Since $V=\tilde{W}\pi_{*}+\varepsilon=W\beta_{*}+\varepsilon$, we
have that $\sigma_{v}^{2}=\beta_{*}^\top \Sigma_{W}\beta_{*}+\sigma_{\varepsilon}^{2}$.
Assumption \ref{assu: regularity condition} implies that 
$$\beta_{*}^\top \Sigma_{X}\beta_{*}+\sigma_{\varepsilon}^{2}=\sigma_{y}^{2}\leq\sigma_{\varepsilon}^{2}/c_{3}$$
and thus $\beta_{*}^\top \Sigma_{X}\beta_{*}\leq(c_{3}^{-1}-1)\sigma_{\varepsilon}^{2}$.
Therefore, 
$$\|\beta_{*}\|_{2}^{2}\leq(\beta_{*}^\top \Sigma_{X}\beta_{*})/\lambda_{\min}(\Sigma_{X})\leq c_{1}^{-1}\beta_{*}^\top \Sigma_{X}\beta_{*}\leq c_{1}^{-1}(1-c_{3}^{-1})\sigma_{\varepsilon}^{2}.$$
Observe that $\Sigma_{W}=M_{a}\Sigma_{X}M_{a}$, where $M_{a}=I_{p}-aa^\top/(a^\top a)$
is a projection matrix. Hence, $\lambda_{\max}(\Sigma_{W})\leq\lambda_{\max}(\Sigma_{X})$
and thus 
$$\sigma_{v}^{2}=\beta_{*}^\top \Sigma_{W}\beta_{*}+\sigma_{\varepsilon}^{2}\leq\lambda_{\max}(\Sigma_{X})\|\beta_{*}\|_{2}^{2}+\sigma_{\varepsilon}^{2}\leq[1+c_{2}c_{1}^{-1}(c_{3}^{-1}-1)]\sigma_{\varepsilon}^{2}.$$
Since $\rho_{0}\leq[1+c_{2}c_{1}^{-1}(c_{3}^{-1}-1)]^{-1/2}\leq\sigma_{\varepsilon}/\sigma_{v}$,
we have that $\sigma_{\varepsilon}^{2}=\rho_{*}\sigma_{v}\sigma_{\varepsilon}\geq\sigma_{v}^{2}\rho_{*}\rho_{0}$.
By the Law of Large Numbers , 
$$n^{-1}V^\top (V-\tilde{W}\pi_{*})=n^{-1}V^\top \varepsilon=(1+o_{P}(1))\sigma_{\varepsilon}^{2},$$
$n^{-1}\|V\|_{2}^{2}=(1+o_{P}(1))\sigma_{v}^{2}$ and thus
\begin{equation}
P\left(\frac{n^{-1}V^\top (V-\tilde{W}\pi_{*})}{\rho_{0}\rho_{*}n^{-1}\|V\|_{2}^{2}}\geq\frac{1}{2}\right)=P\left(\frac{\sigma_{\varepsilon}^{2}(1+o_{P}(1))}{\sigma_{v}^{2}(1+o_{P}(1))\rho_{0}\rho_{*}}\geq\frac{1}{2}\right)\rightarrow1.\label{eq: feasibility eq 3}
\end{equation}
The desired result follows by (\ref{eq: feasibility eq 1}), (\ref{eq: feasibility eq 2}),
(\ref{eq: feasibility eq 3}) and the fact that 
$$\rho_{*}=\sigma_{\varepsilon}\sigma_{v}^{-1}\geq[1+c_{2}c_{1}^{-1}(c_{3}^{-1}-1)]^{-1/2}\geq\rho_{0}.$$ 
\end{proof}

\begin{lem}
\label{lem: nondegeneracy}If $n^{-1}V^{\top}(V-\tilde{W}\hat{\pi})\geq\bar{\eta}$,
then $n^{-1}(V-\tilde{W}\hat{\pi})^{\top}(V-\tilde{W}\hat{\pi})\geq\bar{\eta}^{2}/(n^{-1}V^{\top}V).$ \end{lem}
\begin{proof}
Since $n^{-1}V^{\top}(V-\tilde{W}\hat{\pi})\geq\bar{\eta}$, we have
that, for any $t\geq0$,
\begin{align*}
n^{-1}(V-\tilde{W}\hat{\pi})^{\top}(V-\tilde{W}\hat{\pi})&\geq n^{-1}(V-\tilde{W}\hat{\pi})^{\top}(V-\tilde{W}\hat{\pi})+t\left(\bar{\eta}-n^{-1}V^{\top}(V-\tilde{W}\hat{\pi})\right)\\
&\overset{(i)}{\geq}\min_{\gamma}\left\{ n^{-1}(V-\tilde{W}\gamma)^{\top}(V-\tilde{W}\gamma)+t\left(\bar{\eta}-n^{-1}V^{\top}(V-\tilde{W}\gamma)\right)\right\} 
\\
&=t\bar{\eta}-\frac{1}{4}t^{2}n^{-1}V^{\top}V,
\end{align*}
where $(i)$ follows by the first-order condition of quadratic optimizations.
The desired result follows by maximizing the last line with respect
to $t$ with $t=2\bar{\eta}/\left(n^{-1}V^{\top}V\right)$.
\end{proof}
We now proceed to prove Theorem \ref{thm: unknown variance size X}. 
\begin{proof}[\textbf{Proof of Theorem \ref{thm: unknown variance size X}}]
 Let $V=Y-Zg_{0}$, $s_{*}=\|\gamma_{*}\|_{0}$, $\eta_{\pi}=\eta n^{-1/2}\|V\|_{2}$
and $\lambda_{\gamma}=\lambda n^{-1/2}\|Z\|_{2}$. Notice that 
\begin{equation}
n^{-1/2}(V-\tilde{W}\hat{\pi})^{\top}(Z-\tilde{W}\hat{\gamma})=\underbrace{n^{-1/2}(V-\tilde{W}\hat{\pi})^{\top}u}_{I_{1}}+\underbrace{n^{-1/2}(V-\tilde{W}\hat{\pi})^{\top}\tilde{W}(\gamma_{*}-\hat{\gamma})}_{I_{2}}.\label{eq:size single hypo eq 1}
\end{equation}
Since the eigenvalues of $E\tilde{W}^\top \tilde{W}/n$ is bounded away
from zero and infinity, it follows, as a simple consequence of Theorem
6 in \citet{rudelson2013reconstruction}, that there exists a constant
$\kappa>0$, such that $P(\mathcal{D}_{n}(s_{*},\kappa))\to1$, where
\begin{equation}
\label{eq: RE condition tilde W}
\mathcal{D}_{n}(s_{*},\kappa)=\left\{ \min_{J_{0}\subseteq\{1,\cdots,p\},|J_{0}|\leq s_{*}}\min_{\delta\neq0,\|\delta_{J_{0}^{c}}\|_{1}\leq\|\delta_{J_{0}}\|_{1}}\frac{\|\tilde{W}\delta\|_{2}}{\sqrt{n}\|\delta_{J_{0}}\|_{2}}>\kappa\right\} .
\end{equation}
Define the event $\mathcal{M}=\Bigl\{(\pi_{*},\rho_{*})\ {\rm and}\ \gamma_{*}\ {\rm are\ in\ the\ feasible\ region\ in\  }\ \eqref{eq: dantzig pi} \Bigl\}$.
By Lemma \ref{lem: feasibility null hypo}, with appropriate choice
of tuning parameters as specified in the theorem, we have $P\left(\mathcal{M}\right)\rightarrow1$
and thus 
\begin{equation}
P\left(\mathcal{M}\bigcap\mathcal{D}_{n}(s_{*},\kappa)\right)\rightarrow1.\label{eq: size single hypo eq 2}
\end{equation}
We apply Lemma \ref{lem: DS1} with $(Y,X,\xi_{*})$ replaced by $(Z,\tilde{W},\gamma_{*})$
and obtain that, on the event $\mathcal{M}\bigcap\mathcal{D}_{n}(s_{*},\kappa)$,
\begin{equation}
\|\hat{\gamma}-\gamma_{*}\|_{1}\leq8\lambda_{\gamma}s_{*}\kappa^{-2}\ {\rm and}\ n^{-1/2}\|\tilde{W}(\hat{\gamma}-\gamma_{*})\|_{2}\leq4\lambda_{\gamma}\sqrt{s_{*}}\kappa^{-1}.\label{eq: size single hypo eq 3}
\end{equation}
Thus, on $\mathcal{M}\bigcap\mathcal{D}_{n}(s_{*},\kappa)$, we have
the bound 
$$|I_{2}|\leq n^{1/2}\|n^{-1}\tilde{W}^{\top}(V-\tilde{W}\hat{\pi})\|_{\infty}\|\hat{\gamma}-\gamma_{*}\|_{1}\leq8n^{1/2}\lambda_{\gamma}\eta_{\pi}s_{*}\kappa^{-2},$$
where in the last step we utilized 
$$\|n^{-1}\tilde{W}^{\top}(V-\tilde{W}\hat{\pi})\|_{\infty}\leq\eta\hat{\rho}n^{-1/2}\|V\|_{2}\leq\eta_{\pi}$$
with $\hat{\rho}\leq1$, from the constraints in optimization problem
(\ref{eq: dantzig pi}). Moreover, by constraints in (\ref{eq: dantzig pi})
and Lemma \ref{lem: nondegeneracy}, on $\mathcal{M}\bigcap\mathcal{D}_{n}(s_{*},\kappa)$,
we have that 
$$\hat{\sigma}_{\varepsilon}\geq\rho_{0}\hat{\rho}n^{-1/2}\|V\|_{2}/2\geq\rho_{0}^{2}n^{-1/2}\|V\|_{2}/2$$
and thus, by $\sigma_{u}\geq c_{3}\sigma_{z}$, 
\begin{equation}
\left|\frac{I_{2}}{\hat{\sigma}_{\varepsilon}\sigma_{u}}\right|\leq\frac{16n\lambda_{\gamma}\eta_{\pi}s_{*}\kappa^{-2}}{c_{3}\sigma_{z}\rho_{0}^{2}\|V\|_{2}}=\frac{16\sqrt{n}\lambda\eta s_{*}\kappa^{-2}}{c_{3}\rho_{0}^{2}}\times\frac{n^{-1/2}\|Z\|_{2}}{\sigma_{z}}\overset{(i)}{=}o_{P}(1),\label{eq: size single hypo eq 4}
\end{equation}
where $(i)$ follows by $\rho_{0}^{-1}=O(1)$ and $\lambda,\eta\asymp\sqrt{n^{-1}\log p}$
with $s_{*}=o(\sqrt{n}/\log p)$ and $n^{-1/2}\|Z\|_{2}/\sigma_{z}=1+o_{P}(1)$.
Observe that by the Law of Large Numbers : $n^{-1}\|Z\sigma_{z}^{-1}\|_{2}^{2}$ is the average of $n$
independent $\chi^{2}(1)$ random variables.

For $I_{1}$, notice that under $H_{0}$ in (\ref{eq: null hypo}),
$u$ is independent of $\{V,\tilde{W}\}$. Since $\hat{\pi}$ and
$\hat{\sigma}_{\varepsilon}$ are computed using $\{V,\tilde{W}\}$,
it follows that $u$ is independent of $V-\tilde{W}\hat{\pi}$ and
$\hat{\sigma}_{\varepsilon}$. Thus, under $H_{0}$, 
$$I_{1}\hat{\sigma}_{\varepsilon}^{-1}\sigma_{u}^{-1}\mid(V,\tilde{W})\sim \mathcal N(0,1)$$
and thus $I_{1}\hat{\sigma}_{\varepsilon}^{-1}\sigma_{u}^{-1}\sim \mathcal N(0,1)$.
This, together with (\ref{eq: size single hypo eq 4}), implies that,
under $H_{0}$, 
\begin{equation}
S_{n}\frac{\hat{\sigma}_{u}}{\sigma_{u}}=\frac{n^{-1/2}(V-\tilde{W}\hat{\pi})^{\top}(Z-\tilde{W}\hat{\gamma})}{\hat{\sigma}_{\varepsilon}\sigma_{u}}\rightarrow^{d}\mathcal N(0,1).\label{eq: size single hypo eq 5}
\end{equation}

By (\ref{eq: size single hypo eq 3}) and $s_{*}=o(\sqrt{n}/\log p)$,
\begin{align}
\left|\hat{\sigma}_{u}-n^{-1/2}\|u\|_{2}\right|
&=\left|n^{-1/2}\|Z-\tilde{W}\hat{\gamma}\|_{2}-n^{-1/2}\|u\|_{2}\right|
\\
&\leq n^{-1/2}\|Z-\tilde{W}\hat{\gamma}-u\|_{2}
\\
&=n^{-1/2}\|\tilde{W}(\hat{\gamma}-\gamma_{*})\|_{2}
\\
&=O_{P}(\lambda n^{-1/2}\|Z\|_{2}\sqrt{s_{*}})=o_{P}(n^{-3/4}\|Z\|_{2}).
\end{align}
Therefore, 
\begin{align}
\frac{\left|\hat{\sigma}_{u}-\sigma_{u}\right|}{\sigma_{u}}
& \leq\frac{\left|\hat{\sigma}_{u}-n^{-1/2}\|u\|_{2}\right|}{\sigma_{u}}+\left|n^{-1/2}\|u\sigma_{u}^{-1}\|_{2}-1\right|
\\
&\overset{(i)}{=}\frac{o_{P}(n^{-3/4}\|Z\|_{2})}{\sigma_{u}}+o_{P}(1)
\\
&\overset{(ii)}{=}o_{P}(1),\label{eq: size single hypo eq 6}
\end{align}
where $(i)$ follows by the Law of Large Numbers  ($n^{-1}\|u\sigma_{u}^{-1}\|_{2}^{2}$
is the average of $n$ independent $\chi^{2}(1)$ random variables)
and $(ii)$ follows by $\sigma_{z}/\sigma_{u}\leq c_{3}^{-1}$ (Assumption
\ref{assu: regularity condition}) and $n^{-1/2}\|Z\|_{2}/\sigma_{z}=1+o_{P}(1)$
(as argued in (\ref{eq: size single hypo eq 4})).

By (\ref{eq: size single hypo eq 6}), $\hat{\sigma}_{u}/\sigma_{u}=1+o_{P}(1)$
and the desired result follows by (\ref{eq: size single hypo eq 5})
and Slutsky's lemma. 
\end{proof}

\subsection{Proof of Theorem \ref{thm: power unknown SigmaX}}

We need the some auxiliary results before we prove Theorem \ref{thm: power unknown SigmaX}.
\begin{lem}
\label{lem: sigma u characterization}Let Assumption \ref{assu: regularity condition}
hold. In (\ref{eq: graphical model}), $\sigma_{u}=(a^\top\Omega_Xa)^{-1/2}$. \end{lem}
\begin{proof}
We define $\tilde{\gamma}_{*}=U_{a}\gamma_{*}\in\mathbb{R}^{p}$ and
observe that $z_{i}=\tilde{w}_{i}^\top \gamma_{*}+u_{i}=w_{i}^\top \tilde{\gamma}_{*}+u_{i}$
and $Ew_{i}u_{i}=U_{a}^\top E\tilde{w}_{i}u_{i}=0$. Thus, $\sigma_{u}^{2}=Ez_{i}^{2}-\tilde{\gamma}_{*}^\top Ew_{i}w_{i}^\top \tilde{\gamma}_{*}$. 

Let $M_{a}=I_{p}-aa^\top/(a^\top a)$. Recall that $z_{i}=a^\top x_{i}/(a^\top a)$ and
$w_{i}=M_{a}x_{i}$. Since $Ew_{i}z_{i}=Ew_{i}w_{i}^\top \tilde{\gamma}_{*}$,
$$Ew_{i}w_{i}^\top =M_{a}\Sigma_{X}M_{a}$$
 and $Ew_{i}z_{i}=M_{a}\Sigma_{X}a\|a\|_{2}^{-2}$,
we have 
$$M_{a}\Sigma_{X}M_{a}\tilde{\gamma}_{*}=M_{a}\Sigma_{X}a\|a\|_{2}^{-2}$$
and thus, $M_{a}\Sigma_{X}(M_{a}\tilde{\gamma}_{*}-a\|a\|_{2}^{-2})=0$.
Since $M_{a}$ is the projection matrix onto the $(p-1)$-dimensional
linear space orthogonal to $a$, there exists $k_{1}\in\mathbb{R}$
with 
$$\Sigma_{X}(M_{a}\tilde{\gamma}_{*}-a\|a\|_{2}^{-2})=k_{1}a,$$
implying that $M_{a}\tilde{\gamma}_{*}=k_{1}\Omega_Xa+\|a\|_{2}^{-2}a$.
Next we aim to identify $k_1$.
Observe that
\begin{align*}
\tilde{\gamma}_{*}^\top M_{a}\Sigma_{X}M_{a}\tilde{\gamma}_{*}&\overset{(i)}{=}(M_{a}\tilde{\gamma}_{*})^\top (M_{a}\Sigma_{X}M_{a}\tilde{\gamma}_{*})
\\
&\overset{(ii)}{=}(k_{1}\Omega_Xa+\|a\|_{2}^{-2}a)^\top \Sigma_{X}a\|a\|_{2}^{-2}
\\
&=k_{1}+\|a\|_{2}^{-4}a^\top\Sigma_{X}a.
\end{align*}
where $(i)$ and $(ii)$ follow by $M_{a}^{2}=M_{a}$ and $M_{a}\Sigma_{X}M_{a}\tilde{\gamma}_{*}=M_{a}\Sigma_{X}a\|a\|_{2}^{-2}$,
respectively.
Together with
\begin{align*}
\tilde{\gamma}_{*}^\top M_{a}\Sigma_{X}M_{a}\tilde{\gamma}_{*}
&=(M_{a}\tilde{\gamma}_{*})^\top \Sigma_{X}(M_{a}\tilde{\gamma}_{*})
\\
&=(k_{1}\Omega_Xa+\|a\|_{2}^{-2}a)^\top \Sigma_{X}(k_{1}\Omega_Xa+\|a\|_{2}^{-2}a),
\end{align*}
we can solve for the unknown $k_1$. The above display allows us to obtain $k_{1}=-(a^\top\Omega_Xa)^{-1}$
and thus 
$$\tilde{\gamma}_{*}^\top M_{a}\Sigma_{X}M_{a}\tilde{\gamma}_{*}=\|a\|_{2}^{-4}a^\top\Sigma_{X}a-(a^\top\Omega_Xa)^{-1}.$$
Since 
$$\sigma_{u}^{2}=Ez_{i}^{2}-\tilde{\gamma}_{*}^\top Ew_{i}w_{i}^\top \tilde{\gamma}_{*}=Ez_{i}^{2}-\tilde{\gamma}_{*}^\top M_{a}\Sigma_{X}M_{a}\tilde{\gamma}_{*}$$
and $Ez_{i}^{2}=\|a\|_{2}^{-4}a^\top\Sigma_{X}a$, we have $\sigma_{u}^{2}=(a^\top\Sigma_{X}a)^{-1}$.
The proof is complete. \end{proof}
\begin{lem}
\label{lem: pi sparsity}Let Assumption \ref{assu: regularity condition}
hold. Suppose that at least one of the following conditions holds:\\
(1) $\|a\|_{0}\vee\|\beta_{*}\|_{0}=o(\sqrt{n}/\log p)$ or \\
(2) $\mathcal{M}(a)\bigcap\mathcal{M}(\beta_{*})=\emptyset$ and $\|\beta_{*}\|_{0}=o(\sqrt{n}/\log p)$. 

Then, $\|\pi_{*}\|_{0}=o(\sqrt{n}/\log p)$. \end{lem}
\begin{proof}
We denote $s_{a}=\|a\|_{0}$ and $s_{\beta}=\|\beta_{*}\|_{0}$. Without
loss of generality, we assume that $a=(a_{0}^\top ,0)^\top \in\mathbb{R}^{p}$
with $\|a_{0}\|_{0}=s_{a}$. Let $U_{a_{0}}\in\mathbb{R}^{s_{a}\times(s_{a}-1)}$
satisfy $U_{a_{0}}^\top U_{a_{0}}=I_{s_{a}-1}$ and $U_{a_{0}}U_{a_{0}}^\top =I_{s_{a}}-a_{0}a_{0}^\top /(a_{0}^\top a_{0})$.
It is easy to verify that 
\[
I_{p}-aa^\top/(a^\top a)=\begin{pmatrix}I_{s_{a}}-a_{0}a_{0}^\top /(a_{0}^\top a_{0}) & 0\\
0 & I_{p-s_{a}}
\end{pmatrix}\quad{\rm and}\quad U_{a}=\begin{pmatrix}U_{a_{0}} & 0\\
0 & I_{p-s_{a}}
\end{pmatrix}\in\mathbb{R}^{p\times(p-1)}.
\]

It then follows that 
\begin{equation}
\pi_{*}=U_{a}^\top \beta_{*}=\begin{pmatrix}U_{a_{0}}^\top \beta_{*,\mathcal{M}(a)}\\
\beta_{*,\left[\mathcal{M}(a)\right]^{c}}
\end{pmatrix}\label{eq: pi sparsity eq 1}
\end{equation}

Take note  that 
$$\|\pi_{*}\|_{0}=\|U_{a_{0}}^\top \beta_{*,\mathcal{M}(a)}\|_{0}+\|\beta_{*,\left[\mathcal{M}(a)\right]^{c}}\|_{0}\leq(s_{a}-1)+\|\beta_{*}\|_{0}.$$
This proves the result under condition (1). Under condition (2), $\beta_{*,\mathcal{M}(a)}=0$
and thus $\|\pi_{*}\|_{0}=\|\beta_{*,\left[\mathcal{M}(a)\right]^{c}}\|_{0}=\|\beta_{*}\|_{0}$.
This proves the result under condition (2). \end{proof}
\begin{lem}
\label{lem: feasibility power}Let Assumption \ref{assu: regularity condition}
and $H_{1,n}$ in (\ref{eq: local alternative}) hold. Let $v_{i}=y_{i}-z_{i}g_{0}$,
$\sigma_{v}^{2}=Ev_{i}^{2}$, $\bar{\rho}_{*}=[1+c_{2}c_{1}^{-1}(c_{3}^{-1}-1)]^{-1/2}$
and $h_{n}=n^{-1/2}(a^\top\Omega_Xa)^{1/2}\sigma_{\varepsilon}d$.
Then, 
$$\sigma_{v}=O(1).$$ Moreover, there exists a constant $C>0$, such
that for any $\eta,\lambda>C\sqrt{n^{-1}\log p}$, $\rho_{0}\leq[1+c_{2}c_{1}^{-1}(c_{3}^{-1}-1)]^{-1/2}$,
we have
\[
P\left((\pi_{*}+\gamma_{*}h_{n},\bar{\rho}_{*},\gamma_{*})\ {\rm is\ feasible\ in\  \ \eqref{eq: dantzig pi}}\right)\rightarrow1.
\]
\end{lem}
\begin{proof}
Under $H_{1,n}$, $v_{i}=\tilde{w}_{i}^\top \pi_{*}+\varepsilon_{i}+z_{i}h_{n}$.
Consequently, 
$$\|v_{i}-\tilde{w}_{i}^\top \pi_{*}-\varepsilon_{i}\|_{L^{2}(P)}=\|z_{i}h_{n}\|_{L^{2}(P)}=\sqrt{Ez_{i}^{2}h_{n}^{2}}.$$
Observe that 
\begin{align}
Ez_{i}^{2}h_{n}^{2}
&=\|a\|_{2}^{-4}a^\top\Sigma_{X}ah_{n}^{2}
\\
&=\|a\|_{2}^{-4}(a^\top\Sigma_{X}a)(a^\top\Omega_Xa)(a^\top\Omega_Xa)^{-1}h_{n}^{2}
\\
&\overset{(i)}{\leq}(c_{2}c_{1}^{-1})(a^\top\Omega_Xa)^{-1}h_{n}^{2}=(c_{2}c_{1}^{-1})n^{-1}\sigma_{\varepsilon}^{2}d^{2}=o(1),
\end{align}
where $(i)$ holds by Assumption \ref{assu: regularity condition}.
Hence, by the triangular inequality applied to $L^{2}(P)$-norm, we
have $\sigma_{v}=\|v_{i}\|_{L^{2}(P)}=\|\tilde{w}_{i}^\top \pi_{*}+\varepsilon_{i}\|_{L^{2}(P)}+o(1)$.
By the same argument as in the proof of Lemma \ref{lem: feasibility null hypo},
$$\|\tilde{w}_{i}^\top \pi_{*}+\varepsilon_{i}\|_{L^{2}(P)}^{2}=\pi_{*}^\top E\tilde{w}_{i}\tilde{w}_{i}^\top \pi_{*}+\sigma_{\varepsilon}=\beta_{*}^\top \Sigma_{W}\beta_{*}+\sigma_{\varepsilon}^{2}\leq[1+c_{2}c_{1}^{-1}(c_{3}^{-1}-1)]\sigma_{\varepsilon}^{2}.$$
The first claim follows by $\sigma_{v}\leq[1+c_{2}c_{1}^{-1}(c_{3}^{-1}-1)]^{1/2}\sigma_{\varepsilon}+o(1)=O(1)$. 

Notice that under $H_{1,n}$, the analysis for the feasibility of
$\gamma_{*}$ is the same as under $H_{0}$. Thus, by the argument
in the proof of Lemma \ref{lem: feasibility null hypo}, for some
constant $M_{1}>0$, we have 
\begin{equation}
P\left(\|n^{-1}\tilde{W}^\top (Z-\tilde{W}\gamma_{*})\|_{\infty}>M_{1}\sqrt{n^{-1}\log p}n^{-1/2}\|Z\|_{2}\right)\rightarrow0.\label{eq: feasibility power eq 1}
\end{equation}

(\ref{eq: graphical model}) implies that, under $H_{1,n}$, $v_{i}=\tilde{w}_{i}^\top \pi_{*}+\varepsilon_{i}+z_{i}h_{n}=\tilde{w}_{i}^\top (\pi_{*}+\gamma_{*}h_{n})+\varepsilon_{i}+u_{i}h_{n}$.
Thus, 
$$n^{-1}\tilde{W}^\top (V-\tilde{W}(\pi_{*}+\gamma_{*}h_{n}))=n^{-1}\sum_{i=1}^{n}\tilde{w}_{i}\varepsilon_{i}+n^{-1}\sum_{i=1}^{n}\tilde{w}_{i}u_{i}h_{n}.$$
By a similar argument as in the proof of Lemma \ref{lem: feasibility null hypo},
entries of $\tilde{w}_{i}\varepsilon_{i}$ and $\tilde{w}_{i}u_{i}\sigma_{u}^{-1}$
have bounded sub-exponential norms. As in the proof of Lemma \ref{lem: feasibility null hypo},
we can use Proposition 5.16 of \citet{vershynin2010introduction}
and the union bound to conclude that for some constant $M_{2}>0$
we have $P(\|n^{-1}\tilde{W}^\top u\sigma_{u}^{-1}\|_{\infty}>M_{2}\sqrt{n^{-1}\log p})\rightarrow0$
and $P(\|n^{-1}\tilde{W}^\top \varepsilon\|_{\infty}>M_{2}\sqrt{n^{-1}\log p})\rightarrow0$.
It follows that 
\begin{eqnarray}
 &  & P\left(\|n^{-1}\tilde{W}^\top (V-\tilde{W}(\pi_{*}+\gamma_{*}h_{n}))\|_{\infty}>2M_{2}\sqrt{n^{-1}\log p}\right)\label{eq: feasibility power eq 2}\\
 & = & P\left(\|n^{-1}\tilde{W}^\top (\varepsilon+uh_{n})\|_{\infty}>2M_{2}\sqrt{n^{-1}\log p}\right)\nonumber \\
 & \leq & P\left(\|n^{-1}\tilde{W}^\top \varepsilon\|_{\infty}>M_{2}\sqrt{n^{-1}\log p}\right)\nonumber
 \\
 &+&P\left(\|n^{-1}\tilde{W}^\top u\sigma_{u}^{-1}\|_{\infty}|\sigma_{u}h_{n}|>M_{2}\sqrt{n^{-1}\log p}\right)\overset{(i)}{=}o(1),\nonumber 
\end{eqnarray}
where $(i)$ holds by $|\sigma_{u}h_{n}|=n^{-1/2}\sigma_{\varepsilon}|d|=o(1)$
(by Lemma \ref{lem: sigma u characterization} and the definition
of $h_{n}$).

Notice that 
$$Ev_{i}(u_{i}h_{n}+\varepsilon_{i})\sigma_{\varepsilon}^{-2}=E(u_{i}h_{n}+\varepsilon_{i})^{2}\sigma_{\varepsilon}^{-2}=1+\sigma_{\varepsilon}^{-2}\sigma_{u}^{2}h_{n}^{2}=1+n^{-1}d^{2}=1+o(1).$$
By the Law of Large Numbers , 
$$n^{-1}V^\top (V-\tilde{W}(\pi_{*}+\gamma_{*}h_{n}))\sigma_{\varepsilon}^{-2}
=
E(u_{i}h_{n}+\varepsilon_{i})^{2}\sigma_{\varepsilon}^{-2}+o_{P}(1)=1+o_{P}(1).$$
In the display above, the first \(o_P(1)\) term is equal to \(n^{-1}\sigma_\varepsilon^{-2}(\pi_*+\gamma_*h_n)^\top \tilde{W}^\top (\varepsilon+h_nu)\). Since \(\tilde{W}\) is uncorrelated with \((\varepsilon,u)\), this term is the partial sum of zero-mean independent random variables. Since \(\pi_*+h_n\gamma_* \) has bounded \(L^2\)-norm  by Bernstein's inequality, we have  that this term is \(o_P(1)\).
The Law of Large Numbers also implies that $n^{-1}\|V\|_{2}^{2}\sigma_{v}^{-2}=1+o_{P}(1)$.
Hence, 
\begin{align}
&P\left(\frac{n^{-1}V^\top (V-\tilde{W}(\pi_{*}+\gamma_{*}h_{n}))}{n^{-1}\|V\|_{2}^{2}}>\frac{1}{2}\rho_{0}\bar{\rho}_{*}\right)
\\
&=P\left(\frac{\sigma_{\varepsilon}^{2}(1+o_{P}(1))}{\sigma_{v}^{2}(1+o_{P}(1))}>\frac{1}{2}\rho_{0}\bar{\rho}_{*}\right)\\
&\overset{(i)}{\geq}P\left(\frac{(1+o_{P}(1))}{[1+c_{2}c_{1}^{-1}(c_{3}^{-1}-1)](1+o_{P}(1))}>\frac{1}{2}\rho_{0}\bar{\rho}_{*}\right)\overset{(ii)}{\geq}1+o(1),\label{eq: feasibility power eq 3}
\end{align}
where $(i)$ follows by $\sigma_{v}\leq[1+c_{2}c_{1}^{-1}(c_{3}^{-1}-1)]^{1/2}\sigma_{\varepsilon}+o(1)$
(shown at the beginning of the proof) and $(ii)$ follows by $\rho_{0}\leq\bar{\rho}_{*}=[1+c_{2}c_{1}^{-1}(c_{3}^{-1}-1)]^{-1/2}$.
According to (\ref{eq: feasibility power eq 1}), (\ref{eq: feasibility power eq 2})
and (\ref{eq: feasibility power eq 3}), $P\left((\pi_{*}+\gamma_{*}h_{n},\bar{\rho}_{*},\gamma_{*})\ {\rm is\ feasible\ in\  \ \eqref{eq: dantzig pi}}\right)\rightarrow1.$
The proof is complete. 
\end{proof}
Now we are ready to prove Theorem \ref{thm: power unknown SigmaX}. 
\begin{proof}[\textbf{Proof of Theorem \ref{thm: power unknown SigmaX}}]
 Let $V=Y-Zg_{0}$, $s_{*}=\|\gamma_{*}\|_{0}+\|\pi_{*}\|_{0}$,
$h_{n}=n^{-1/2}(a^\top\Omega_Xa)^{1/2}\sigma_{\varepsilon}d$,
$\lambda_{\gamma}=\lambda n^{-1/2}\|Z\|_{2}$, $\eta_{\pi}=\eta n^{-1/2}\|V\|_{2}$
and $\sigma_{v}^{2}=EV^\top V/n$ . Notice that $\|\gamma_{*}\|_{0}\leq s_{*}$
and $\|\pi_{*}+\gamma_{*}h_{n}\|_{0}\leq s_{*}$. By Lemmas \ref{lem: sigma u characterization}
and \ref{lem: pi sparsity}, 
\begin{equation}
s_{*}=o(\sqrt{n}/\log p)\quad{\rm and}\quad h_{n}=n^{-1/2}\sigma_{u}^{-1}\sigma_{\varepsilon}d.\label{eq: power unknown SigmaX eq 0}
\end{equation}

Under $H_{1,n}$, $V=Zh_{n}+\tilde{W}\pi_{*}+\varepsilon=uh_{n}+\tilde{W}(\gamma_{*}h_{n}+\pi_{*})+\varepsilon$
and %
thus
\begin{align}
&n^{-1/2}(V-\tilde{W}\hat{\pi})^{\top}(Z-\tilde{W}\hat{\gamma}) \label{eq: power unknown SigmaX eq 0.5}
\\
&\qquad =n^{-1/2}(V-\tilde{W}\hat{\pi})^\top \tilde{W}(\gamma_{*}-\hat{\gamma})+n^{-1/2}(V-\tilde{W}\hat{\pi})^\top u\nonumber\\
&\qquad =\underbrace{n^{-1/2}(V-\tilde{W}\hat{\pi})^\top \tilde{W}(\gamma_{*}-\hat{\gamma})}_{I_{1}}+\underbrace{n^{-1/2}\varepsilon^\top u}_{I_{2}}\nonumber
\\ \nonumber
&\qquad \qquad \qquad +\underbrace{n^{-1/2}h_{n}u^\top u}_{I_{3}}+\underbrace{n^{-1/2}(\pi_{*}-\hat{\pi}+\gamma_{*}h_{n})^\top \tilde{W}^\top u}_{I_{4}}.
\end{align}
We next treat each of the four terms in the decomposition above separately.

As argued in the proof of Theorem \ref{thm: unknown variance size X},
there exists a constant $\kappa>0$ such that $P(\mathcal{D}_{n}(s_{*},\kappa))\to1$.
Define the event $\mathcal{M}=\{(\pi_{*}+\gamma_{*}h_{n},\bar{\rho}_{*},\gamma_{*})\ {\rm is\ feasible\ in\  \ \eqref{eq: dantzig pi}}\}$.
By Lemma \ref{lem: feasibility power}, $P(\mathcal{M})\rightarrow1$
and thus
\begin{equation}
P\left(\mathcal{M}\bigcap\mathcal{D}_{n}(s_{*},\kappa)\right)\rightarrow1.\label{eq: power unknown sigmaX eq 1}
\end{equation}

Since $\hat{\gamma}$ does not depend on whether $h_{n}=0$, we conclude,
as argued in the proof of Theorem \ref{thm: unknown variance size X},
that $\hat{\sigma}_{u}/\sigma_{u}=1+o_{P}(1)$ and that on the event
$\mathcal{M}\bigcap\mathcal{D}_{n}(s_{*},\kappa)$, 
\begin{equation}
\|\hat{\gamma}-\gamma_{*}\|_{1}\leq8\lambda_{\gamma}s_{*}\kappa^{-2}\ {\rm and}\ n^{-1/2}\|\tilde{W}(\hat{\gamma}-\gamma_{*})\|_{2}\leq4\lambda_{\gamma}\sqrt{s_{*}}\kappa^{-1}.
\label{eq: power unknown sigmaX eq 2}
\end{equation}

By the definition of $\hat{\pi}$, 
$$\|n^{-1}\tilde{W}^\top (V-\tilde{W}\hat{\pi})\|_{\infty}\leq\eta\hat{\rho}n^{-1/2}\|V\|_{2}\leq\eta_{\pi};$$  thus, by (\ref{eq: power unknown sigmaX eq 2}),
\begin{align}
\frac{|I_{1}|}{\hat{\sigma}_{u}\sigma_{\varepsilon}}
&\leq
\frac{\sqrt{n}\|n^{-1}\tilde{W}^\top (V-\tilde{W}\hat{\pi})\|_{\infty}\|\hat{\gamma}-\gamma_{*}\|_{1}}{\hat{\sigma}_{u}\sigma_{\varepsilon}}
\\
&\leq\frac{8\sqrt{n}\eta_{\pi}\lambda_{\gamma}s_{*}\kappa^{-2}}{\sigma_{u}(1+o_{P}(1))\sigma_{\varepsilon}}\\
&\overset{(i)}{=}\frac{(s_{*}n^{-1/2}\log p)O_{P}(\sigma_{v}\sigma_{z})}{\sigma_{u}\sigma_{\varepsilon}}\overset{(ii)}{=}o_{P}(1),\label{eq: power unknown sigmaX eq 3}
\end{align}
where $(i)$ follows by $n^{-1}\|Z\|_{2}^{2}\sigma_{z}^{-2}=1+o_{P}(1)$
and $n^{-1}\|V\|_{2}^{2}\sigma_{v}^{-2}=1+o_{P}(1)$ (by the Law of Large Numbers  since
both $n^{-1}\|Z\|_{2}^{2}\sigma_{z}^{-2}$ and $n^{-1}\|V\|_{2}^{2}\sigma_{v}^{-2}$
are averages of $n$ independent $\chi^{2}(1)$ random variables)
and $(ii)$ holds by (\ref{eq: power unknown SigmaX eq 0}),$\sigma_{u}/\sigma_{z}\geq c_{3}$,
$\sigma_{\varepsilon}\geq c_{1}$ and $\sigma_{v}=O(1)$ (Lemma \ref{lem: feasibility power}).

By CLT, $I_{2}/\sigma_{u}\rightarrow^{d}\mathcal N(0,\sigma_{\varepsilon}^{2})$.
Since $\hat{\sigma}_{u}/\sigma_{u}=1+o_{P}(1)$, the Slutsky's lemma
implies that 
\begin{equation}
\frac{I_{2}}{\hat{\sigma}_{u}\sigma_{\varepsilon}}\rightarrow^{d}\mathcal N(0,1).\label{eq: power unknown sigmaX eq 4}
\end{equation}

By (\ref{eq: power unknown SigmaX eq 0}), 
$$n^{-1/2}h_{n}u^\top u/(\hat{\sigma}_{u}\sigma_{\varepsilon})=d(\sigma_{u}/\hat{\sigma}_{u})(n^{-1}u^\top u\sigma_{u}^{-2}).$$
Notice that $n^{-1}u^\top u\sigma_{u}^{-2}$ is the average of $n$ independent
$\chi^{2}(1)$ random variables. It follows, by the Law of Large Numbers  and $\sigma_{u}/\hat{\sigma}_{u}=1+o_{P}(1)$,
that 
\begin{equation}
\frac{I_{3}}{\hat{\sigma}_{u}\sigma_{\varepsilon}}=d+o_{P}(1).\label{eq: power unknown sigmaX eq 5}
\end{equation}

On the event $\mathcal{M}\bigcap\mathcal{D}_{n}(s_{*},\kappa)$, we
have that 
$$\|\hat{\pi}\|_{1}\leq\|\pi_{*}+\gamma_{*}h_{n}\|_{1},$$
$$\|n^{-1}\tilde{W}^\top (V-\tilde{W}\hat{\pi})\|_{\infty}\leq\eta\hat{\rho}n^{-1/2}\|V\|_{2}\leq\eta_{\pi}$$
and 
$$\|n^{-1}\tilde{W}^\top (V-\tilde{W}(\pi_{*}+\gamma_{*}h_{n}))\|_{\infty}\leq\eta\rho_{*}n^{-1/2}\|V\|_{2}\leq\eta_{\pi}.$$
We apply Lemma \ref{lem: DS1} with $(Y,X,\xi_{*})$, replaced by $(V,\tilde{W},\pi_{*}+\gamma_{*}h_{n})$,
and obtain
\begin{equation}
\left\|\hat{\pi}-(\pi_{*}+\gamma_{*}h_{n})\right\|_{1}\leq8\eta_{\pi}s_{*}\kappa^{-2}\ {\rm and} \ n^{-1/2} \left \|\tilde{W}[\hat{\pi}-(\pi_{*}+\gamma_{*}h_{n})] \right\|_{2}\leq4\eta_{\pi}\sqrt{s_{*}}\kappa^{-1}.\label{eq: power unknown sigmaX eq 6}
\end{equation}

Observe that, on the event $\mathcal{M}\bigcap\mathcal{D}_{n}(s_{*},\kappa)$,
$\|n^{-1}\tilde{W}^\top u\|_{\infty}=\|n^{-1}\tilde{W}^\top (Z-\tilde{W}\gamma_{*})\|_{\infty}\leq\lambda_{\gamma}$
and thus
\begin{align}
\frac{|I_{4}|}{\hat{\sigma}_{u}\sigma_{\varepsilon}}
&\leq\frac{\sqrt{n}\|n^{-1}\tilde{W}^\top u\|_{\infty}\|\hat{\pi}-(\pi_{*}+\gamma_{*}h_{n})\|_{1}}{\sigma_{u}\sigma_{\varepsilon}}\cdot\frac{\sigma_{u}}{\hat{\sigma}_{u}}\\
&\overset{(i)}{\leq}\frac{8\sqrt{n}\lambda_{\gamma}\eta_{\pi}s_{*}\kappa^{-2}}{\sigma_{u}\sigma_{\varepsilon}}\cdot(1+o_{P}(1))
\\
&=O_{P}(s_{*}n^{-1}\log p)\frac{n^{-1/2}\|V\|_{2}}{\sigma_{\varepsilon}}\cdot\frac{n^{-1/2}\|Z\|_{2}}{\sigma_{u}}\overset{(ii)}{=}o_{P}(1),\label{eq: power unknown sigmaX eq 7}
\end{align}
where $(i)$ follows by (\ref{eq: power unknown sigmaX eq 6}) and
$\hat{\sigma}_{u}/\sigma_{u}=1+o_{P}(1)$ and $(ii)$ follows by the
same argument as (\ref{eq: power unknown sigmaX eq 3}). By Slutszky's
lemma, together with (\ref{eq: power unknown SigmaX eq 0.5}), (\ref{eq: power unknown sigmaX eq 3}),
(\ref{eq: power unknown sigmaX eq 4}), (\ref{eq: power unknown sigmaX eq 5}),
(\ref{eq: power unknown sigmaX eq 7}), we have 
\begin{equation}
S_{n}\frac{\hat{\sigma}_{\varepsilon}}{\sigma_{\varepsilon}}=\frac{n^{-1/2}(V-\tilde{W}\hat{\pi})^{\top}(Z-\tilde{W}\hat{\gamma})}{\hat{\sigma}_{u}\sigma_{\varepsilon}}\rightarrow^{d}\mathcal N(d,1).\label{eq: power unknown sigmaX eq 8}
\end{equation}

Since $\sigma_{\varepsilon}$ is bounded away from zero, it remains
to show that $\hat{\sigma}_{\varepsilon}=\sigma_{\varepsilon}+o_{P}(1)$.
Note that
\begin{align*}
\left|\hat{\sigma}_{\varepsilon}-n^{-1/2}\|\varepsilon\|_{2}\right|
&=\left|n^{-1/2}\|V-\tilde{W}\hat{\pi}\|_{2}-n^{-1/2}\|\varepsilon\|_{2}\right|
\\
&\leq n^{-1/2}\|V-\tilde{W}\hat{\pi}-\varepsilon\|_{2}\\
&\overset{(i)}{\leq}n^{-1/2}\|\tilde{W}(\pi_{*}+\gamma_{*}h_{n}-\hat{\pi})\|_{2}+n^{-1/2}\|u\sigma_{u}^{-1}\|_{2}|\sigma_{u}h_{n}|
\\
&\overset{(ii)}{=}O_{P}(\eta_{\pi}\sqrt{s_{*}})+O_{P}(1)n^{-1/2}\sigma_{\varepsilon}|d|\overset{(iii)}{=}o_{P}(1),
\end{align*}
where $(i)$ holds by $V=\tilde{W}(\pi_{*}+\gamma_{*}h_{n})+uh_{n}+\varepsilon$,
$(ii)$ holds by (\ref{eq: power unknown sigmaX eq 6}) and $n^{-1/2}\|u\sigma_{u}^{-1}\|_{2}=O_{P}(1)$
(by the Law of Large Numbers ) and $(iii)$ holds by $\eta_{\pi}\sqrt{s_{*}}=n^{-1/2}\|V\|_{2}\sqrt{s_{*}n^{-1}\log p}$,
(\ref{eq: power unknown SigmaX eq 0}), $n^{-1/2}\|V\|_{2}=O_{P}(1)$
(argued in (\ref{eq: power unknown sigmaX eq 3})).

Law of Large Numbers also
implies that $n^{-1/2}\|\varepsilon\|_{2}^{2}=\sigma_{\varepsilon}^{2}+o_{P}(1)$.
This, together with the above display, implies that $\hat{\sigma}_{\varepsilon}=\sigma_{\varepsilon}+o_{P}(1)$.
Hence, by Slutszky's lemma and (\ref{eq: power unknown sigmaX eq 8}),
$S_{n}\rightarrow^{d}\mathcal N(d,1)$. It follows that 
\begin{eqnarray*}
P\left(|S_{n}|>\Phi\left(1-\frac{\alpha}{2}\right)\right) & = & P\left(S_{n}<-\Phi^{-1}\left(1-\frac{\alpha}{2}\right)\right)+P\left(S_{n}>\Phi^{-1}\left(1-\frac{\alpha}{2}\right)\right)\\
 & = & P\left(S_{n}-d<-\Phi^{-1}\left(1-\frac{\alpha}{2}\right)-d\right)
 \\
 & & \qquad +
   P\left(S_{n}-d>\Phi^{-1}\left(1-\frac{\alpha}{2}\right)-d\right)\\
 & \rightarrow & \Phi\left(-\Phi^{-1}\left(1-\frac{\alpha}{2}\right)-d\right)
 \\
 && \qquad +1-\Phi\left(\Phi^{-1}\left(1-\frac{\alpha}{2}\right)-d\right).
\end{eqnarray*}

The desired result follows by noticing that 
$$1-\Phi\left(\Phi^{-1}\left(1-\frac{\alpha}{2}\right)-d\right)=\Phi\left(-\Phi^{-1}\left(1-\frac{\alpha}{2}\right)+d\right).$$ 
\end{proof}

\end{document}